\newif\ifcomment
\newif\ifdraft
\newif\ifsub

\draftfalse

\subtrue 

\ifdraft
\commenttrue

\else
\commentfalse

\fi

\newif\iffull\fulltrue 

\documentclass[acmsmall, noacm]{acmart}
\usepackage{amsmath}
\usepackage{mathtools}
\usepackage{stmaryrd}
\usepackage{cleveref}
\usepackage{todonotes}[disable]
\usepackage{mathpartir}
\usepackage{mdframed}
\usepackage{tikz-cd}
\usepackage{adjustbox}
\usepackage{multirow}
\usepackage[all]{xy} \SelectTips{cm}{}  
\usetikzlibrary{calc,automata,positioning,decorations.pathreplacing}
\usepackage{scalerel,stackengine}
\stackMath
\newcommand\reallywidehat[1]{%
\savestack{\tmpbox}{\stretchto{%
  \scaleto{%
    \scalerel*[\widthof{\ensuremath{#1}}]{\kern.1pt\mathchar"0362\kern.1pt}%
    {\rule{0ex}{\textheight}}
  }{\textheight}%
}{2.4ex}}%
\stackon[-8.5pt]{#1}{\tmpbox}%
}

\AtEndPreamble{%
   \theoremstyle{acmdefinition}
   \newtheorem{remark}[theorem]{Remark}}
\renewcommand{\cref}[1]{\Cref{#1}}
\creflabelformat{enumi}{#2(#1)#3}
\crefname{theorem}{Thm.}{Theorems}
\crefname{definition}{Def.}{Defs}
\crefname{proposition}{Prop.}{Props}
\Crefname{equation}{Eq.}{Eqs}
\crefname{equation}{Eq.}{Eqs}
\crefname{lemma}{Lem.}{Lemmas}
\crefname{remark}{Rem.}{Remarks}
\crefname{example}{Ex.}{Examples}
\crefname{proof}{Proof.}{Proofs}
\crefname{appendix}{Appendix}{Appendixes}
\crefformat{section}{{\S}#2#1#3}
\crefname{figure}{Fig.}{Figs}

\AtBeginDocument{%
  }

\newcommand{\uncurryStep}{SPS transformation}

\newcommand{\subdist}{\mathcal{D}_{\leq 1}}
\newcommand{\nat}{\mathbb{N}}

\newcommand{\nonnegrat}{\mathbb{Q}_{\geq 0}}
\newcommand{\real}{\mathbb{R}}
\newcommand{\nonnegreal}{\mathbb{R}_{\geq 0}}

\newcommand{\boolsets}{\mathbf{2}}
\newcommand{\fstrings}{A^{\ast}}

\newcommand{\Acoin}{A_\coinflip}
\newcommand{\fstringscoin}{\Acoin^{\ast}}
\newcommand{\fstringssubdistcoin}{\subdist(\fstringscoin)}
\newcommand{\Ycoin}{Y_\coinflip}
\newcommand{\sets}{\mathbf{Set}}
\newcommand{\hoaremonad}{\mathcal{L}}

\newcommand{\omegaCPO}{\mathbf{\mathbf{\omega CPO}}}
\newcommand{\dCPO}{\mathbf{\mathbf{DCPO}}}
\newcommand{\scottTop}[1]{\mathcal{O}(#1)}

\newcommand{\comProbPowerdommonad}{\mathcal{M}}

\newcommand{\gendfa}[1]{#1^{A^{\ast}}}
\newcommand{\defeq}{\coloneqq}

\newcommand{\id}{\mathrm{id}}

\newcommand{\type}{\mathbf{Typ}(B)}
\newcommand{\gtype}{\mathbf{GTyp}(B)}

\newcommand{\arfunc}{\mathrm{ar}}
\newcommand{\carfunc}{\mathrm{car}}
\newcommand{\arfuncl}[1]{\mathrm{ar}}
\newcommand{\carfuncl}[1]{\mathrm{car}}
\newcommand{\term}{\mathbf{Term}}

\newcommand{\aitp}[3]{{#1\llbracket #2\rrbracket_{#3}}}
\newcommand{\itp}[3]{{\llbracket #2\rrbracket_{#3}}}
\newcommand{\litp}[4]{{\llbracket #2\rrbracket_{#3}^{#4}}}

\newcommand{\ev}[2]{\mathrm{ev}_{#1, #2}} 
\newcommand{\evsyb}{\mathrm{ev}}
\newcommand{\wpcond}[4]{\mathrm{wp}[#1,#4,#2](#3)}
\newcommand{\rtanmod}[2]{{#2}^{#1}}

\newcommand{\uncurry}[1]{#1^{\_ \times Y}}
\newcommand{\uncurrycoin}[1]{#1^{\_ \times \Ycoin}}
\newcommand{\rec}{\mathrm{Rec}}
\newcommand{\varintro}{\mathrm{Var}}
\newcommand{\exchange}{\mathrm{Ex}}
\newcommand{\conintro}{\mathrm{Con}}
\newcommand{\genintro}{\mathrm{Gen}}
\newcommand{\unitprod}{\mathrm{Unit}_{\times}}
\newcommand{\prodintro}{\mathrm{Prod}}
\newcommand{\projintroone}{\mathrm{Proj}_1}
\newcommand{\projintrotwo}{\mathrm{Proj}_2}
\newcommand{\unitcoprod}{\mathrm{Unit}_{+}}
\newcommand{\coprodone}{\mathrm{Coprod}_{1}}
\newcommand{\coprodtwo}{\mathrm{Coprod}_{2}}
\newcommand{\coprodelm}{\mathrm{Coprod}}
\newcommand{\abst}{\mathrm{Abst}}
\newcommand{\appl}{\mathrm{App}}
\newcommand{\mapLambdac}[2]{#2(#1)}
\newcommand{\mapAssignment}[2]{#1(#2)}

\newcommand{\distterm}[3]{\mathrm{d}}
\newcommand{\primetrans}[1]{#1{'}}
\newcommand{\cal}[3]{\mathrm{al}^{#1}_{#2, #3}}
\newcommand{\leastelement}[2]{\bot^{#1}_{#2}}
\newcommand{\fixop}[3]{\mathrm{fix}^{#1}_{#2}(#3)}
\newcommand{\lstrength}[2]{{t^{#1}_{#2}}} 
\newcommand{\rstrength}[2]{{s^{#1}_{#2}}} 
\newcommand{\op}{\mathrm{op}}
\newcommand{\bt}{\mathbf{t}}

\newcommand{\cljoin}{\mathbf{CL}_{\vee}}

\newcommand{\letrec}[4]{\mathbf{let}\ \mathbf{rec}\ #1\ #2 = #3\ \mathbf{in}\ #4}

\newcommand{\ifexpr}[3]{\mathbf{if}\ #1\ \mathbf{then}\ #2 \ \mathbf{else}\ #3}
\newcommand{\probbranch}[3]{#1 \oplus_{#2} #3}
\newcommand{\exprobbranch}[4]{#1 \oplus_{#2,#3} #4}
\newcommand{\coinflip}{\mathtt{coin}}
\newcommand{\nefunc}{\mathtt{ne}}
\newcommand{\gainreward}{\mathtt{gr}}
\newcommand{\randomwalk}{\mathtt{rw}}

\newcommand{\success}{\mathtt{s}}
\newcommand{\fail}{\mathtt{f}}
\newcommand{\heads}{\mathtt{h}}
\newcommand{\tails}{\mathtt{t}}
\newcommand{\up}{\mathtt{u}}
\newcommand{\down}{\mathtt{d}}
 \newcommand{\writech}{\mathtt{w}}
  \newcommand{\closech}{\mathtt{c}}
  \newcommand{\terminatech}{\mathtt{t}}
\newcommand{\unit}{\mathbf{unit}}

\newcommand{\sstructure}{\mathcal{A}_S}
\newcommand{\tstructure}{\mathcal{A}_T}
\newcommand{\lambdacfix}{\lambda_c^\mathrm{fix}}
\newcommand{\darrow}{\mathrel{\dot\to}}
\newcommand{\typeReal}{\mathbf{real}}
\newcommand{\natu}[1]{u^{#1}}
\newcommand{\pullbackmark}[2]{\save ;p+<.8pc,0pc>:(0,-1)::%
(#1) *{\phantom{Z}} %
;p+(#2)-(0,0) **@{-}%
;p-(#1)+(0,0) *{\phantom{Z}} **@{-} \restore}
\begin{document}

\title{A Denotational Product Construction for Temporal Verification of Effectful Higher-Order Programs}

\author{Kazuki Watanabe}
\email{kazukiwatanabe@nii.ac.jp}
\affiliation{%
  \institution{National Institute of Informatics}
  \country{Japan}
}
\author{Mayuko Kori}
\email{mkori@kurims.kyoto-u.ac.jp}
\affiliation{%
  \institution{Kyoto University}
  \country{Japan}
}
\author{Taro Sekiyama}
\email{tsekiyama@acm.org}
\affiliation{%
  \institution{National Institute of Informatics}
  \country{Japan}
}
\author{Satoshi Kura}
\email{satoshi.kura@aoni.waseda.jp}
\affiliation{%
  \institution{Waseda University}
  \country{Japan}
}
\author{Hiroshi Unno}
\email{hiroshi.unno@acm.org}
\affiliation{%
  \institution{Tohoku University}
  \country{Japan}
}









\begin{abstract}
We propose a categorical framework for linear-time temporal verification of effectful higher-order programs, including probabilistic higher-order programs.
Our framework provides a generic denotational reduction---namely, a denotational product construction---from linear-time safety verification of effectful higher-order programs to computation of weakest pre-conditions of product programs. 
This reduction enables us to apply existing algorithms for such well-studied computations of weakest pre-conditions, some of which are available as off-the-shelf solvers.
We show the correctness of our denotational product construction by proving a preservation theorem under strong monad morphisms and an existence of suitable liftings along a fibration.
We instantiate our framework with both probabilistic and angelic nondeterministic higher-order programs, and implement an automated solver for the probabilistic case based on the existing solver developed by Kura and Unno.
To the best of our knowledge, this is the first automated verifier for linear-time temporal verification of probabilistic higher-order programs with recursion. 
\end{abstract}

\settopmatter{printacmref=false}
\setcopyright{none}
\renewcommand\footnotetextcopyrightpermission[1]{}
\pagestyle{plain}

\maketitle

\section{Introduction}
\emph{Temporal verification} is a classic but central subject in formal verification for decades (e.g.~\cite{Pnueli77,Vardi85,VardiW86,Clarke18,BaierK08}).
Given a system model (like programs or transition systems) and a specification on its \emph{temporal behaviour}, temporal verification asks to check whether the model behaves as described by the specification.
For example, \emph{linear-time} temporal verification aims to reason about \emph{traces}, i.e., sequences of observations that the system model makes~\cite{BaierK08,Clarke18}. 
Its applications encompass time or space cost analysis, as well as reachability analysis to bad states.

A well-known approach to temporal verification is the use of \emph{product constructions}~\cite{Pnueli77, VardiW86, CookGPRV07,Vardi85}.
The key idea is to construct a \emph{product} of a given model and specification, thereby reducing the temporal verification problem to a non-temporal
one on the product.
Since the problem over the product no longer involves any temporal specifications, it can be solved using off-the-shelf non-temporal verification techniques~\cite{Clarke18,BaierK08}.
This fundamental idea has been applied to a wide range of temporal verification problems.
For example, it is well known that temporal safety properties of a transition system can be reduced to the emptiness checking of the corresponding product finite automaton (see e.g.~\cite{BaierK08}). 
In the probabilistic setting, regular safety properties of (labelled) Markov chains can be reduced to reachability probability problems over (unlabelled) product Markov chains composed with deterministic finite automata (e.g.~\cite{BaierK08}).

\citet{WatanabeJRH25} proposed a categorical framework for product constructions, demonstrating that it accommodates a variety of such constructions, including temporal safety properties of nondeterministic transition systems and Markov chains.
However, the framework is tailored specifically to transition systems as system models and thus offers no insight into \emph{symbolic} product constructions for programs.
This lack of symbolic reasoning is critical, as it provides no guidance on how to systematically reduce a temporal verification problem to a non-temporal one in a modular (i.e., inductive) manner.

In this work, we introduce a \emph{denotational product construction} that provides a general recipe for reducing linear-time temporal safety verification of effectful higher-order programs with recursion to non-temporal problems, by generalizing the traditional product constructions in a symbolic manner.
For instance, our framework reduces the problem of verifying the likelihood of accepting terminating traces in probabilistic programs to that of verifying the reachability probabilities of corresponding probabilistic product programs.
To systematically handle a variety of effects, including nondeterminism and probability, we build our framework on Moggi’s computational $\lambda$-calculus~\cite{Moggi89} with generic effects~\cite{PlotkinP03}, where effects are interpreted via strong monads.

\Cref{fig:outline} illustrates our framework.
We assume that a given program $M$ is a $T(\_ \times A^{\ast})$-effectful higher-order term, where $T(\_ \times A^{\ast})$ is the monad composed of a base monad $T$---which may represent nondeterminism or probability---and the writer (or action) monad $\_ \times A^{\ast}$, which yields traces consisting of observations over an alphabet $A$.
The specification $S$ is a trace property that defines the set of accepting traces over $A$. Specifically, we assume that it is given as a \emph{deterministic transition system}---a standard instance being a deterministic finite automaton (DFA) for safety properties---with a state set $Y$ and the shared alphabet $A$.

In the first step of~\cref{fig:outline}, we construct a \emph{synchronised program} by translating both the program $M$ and the specification $S$.
The synchronised program $M_s$ is a $T(\_ \times Y)^Y$-effectful higher-order term\footnote{Here we omit the parenthesis and it corresponds to $\big(T(\_ \times Y)\big)^Y$.}, where $T(\_ \times Y)^Y$ can be understood as an \emph{effectful} state monad, that is, the monad obtained by composing the base monad $T$ with the state monad $(\_ \times Y)^Y$, where $Y$ is the state set $Y$ of the transition system for the specification $S$.
We refer to this translation as \emph{synchronisation}, as it aligns the observations made by the program $M$ with the transitions in the specification $S$, effectively synchronising the program with its specification.
We formalise this synchronisation via a (strong) monad transformation $\alpha\colon T(\_ \times A^{\ast}) \Rightarrow T(\_ \times Y)^Y$, and establish its correctness by extending existing work~\cite{Katsumata13, AguirreKK22, WatanabeJRH25}.

In the second step, we translate the synchronised $T(\_\times Y)^Y$-effectful term $M_s$ into a $T$-effectful higher-order term $M_c$ by extending the \emph{store-passing style (SPS) transformation}~\cite{LaunchburyJ94} to the effectful state monad $T(\_\times Y)^Y$ for an arbitrary strong monad $T$. 
We refer to the resulting term $M_c$ as a \emph{product term}, and to the transformation as the \emph{effectful SPS transformation}.
The key idea of this transformation is to decompose the effect $T(\_\times Y)^Y$ into two separate effects: $T$ and $(\_\times Y)^Y$. The latter is encoded as a pure computation, while the former, $T$, remains as an effect in the resulting product term $M_c$.

We show that, as a main contribution, the original verification problem for the program $M$ and specification $S$ can be reduced to the problem of computing the \emph{weakest pre-condition}~\cite{Dijkstra75} of the product term $M_c$.
For instance, given a probabilistic program $M$ and a specification $S$ that defines the set of accepting traces, verifying the likelihood of accepting terminating traces generated by $M$ reduces to computing (or approximating) the weakest pre-expectation of the probabilistic product term $M_c$, which corresponds to the reachability probability~\cite{McIverM05,Kaminski19}.
For this concrete reduced problem, we can apply an existing automated solver~\cite{KuraUnno2024} as an off-the-shelf tool to compute an over-approximation (i.e., an upper bound) of the weakest pre-expectation.

A major challenge in proving the correctness of this reduction—namely, the denotational product construction—lies in establishing the correctness of the effectful SPS transformation.
To this end, we introduce a fibrational framework, following~\cite{Katsumata13,hermidaThesis}.
Within this framework, we demonstrate the existence of liftings along suitable fibrations, which is precisely equivalent to our key lemmas that establish the correctness of the effectful SPS transformation.
Technically, we extend the interpretation of effectful terms and design an appropriate fibration to accommodate our effectful SPS transformation, thereby enabling a comparison between the interpretations of two syntactically distinct terms: the synchronised term and the product term.

\begin{figure}
	\begin{tikzpicture}
		\node[draw, align=center] at (-3, -1) (program) {Program $M$: \\ $T(\_\times A^{\ast})$-effectful higher-order term (\cref{sec:preliminaries})};
		\node[draw, align=center] at (4, -1) (spec) {Specification $S$: deterministic transition \\ system with state set $Y$ and alphabet $A$ (\cref{sec:prodLambdaCalc})};
    \path (program) -- node[midway] (progspec) {+} (spec);
    \node[draw] at (0.45, -3) (product) {$T(\_\times Y)^Y$-effectful higher-order synchronised term $M_s$ (\cref{sec:prodLambdaCalc})};
    \node[draw] at (0.45, -5) (simpleproduct) {$T$-effectful higher-order product term $M_c$ (\cref{sec:uncurrying,sec:recursion})};
    \draw[->, very thick, align=left] (progspec) -- node[midway, right] {Our synchronisation \\by strong monad morphism (\cref{sec:prodLambdaCalc}).  } (product);
    \draw[->, very thick, align=left] (product) -- node[midway, right] {Our effectful \uncurryStep\\(\cref{sec:uncurrying,sec:recursion}).  } (simpleproduct);
	\end{tikzpicture}
	\caption{Outline of our denotational product construction.}
	\label{fig:outline}
\end{figure}
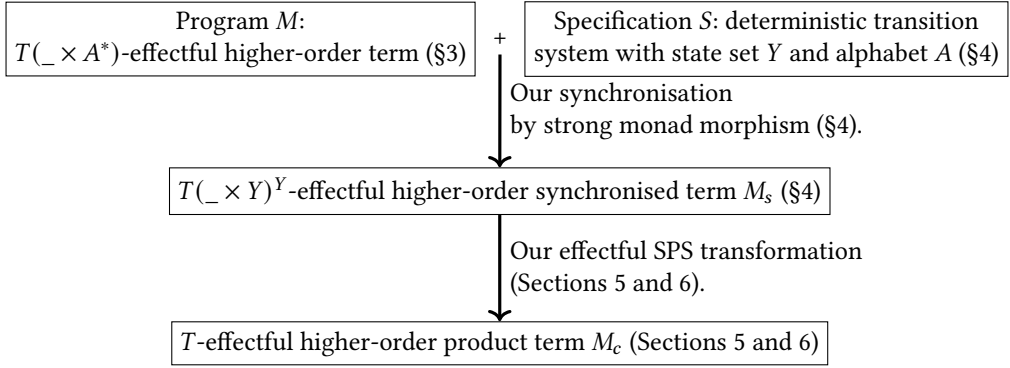

\begin{table}[t]
  \caption{Instances of our denotational product construction. ATT denotes accepting terminating trace. PHOP denotes probabilistic higher-order program, and NHOP denotes angelic non-deterministic higher-order program. 
    }
    \scalebox{1}{
    \begin{tabular}{lcc}\toprule
    Section & Verification Problem & Reduced Problem\\\midrule
    \cref{subsec:ovSafety},\cref{subsec:probAcceptingTrace}   &  Probability of ATTs of PHOPs      & Reachability Probability of PHOPs\\
    \cref{subsec:expectedRewardAcceptingTrace}   &  Expected Reward of ATTs of PHOPs      & Partial Expected Reward of PHOPs \\
    \cref{subsec:emptinessCheckingAcceptingTraces}   &  Emptiness Checking of ATTs of NHOPs     & Reachability of NHOPs \\
    \cref{subsec:maximumRewardsAcceptingTraces}   &  Optimal Reward of ATTs of NHOPs     & Optimal Reward of NHOPs
    \\\bottomrule
    \end{tabular}
    }
    \label{fig:list_of_instances}
\end{table}

We demonstrate the generality and expressiveness of our product construction framework by presenting several instances of our generic product constructions.
These instances are listed in the~\cref{fig:list_of_instances}.
We demonstrate the feasibility of our general approach to temporal verification of probabilistic higher-order programs through a preliminary experiment.
Specifically, we reduce probabilistic temporal verification problems (corresponding to the first two problems in~\cref{fig:list_of_instances}) to threshold problems for weakest pre-conditions (as over-approximations), and then apply the existing automated algorithm~\cite{KuraUnno2024}.
To the best of our knowledge, this is the first automated verifier for temporal safety verification of probabilistic higher-order programs with recursion.

In summary, our contributions are as follows:
\begin{itemize}
  \item  We propose a general, denotational product construction, which translates a pair of a $T(\_\times A^{\ast})$-effectful higher-order program and temporal specification represented by a transition system, into a $T$-effectful higher-order program via synchronisation (\cref{sec:prodLambdaCalc}) and effectful \uncurryStep~(\cref{sec:uncurrying}).
  \item  We show that our approach naturally extends to recursive programs with  a domain-theoretic extension (\cref{sec:recursion}).
  \item  We develop a fibrational framework for proving the correctness of our effectful SPS transformation (\cref{sec:fibrationalapproach}).
  \item  We instantiate temporal verification problems that are listed in \cref{fig:list_of_instances} (\cref{sec:caseStudy}), derived systematically from 
  our generic recipe developed in \cref{sec:genericSynchronisation}.
  \item  We propose an automated verifier for the first two probabilistic instances in~\cref{fig:list_of_instances} using the existing automated solver~\cite{KuraUnno2024} designed for weakest pre-conditions (\cref{sec:experiments}).
\end{itemize}

The rest of the paper is organised as follows.
We give an overview of our framework with a running example in \cref{sec:overview}.
\cref{sec:preliminaries} introduces preliminaries for the formal development of our theoretical contributions (\cref{sec:prodLambdaCalc}--\cref{sec:caseStudy}).
\cref{sec:experiments} shows our automated verifier and a preliminary experiment.
Finally, after discussing related work in \cref{sec:relatedWork}, we conclude in \cref{sec:conclusion}.

\section{Overview}
\label{sec:overview}

In this section, we illustrate how our general framework enables a systematic reduction of a temporal verification problem to a non-temporal one, through a specific instance: temporal safety verification for probabilistic programs.
Additional instances formally derived from our framework are presented in \cref{sec:caseStudy}.

\subsection{Temporal Verification Problem of Probabilistic Programs}
\label{subsec:ovSafety}

Consider the following probabilistic program $\vdash M\colon \unit$ in an ML-like syntax:
\begin{equation}
  \label{eq:exProbProg}
  M \defeq
  \letrec{\coinflip}{x}{\probbranch{\big(\probbranch{(\coinflip\ ())^{\heads}}{1/2}{(\coinflip\ ())^{\tails}}\big)}{1/4}{()}}{\coinflip\ ()}
\end{equation}
Here, we use two generic effects~\cite{PlotkinP03}: \emph{probabilistic branching} $\probbranch{M_1}{p}{M_2}$, which evaluates the expression $M_1$ with probability $p$ and $M_2$ with probability $1-p$ for some real number $p \in [0,1]$; and \emph{event-raising} $(M)^a$, which emits the event $a$ to be observed (as a character of traces), and then evaluates the expression $M$.
The program $M$ invokes a recursive function $\coinflip$ defined with these effects: it terminates immediately with probability $3/4$, or flips a coin and recursively calls itself with probability $1/4$. The coin flip yields heads ($\heads$) or tails ($\tails$), each with probability $1/2$.
Let $\subdist$ denote the monad of subdistributions, and let $\fstringscoin$ be the set of finite traces over the alphabet $\Acoin \defeq \{\heads, \tails\}$.
From a denotational perspective, the program can then be interpreted as a $\subdist(\_ \times \fstringscoin)$-effectful program, where the effect $\subdist(\_ \times \fstringscoin)$ is the composition of the probabilistic monad $\subdist$ and the \emph{writer} (or \emph{action}) monad $\_ \times \fstringscoin$ for trace accumulation.

The verification problem we consider here is: \emph{how likely is it that the program $M$ generates accepting terminating traces?}
For example, suppose that only traces containing the event $\heads$ at least once are considered accepting.
A specification for such traces can be represented by a DFA $D$ with two states $\Ycoin = \{y_1, y_2\}$ and alphabet $\Acoin$, defined as follows:
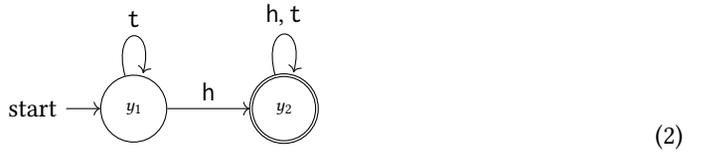
\begin{equation}
  \label{eq:exProbSpec-inOverview}
\begin{tikzpicture}
      \node[state, initial] (y1) at (0, 0) {\tiny $y_1$};
      \node[state, accepting] (y2) at (2, 0) {\tiny $y_2$};
      \draw[->] (y1) to node[pos=0.5, inner sep=3pt, above] {$\heads$} (y2);
      \draw[->] (y1) edge [loop above] node[pos=0.5, inner sep=3pt, above] {$\tails$} (y1);
      \draw[->] (y2) edge [loop above] node[pos=0.5, inner sep=3pt, above] {$\heads, \tails$} (y2);
\end{tikzpicture}
\end{equation}

This specification DFA $D$ defines the \emph{inference query} $q \colon \fstringssubdistcoin \rightarrow [0, 1]^{\Ycoin}$, which maps a subdistribution over finite traces to the probability that the DFA accepts those traces from each state $y \in \Ycoin$ as the initial state.
The verification problem then reduces to computing $q(\nu)(y_1) \in [0, 1]$, where $\nu \in \fstringssubdistcoin$ is the subdistribution of terminating traces generated by the program $M$.

The subdistribution $\nu$ can be represented as a \emph{weakest pre-condition} of the program $M$ with respect to a post-condition $Q$, where $Q$ is given by the Dirac distribution $d_{\epsilon}$ concentrated on the empty trace $\epsilon$.
This weakest pre-condition for the program $\vdash M\colon \unit$ w.r.t. the post condition $Q$ is precisely the denotation of the program, that is, the subdistribution over the terminating traces of the program $M$.
To formalize this representation, we follow the existing approach to \emph{denotational weakest pre-conditions}~\cite{AguirreKK22,Kura2023}, which expresses weakest pre-conditions via a \emph{semantics} $\mathcal{A}$, which defines an interpretation of the generic effects (in this case, probabilistic branching and event-raising), together with an \emph{algebra} $\tau^{\subdist(\_\times \fstringscoin)} \colon \subdist\big(\Omega \times \fstringscoin\big) \to \Omega$, which is used to resolve the effect $\subdist(\_ \times \fstringscoin)$ of the program $M$ in computing the weakest pre-condition.
Here, $\Omega \defeq \fstringssubdistcoin$ serves as the \emph{semantic domain} of $M$.

With these data,  we formulate the problem as follows: 
\vspace{0.1cm}
\begin{mdframed}
{\bf Problem: Probability of Accepting Terminating Traces of Probabilistic Program.}
Let $M$ be a closed $\subdist(\_\times \fstrings)$-effectful program, and $D$ be a  DFA.
Compute the value
\begin{equation*}
  \label{eq:startingpoint}
\big(q\circ (\wpcond{M}{\tau^{\subdist(\_\times \fstrings)}}{Q}{\mathcal{A}})\big)(\star)(y_1)\in  [0, 1],
\end{equation*}
where $q$ is the inference query for the DFA $D$ and $\wpcond{M}{\tau^{\subdist(\_\times \fstrings)}}{Q}{\mathcal{A}}\colon \{\star\}\rightarrow \Omega$ is the weakest pre-condition of $M$ with respect to the post-condition $Q\colon \{\star\}\rightarrow \Omega$ that is the function to the Dirac distribution $d_{\epsilon}$. Note that $y_1$ is the initial state of the DFA. 
\end{mdframed}
\vspace{0.1cm}
Here, $\star$ is the unit value that works as the valuation for the closed program $M$. 
In our denotational semantics, $\{\star\}$ is the interpretation of the empty context and $\unit$ type; in fact, we have $\vdash M\colon \unit$.

By the definition of the temporal problem, it suffices to compute $\wpcond{M}{\tau^{\subdist(\_\times \fstrings)}}{Q}{\mathcal{A}}(\star)(y_1)\allowbreak\in \fstringssubdistcoin$ to solve the problem. 
However, directly computing the weakest pre-condition can be computationally expensive, as it involves reasoning about essentially all (terminating) traces generated by the program.
Yet, not all traces are necessary, since we are concerned only with \emph{accepting} traces.
Moreover, we are interested solely in the probability of accepting traces, not in the full distribution over them.
This naturally raises the question: is it necessary to compute the entire distribution $\wpcond{M}{\tau^{\subdist(\_\times \fstrings)}}{Q}{\mathcal{A}}(\star)(y_1) \in \Omega$ at all?
In the remainder of this overview, we show that—thanks to our novel denotational product construction—the answer is, in fact, no.

Specifically, our product construction framework comprises two steps: \emph{synchronisation} and \emph{store-passing style (SPS) transformation}.
In what follows, we illustrate the key ideas behind these steps using the probabilistic program $M$ given in \cref{eq:exProbProg}.

\subsection{Synchronisation}

First, given the $\subdist(\_ \times \fstringscoin)$-effectful program $M$ and the specification DFA $D$ with state set $\Ycoin$,
we construct a new $\subdist(\_ \times \Ycoin)^{\Ycoin}$-effectful program, referred to as the \emph{synchronised} program $M_s$.
The effect $\subdist(\_ \times \Ycoin)^{\Ycoin}$ can be understood as a $\subdist$-\emph{effectful} state monad over the state space $\Ycoin$, which determines the next state $y' \in \Ycoin$ stochastically from the current state $y \in \Ycoin$.
This synchronisation is achieved via a strong monad morphism $\alpha\colon \subdist(\_ \times \fstringscoin) \Rightarrow \subdist(\_ \times \Ycoin)^{\Ycoin}$.
Notably, this transformation does not alter the syntax of $M$ but may change the interpretation of the generic effects it uses, in order to update the current state of the DFA.
Recall that $M$ involves two generic effects: probabilistic branching $\probbranch{\_}{p}{\_}$ and event-raising $(\_)^a$.
In the synchronised program, the interpretation of probabilistic branching $\probbranch{\_}{p}{\_}$ remains unchanged and preserves the current DFA state.
In contrast, the interpretation of event-raising $(\_)^a$ updates the current state according to the emitted event $a$, transitioning to the successor state $y'$.
For example, if the current state is $y_1$, the effect $(\_)^{\heads}$ triggers a transition to the accepting state $y_2$.

This synchronisation allows us to reformulate the verification problem: instead of reasoning about the original program, we now consider the weakest pre-condition of the synchronised program $M_s$ with respect to the post-condition $q \circ Q$, under the translated semantics $\alpha(\mathcal{A})$ for the synchronised program and a new algebra $\tau^{\subdist(\_\times \Ycoin)^{\Ycoin}}\colon \subdist([0, 1]^{\Ycoin}\times \Ycoin)^{\Ycoin}\rightarrow [0, 1]^{\Ycoin}$,
where $[0,1]^{\Ycoin}$ serves as the semantic domain of the \emph{synchronised program}. The resulting weakest pre-condition $\wpcond{M}{\tau^{\subdist(\_\times \Ycoin)^{\Ycoin}}}{q\circ Q}{\mapLambdac{\mathcal{A}}{\alpha}}$ computes, roughly speaking, the probability of generating accepting terminating traces starting from \emph{each} state $y \in \Ycoin$, regarded as an initial state.
This explains why the semantic domain is $[0,1]^{\Ycoin}$ rather than just $[0,1]$.

In general, our synchronisation ensures the following equation holds:
\begin{equation}
  \label{eq:overviewCorrectSync}
q\circ (\wpcond{M}{\tau^{\subdist(\_\times \fstringscoin)}}{Q}{\mathcal{A}}) = \wpcond{M_s}{\tau^{\subdist(\_\times \Ycoin)^{\Ycoin}}}{q\circ Q}{\mapLambdac{\mathcal{A}}{\alpha}}.
\end{equation}
From this equation, it suffices to compute the right-hand side $\wpcond{M_s}{\tau^{\subdist(\_\times \Ycoin)^{\Ycoin}}}{q\circ Q}{\mapLambdac{\mathcal{A}}{\alpha}}\colon \{\star\}\rightarrow [0, 1]^{\Ycoin}$, where traces are no longer tracked explicitly.
To ensure the correctness of this equation, we provide a sufficient condition concerning the \emph{one-step compositionality} of the transformation $\alpha$ with respect to the two algebras $\tau^{\subdist(_ \times \Ycoin)^{\Ycoin}}$ and $\tau^{\subdist(_ \times \fstringscoin)}$, mediated by the inference query $q$.
This is the only condition that needs to be checked in order to apply our framework; the subsequent step is guaranteed to be correct without any additional conditions.

\subsection{SPS Transformation}

Secondly, from the synchronised program $M_s$, we obtain the \emph{product} program $M_c \defeq \uncurrycoin{M_s}$ via our \emph{effectful store-passing style (SPS) transformation}.
Specifically, given the $\subdist(\_ \times \Ycoin)^{\Ycoin}$-effectful program $M_s$, we construct a $\subdist$-effectful program—that is, a probabilistic program that produces no trace.
Unlike the synchronisation step, this transformation \emph{does} modify the syntax of $M_s$: it separates the state monad $(\_ \times \Ycoin)^{\Ycoin}$ from the combined monad $\subdist(\_ \times \Ycoin)^{\Ycoin}$, reifies the hidden state into explicit program values, and encodes state transitions as program-level operations.
For example, the resulting product program $\uncurrycoin{M_s}$ is given as follows:\footnote{More precisely, the product program shown here is obtained by extending our SPS transformation to \emph{selective} transformation~\cite{Nielsen01}, which only transforms effectful terms and does not modify pure terms, and eliminating \emph{administrative redexes}~\cite{Plotkin75}, which are redexes introduced by the transformation and do not appear in the source program. These optimisations do not affect our theoretical results.}
\begin{align*}
  \label{eq:exUncurryProg}
  \letrec{\coinflip}{x\ y}{\probbranch{\big(\probbranch{\coinflip\ ((), y_2)}{1/2}{\coinflip\ ((), y)}\big)}{1/4}{((), y)}}{\coinflip\ ((), y)}.
\end{align*}
Note that in the product program, the current state $y$ is passed as an argument and returned as a part of the return value, and the event-raising effect is modified to update the current state according to the transitions of the DFA $D$.
For example, if the two probabilistic branching effects both choose the left branches, then the current state is changed to $y_2$ because the resulting branch invokes the event-raising effect $(\_)^{\heads}$ and the transition by $\heads$ always changes the current state to $y_2$ in the DFA $D$.
In contrast, the effect $(\_)^{\tails}$, as well as effect-free constructs such as values, does not change the current state.

For $\uncurrycoin{M_s}$, we can automatically derive the corresponding semantics $\uncurrycoin{\alpha(\mathcal{A})}$ and algebra $\tau^{\subdist} \colon \subdist([0,1]) \to [0,1]$, which computes the expectation via $\tau^{\subdist}(\nu) = \int_p p d\nu$, based on those for $M_s$.
Here, the semantic domain for the product program is the interval $[0, 1]$. 
As a main technical contribution, we show that the following equation holds:
\begin{equation}
  \label{eq:uncurryEx}
  \big(\wpcond{M_s}{\tau^{\subdist(\_\times \Ycoin)^{\Ycoin}}}{q \circ Q}{\alpha(\mathcal{A})}\big)^{\dagger} =  \wpcond{\uncurrycoin{M_s}}{\tau^{\subdist}}{(q \circ Q)^{\dagger}}{\uncurrycoin{\alpha(\mathcal{A})}},
\end{equation}
where $(\_)^\dagger$ is an operation to produce \emph{transposes}, and
\[\begin{array}{r@{\ }c@{\ }l}
 \big(\wpcond{M_s}{\tau^{\subdist(\_\times \Ycoin)^{\Ycoin}}}{q\circ Q}{\alpha(\mathcal{A})}\big)^{\dagger} &\colon& \{\star\}\times \Ycoin\rightarrow [0, 1] \text{\ and} \\
 (q\circ Q)^{\dagger} &\colon& \{\star\}\times \Ycoin\rightarrow [0, 1].
  \end{array}
\]
Recall that $\{\star\}$ is the interpretation of the empty context and $\unit$, and $\vdash M\colon \unit$.
\cref{eq:uncurryEx} ensures the desired \emph{correctness} of our effectful SPS transformation, that is,  
the transpose of the weakest pre-condition of the synchronised program $M_s$ (the LHS of \cref{eq:uncurryEx}) coincides with the weakest pre-condition of the SPS-transformed program $\uncurrycoin{M_s}$ (the RHS of \cref{eq:uncurryEx}). 
With~\cref{eq:overviewCorrectSync}, we can now see that the temporal verification problem for the program $M$ with the DFA $D$ is reduced to the computation of the weakest pre-condition of $\uncurrycoin{M_s}$:
\begin{align}
  q\circ (\wpcond{M}{\tau^{\subdist(\_\times \fstringscoin)}}{Q}{\mathcal{A}}) = \big(\wpcond{\uncurrycoin{M_s}}{\tau^{\subdist}}{(q\circ Q)^{\dagger}}{\uncurrycoin{\alpha(\mathcal{A})}}\big)^{\dagger}.
\end{align}
Notably, our effectful SPS transformation is used to reduce the problem to a rather standard problem. 
For instance, the weakest pre-condition $\wpcond{\uncurrycoin{M_s}}{\tau^{\subdist}}{(q \circ Q)^{\dagger}}{\uncurrycoin{\alpha(\mathcal{A})}}\colon \{\star\}\times \Ycoin\rightarrow [0, 1]$ is the standard weakest pre-expectation for reachability probabilities, which has been actively studied (e.g.~\cite{Kaminski19,McIverM05,KuraUnno2024}). 
This enables us to use the existing optimised solver~\cite{KuraUnno2024} as an off-the-shelf tool to solve the temporal verification problem, without the need for new algorithms for less-explored weakest pre-conditions.

We conclude this overview by presenting the list of instantiated examples in~\cref{fig:list_of_coalgebraic_inference}.
As we have seen, the source program $M$, used as our running example, is reduced to a target program whose property of interest is reachability probability.
\begin{table}
      \caption{Examples of source programs and target programs in our framework.}
    \scalebox{0.7}{
    \begin{tabular}{lccc}\toprule
    & \multicolumn{1}{c}{Source Program $M$} &\multicolumn{2}{c}{Target Program $\uncurry{M_s}$}  
    \\\cmidrule(lr){2-2}\cmidrule(lr){3-4}
    Section& Property  & Property & Semantic domain $\Omega$ \\\midrule
    \cref{subsec:ovSafety},\cref{subsec:probAcceptingTrace} & (sub)distribution of terminating traces  & reachability probability (e.g.~\cite{McIverM05,Kaminski19}) & $[0, 1]$\\
    \cref{subsec:expectedRewardAcceptingTrace} & (sub)distribution of terminating traces with rewards  & partial expected reward (e.g.~\cite{Baier0KW17,WatanabeJRH25}) & $[0, 1]\times [0, \infty]$\\
    \cref{subsec:emptinessCheckingAcceptingTraces}  & sets of terminating traces  & (angelic) reachability (e.g.~\cite{UnnoST18,MaillardAAMHRT19,KobayashiTW2018,Floyd67}) & $\boolsets$ (Boolean domain) \\
    \cref{subsec:maximumRewardsAcceptingTraces}  & (downward-closed) sets of terminating traces with rewards & optimal reward & $[0, \infty]$\\
     \bottomrule
    \end{tabular}
    }
    \label{fig:list_of_coalgebraic_inference}
\end{table}

\paragraph{Notation.}
We prepare some notations for our framework; see e.g.~\cite{Moggi89, BentonHM00} for further details. 
Throughout the paper, we assume that categories are \emph{bicartesian closed categories (bi-CCC)} and \emph{stable}~\cite{FioreS99} (see~\cref{def:stableCoproduct} for the definition).
A \emph{strong monad} $(T, \eta^T, \mu^T, \lstrength{T}{})$ is a monad $(T, \eta^T, \mu^T)$ equipped with a natural transformation $\lstrength{T}{}\colon (\times) \circ (\id \times T)\rightarrow T \circ (\times)$, called the \emph{left-strength}, which satisfies coherence conditions with respect to the cartesian structure of $\mathbb{C}$ and the monad structure of $T$ (see e.g.~\cite{kock1972strong}).
We write $\rstrength{T}{}$ for the \emph{right-strength} of $T$. 
We write 
$\pi_1, \pi_2$ for the first and second projections of a product, respectively,
$\ev{X}{Y}$ for the counit $\epsilon_X\colon X^Y\times Y\rightarrow X$,
and $(\_)^{\dagger}$ for the transpose 
of the adjunction $(\_\times Y)\dashv (\_)^Y$. 
Given a strong monad $T$ and an object $Y$, we define a natural transformation $\natu{T, Y}_{X}\colon T(X^Y)\rightarrow T(X)^Y$
by $\natu{T, Y}_{X}\defeq \big(T(\ev{X}{Y}) \circ \rstrength{T}{X^Y, Y}\big)^\dagger$.
Given a strong monad $T$ and objects $X$ and $Z$, we write $\cal{T}{X}{Z}$ for the canonical Eilenberg-Moore algebra $\cal{T}{X}{Z}\colon T\big(T(Z)^X\big)\rightarrow T(Z)^X$ defined by $\cal{T}{X}{Z}\defeq 
\big(\mu^T_{Z} \circ T(\ev{T(Z)}{X}) \circ \rstrength{T}{T(Z)^X,X}\big)^\dagger$.
We often omit superscript and subscript if they are clear from the context.

\section{Source Program: $\lambda_c$-Calculus with Generic Effects}
\label{sec:preliminaries}
Our source language is based on the $\lambda_c$-calculus~\cite{Moggi89} extended with \emph{generic effects}, which correspond bijectively to algebraic operations~\cite[Theorem 3]{PlotkinP03}.
We adopt this calculus as the foundation of our framework.
In~\cref{sec:recursion}, we further extend it with recursion, following the approaches developed in~\cite{Katsumata13,HasegawaK02,SimpsonP00}.
\subsection{Syntax}
Let $B$ be the set of \emph{base types}.
We define the set $\type$ of \emph{types}, and the set $\gtype$ of \emph{ground types}, respectively:
\begin{align*}
  \type \ni&  \quad \bt ::=  b \mid \mathbf{1} \mid \bt_1 \times \bt_2 \mid \mathbf{0} \mid \bt_1 + \bt_2 \mid \bt_1 \rightarrow \bt_2,\\
  \gtype \ni&  \quad \bt ::=  b \mid \mathbf{1} \mid \bt_1 \times \bt_2 \mid \mathbf{0} \mid \bt_1 + \bt_2 .
\end{align*}
Let $K$ be a set of \emph{effect-free constants}, and $E$ be a set of \emph{generic effects}.
Let $\arfunc, \carfunc\colon (K + E)\rightarrow \gtype$ be two functions 
assigning \emph{arities} and \emph{coarities}, respectively.
\begin{definition}[$\lambda_c$-signature]
  A \emph{$\lambda_c$-signature} $\Sigma$ is the tuple $(B, K, E, \arfunc, \carfunc)$.
  For an effect-free constant $c\in K$ and a generic effect $e\in E$, we often write $c\colon \arfunc(c)\rightarrow \carfunc(c)$ and $e\colon \arfunc(e)\rightarrow \carfunc(e)$.
\end{definition}
We then define the set $\term$ of \emph{terms} as follows:
\begin{align*}
  \term \ni \quad M, N ::= \, \, &x\mid c\ M\mid e\ M \mid () \mid (M, N) \mid \pi_i\ M\mid \delta(M) \mid \iota_i\ M\mid \\
                           &\delta(M, x_1\colon \bt_1.\ N_1, x_2\colon \bt_2.\ N_2) \mid \lambda x\colon \bt.\ M\mid M\ N,
\end{align*}
where $i\in \{1, 2\}$, $c\in K$, and $e\in E$.
As usual, a \emph{context} $\Gamma$ is a (linearly ordered) mapping from variables to types $\Gamma\defeq x_1\colon \bt_1,\dots,  x_n\colon \bt_n$.
We write $x\in \Gamma$ if the variable $x$ appears in $\Gamma$.
A \emph{well-typed term} $\Gamma\vdash M\colon \bt$ is then defined by the standard typing rule: See~\cref{sec:appSourceProgram} for the full typing rules.

\paragraph{Syntax Sugar.}
\begin{table}
	\caption{Syntactic sugars for types and terms.}\label{tab:syntactic-sugar}
	\begin{tabular}{l|l}
		Syntactic sugar & Definition \\
		\hline
		$\mathbf{bool}$ \hspace{2em} (boolean type) & $1 + 1$ \\
		$\mathbf{let}\ x = M\ \mathbf{in}\ N$ & $(\lambda x. N)\ M$ \\
		$M; N$ & $\mathbf{let}\ x = M\ \mathbf{in}\ N$\hspace{2em} ($x$ is not in $N$)\\
		$\letrec{f}{x}{M}{N}$ & $\mathbf{let}\ f = \mu fx. M\ \mathbf{in}\ N$ \\
		$\mathbf{if}\ M\ \mathbf{then}\ N_1\ \mathbf{else}\ N_2$ \hspace{2em} (for $M : \mathbf{bool}$) & $\delta(M, z_1. N_1, z_2. N_2)$ \hspace{2em} ($z_1, z_2$ are fresh)
	\end{tabular}
\end{table}
We use an ML-like syntax sugar, which is also used in~\cite{Kura2023,KuraUnno2024} for examples and implementations.
We list them in~\cref{tab:syntactic-sugar}. 
Note that we will introduce recursive terms $\mu fx. M$ in~\cref{sec:recursion}. 
We also often adopt an abuse of notation when dealing with generic effects as long as they are fairly standard. For instance, $(\coinflip\ ())^{\heads}$ in~\cref{eq:exProbProg} means $(e^{\heads}\ ()); (\coinflip\ () )$  with the corresponding generic effect $e^{\heads}\colon \unit\rightarrow \unit$.

\subsection{Semantics and Weakest Precondition}
We recall the \emph{semantics} of $\lambda_c$-calculus with generic effect~\cite{Katsumata13}.
  Let $\mathbb{C}$ be a bi-CCC, and
  $T$ be a strong monad on $\mathbb{C}$.
  A \emph{generic effect} (on $\mathbb{C}$ with $T$) is a morphism of the form $C \to TD$~\cite{PlotkinP03}.

\begin{definition}[$\lambda_c(\Sigma)$-structure and $\aitp{\mathcal{A}}{\_}{T}$] \label{def:lambdac}
  Given a \emph{$\lambda_c$-signature} $\Sigma$,
a \emph{$\lambda_c(\Sigma)$-structure} is a tuple $\mathcal{A} \defeq (\mathbb{C}, T, A, a)$,
where (i) $\mathbb{C}$ is a bi-CCC and stable;
(ii) $T$ is a strong monad on $\mathbb{C}$;
(iii) $A\colon B\rightarrow \mathbb{C}$ is a functor (the set $B$ of bases is regarded as a discrete category);
and (iv) $a$ is an assignment of interpretations for each effect-free constant and generic effects, which is defined as follows.
We extend $A\colon  B\rightarrow \mathbb{C}$  into
$\itp{\mathcal{A}}{\_}{T}\colon \type\rightarrow \mathbb{C}$ by the bi-CCC structures.
Then, for each $c\in K$, the assignment $a(c)$ is a morphism $a(c)\colon \itp{\mathcal{A}}{\arfunc(c)}{T}\rightarrow\itp{\mathcal{A}}{\carfunc(c)}{T} $,
and for each $e\in E$, the assignment $a(e)$ is a generic effect $a(e)\colon\itp{\mathcal{A}}{\arfunc(e)}{T}\rightarrow T(\itp{\mathcal{A}}{\carfunc(e)}{T})$.
  For each context $\Gamma\defeq x_1\colon \bt_1, \dots, x_n\colon \bt_n$,
  we define $\itp{\mathcal{A}}{\Gamma}{T}$ to be the product $\prod_{i=1}^n \itp{\mathcal{A}}{\bt_i}{T}$.

  Then for a $\lambda_c(\Sigma)$-structure $\mathcal{A}$,
  the \emph{interpretation} $\aitp{\mathcal{A}}{M}{T}\colon \itp{\mathcal{A}}{\Gamma}{T} \to T(\itp{\mathcal{A}}{\bt}{T})$ of a well-typed term $\Gamma \vdash M\colon \bt$ is inductively defined by 
  the typing rules (see e.g.~\cite{Moggi89,BentonHM00,Katsumata13,Kura2023} for the details);
  specifically, the interpretation of effect-free constants and generic effects are given by
  \begin{align*}
    \aitp{\mathcal{A}}{c\ M}{T} &\defeq T(a(c))\circ \aitp{\mathcal{A}}{M}{T},\\
    \aitp{\mathcal{A}}{e\ M}{T}&\defeq \mu_{T\itp{\mathcal{A}}{\carfunc(e)}{T}}\circ T\big( a(e)\big)\circ \aitp{\mathcal{A}}{M}{T}.
  \end{align*}
When no confusion arises, we write $\itp{\mathcal{A}}{M}{T}$ for $\aitp{\mathcal{A}}{M}{T}$.
\end{definition}

Note that the interpretation of the arrow type $\bt_1 \rightarrow \bt_2$ is given by  $\itp{\mathcal{A}}{\bt_1 \rightarrow\bt_2}{T} \defeq T(\itp{\mathcal{A}}{\bt_2}{T})^{\itp{\mathcal{A}}{\bt_1}{T}}$. 

Finally, we recap the \emph{(semantic) weakest pre-condition} for $\lambda_c$-calculus with generic effect. 
There, an Eilenberg-Moore algebra $\tau\colon T\Omega\rightarrow \Omega$ determines the type of weakest pre-condition (see~\cite{AguirreKK22} for the details). 
We call $\Omega$ \emph{semantic domain}; we can think $\Omega$ as the set of (generalized) truth-values for (generalized) predicates. 
\begin{definition}[weakest pre-condition~\cite{AguirreKK22}]
  Let $\mathcal{A} = (\mathbb{C}, T, A, a)$ be a $\lambda_c(\Sigma)$-structure, and $\tau\colon T\Omega\rightarrow \Omega$ be an Eilenberg-Moore algebra.
  Given a well-typed term $\Gamma\vdash M\colon \bt$ and a morphism $Q\colon \itp{\mathcal{A}}{\bt}{T}\rightarrow \Omega$,
  the \emph{weakest pre-condition} $\wpcond{M}{\tau}{Q}{\mathcal{A}}\colon \itp{\mathcal{A}}{\Gamma}{T}\rightarrow \Omega$ of $Q$ with $\Gamma\vdash M\colon \bt$ and $\tau$ is defined by
  \begin{align*}
    \wpcond{M}{\tau}{Q}{\mathcal{A}}\defeq \tau \circ T(Q)\circ \itp{\mathcal{A}}{M}{T}.
  \end{align*}
\end{definition}

\begin{example}[angelic nondeterminism]
  \label{ex:angnonSec3}
We informally illustrate what a weakest pre-condition looks like for programs with angelic nondeterminism.
Suppose we are interested in whether a given $T$-effectful program $\Gamma \vdash M \colon \bt$ may terminate (i.e., reaching a state in $\itp{\mathcal{A}}{\bt}{T}$), where the monad $T$ models angelic nondeterminism.
The semantic domain $\Omega$ in this setting is the Boolean domain $\boolsets = \{\bot, \top\}$, ordered by $\bot \prec \top$.
In what follows, we show that $\wpcond{M}{\tau}{Q}{\mathcal{A}}(x) = \top$ if and only if $M$ may terminate (in $\itp{\mathcal{A}}{\bt}{T}$) when executed from $x \in \itp{\mathcal{A}}{\Gamma}{T}$, for appropriately chosen $Q$ and $\tau$.

We interpret $\itp{\mathcal{A}}{M}{T} \colon \itp{\mathcal{A}}{\Gamma}{T} \rightarrow T(\itp{\mathcal{A}}{\bt}{T})$ as representing the set of reachable terminating states $\itp{\mathcal{A}}{M}{T}(x) \subseteq \itp{\mathcal{A}}{\bt}{T}$ of the program $M$ from an initial state $x \in \itp{\mathcal{A}}{\Gamma}{T}$.
Since we focus on the termination property into some states in $\itp{\mathcal{A}}{\bt}{T}$,
we define the postcondition $Q \colon \itp{\mathcal{A}}{\bt}{T} \rightarrow \boolsets$ as the constant function that always returns $\top$.
The Eilenberg-Moore algebra $\tau \colon T(\boolsets) \to \boolsets$ resolves nondeterminism in an angelic manner by taking the disjunction: $\tau(S) = \bigvee S$.
The weakest pre-condition $\wpcond{M}{\tau}{Q}{\mathcal{A}} \colon \itp{\mathcal{A}}{\Gamma}{T} \rightarrow \boolsets$ is then given as follows: for any initial state $x$,
$\wpcond{M}{\tau}{Q}{\mathcal{A}}(x) = \top$ iff $\itp{\mathcal{A}}{M}{T}(x)\not = \emptyset$, that is, program $M$ may terminate in $\itp{\mathcal{A}}{\bt}{T}$ when executed from $x$.

\end{example}

\begin{example}[probability]
  \label{ex:overviewDenotation}
  We illustrate the weakest pre-condition for probabilistic programs introduced in~\cref{subsec:ovSafety}, through a simple example.
To simplify the presentation, we work informally as if we are in the category of sets, although we will later introduce domain-theoretic models in~\cref{sec:caseStudy} to properly accommodate recursive functions.
Let us write $U$ for the strong monad $\subdist(\_\times \fstringscoin)$ to simplify notation.
Consider the program $\vdash N \colon \unit$ defined by $N \defeq \probbranch{()^{\heads}}{1/2}{()^{\tails}}$.
Its interpretation $\itp{\mathcal{A}}{N}{U}\colon \{\star\} \to U(\{\star\})$ is given by $\itp{\mathcal{A}}{N}{U}(\star)(\star, \heads) = 1/2$ and  $\itp{\mathcal{A}}{N}{T}(\star)(\star, \tails) = 1/2$.
It means that the program $N$ terminates with probability $1/2$ while raising the event $\heads$, and with probability $1/2$ while raising the event $\tails$.
Note that the interpretation of the type $\unit$ under $U$ is the singleton set $\{\star\}$.
Let $Q\colon \{\star\} \rightarrow \fstringssubdistcoin$ be a post-condition defined by $Q(\star) = d_{\epsilon}$, where $d_{\epsilon}$ denotes the Dirac distribution concentrated at the empty string $\epsilon$.
We compute the weakest pre-condition of the program $N$ with respect to the post-condition $Q$ and an algebra $\tau^U\colon U(\Omega) \rightarrow \Omega$, where $\Omega = \fstringssubdistcoin$.
Define $f \defeq \big(U(Q) \circ \itp{\mathcal{A}}{N}{U}\big)\colon \{\star\} \rightarrow U(\Omega)$.
This function yields a distribution $f(\star)$ such that $f(\star)(d_{\epsilon}, \heads) = 1/2$ and $f(\star)(d_{\epsilon}, \tails) = 1/2$.
Then, the value $\tau^U(f(\star)) \in \fstringssubdistcoin$, which corresponds to the weakest pre-condition $\wpcond{N}{\tau^{U}}{Q}{\mathcal{A}}(\star)$, is given as follows:
 \begin{align*}
  \tau^{U}\big(f(\star)\big)(\heads) = f(\star)(d_{\epsilon}, \heads)\cdot d_{\epsilon}(\epsilon) = 1/2, \quad \tau^{U}\big(f(\star)\big)(\tails) = f(\star)(d_{\epsilon}, \tails)\cdot d_{\epsilon}(\epsilon) = 1/2.
 \end{align*}
 In fact,  $\wpcond{N}{\tau^{U}}{Q}{\mathcal{A}}(\star)$ is precisely the distribution of traces generated by $N$: it raises $\heads$ with the probability $1/2$, and raises $\tails$ with the probability $1/2$.  
\end{example}

\section{Synchronisation of $\lambda_c$-Calculus with Specifications} 
\label{sec:prodLambdaCalc}

We begin with the first step of our generic product construction (see~\cref{fig:outline}), namely, the synchronisation of terms in $\lambda_c$-calculus with specifications. 
\subsection{Translation of Effectful Programs by Strong Monad Morphisms}
We employ a strong monad morphism, which enables us to translate 
a source $T$-effectful program into a synchronised $U$-effectful program.
This translation realizes our synchronisation.
\begin{definition}[strong monad morphism]
  Given two strong monads $(T, \eta^T, \mu^T, t^T)$ and $(U, \eta^U, \mu^U, t^U)$ on the same category,
  a natural transformation $\alpha\colon T\Rightarrow U$ is a \emph{strong monad morphism} if it satisfies the following conditions for each object $X, Y \in \mathbb{C}$:
  \begin{align*}
    \eta^U_X &= \alpha_X \circ \eta^T_X,\quad \alpha_X \circ \mu^T_X = \mu^U_X\circ U(\alpha_X) \circ \alpha_{TX},\quad\text{and}\quad  \alpha_{X\times Y}\circ t^T = t^U \circ (\id_X\times \alpha_Y).
  \end{align*}
\end{definition}
\begin{definition}[map of $\lambda_c(\Sigma)$-structure~\cite{Katsumata13}]
  Given a $\lambda_c(\Sigma)$-structure $\mathcal{A} = (\mathbb{C}, T, A, a)$ and strong monad morphism $\alpha\colon T\Rightarrow U$,
  we define a \emph{$\lambda_c(\Sigma)$-structure} $\mapLambdac{\mathcal{A}}{\alpha}$ to be the tuple $(\mathbb{C}, U, A, \mapAssignment{\alpha}{a})$, where
    $\mapAssignment{\alpha}{a}(c) \defeq a(c)$ and $\mapAssignment{\alpha}{a}(e) \defeq \alpha_{\itp{\mathcal{A}}{\carfunc(e)}{T}} \circ a(e)$,
  for each $c\in K$ and $e\in E$.
\end{definition}

The following lemma shows that a strong monad morphism provides a correct translation of $\lambda_c$-terms with respect to the interpretation under the source and target monads.
\begin{lemma}[\cite{Katsumata13,Kura2023}]
  \label{lem:monadTransGroundTypes}
  Let $\alpha\colon T\Rightarrow U$ be a strong monad morphism.
  For each well-typed $\lambda_c$-term $x_1\colon \bt_1,\dots, x_n\colon \bt_n\vdash M\colon \bt$,
  where $\bt_1, \dots, \bt_n, \bt$ are ground types, we have
  \begin{equation*}
    \pushQED{\qed} 
    \alpha_{\itp{\mathcal{A}}{\bt}{T}} \circ \aitp{\mathcal{A}}{M}{T} = \aitp{\mapLambdac{\mathcal{A}}{\alpha}}{M}{U}. \qedhere
    \popQED
  \end{equation*}
\end{lemma}
It should be noted that the stability of $\mathbb{C}$ is required in the proof of~\cref{lem:monadTransGroundTypes}; see~\cite{Kura2023}. 
We also remark that~\cref{lem:monadTransGroundTypes} only holds for ground types: in~\cite{Katsumata13} this has been proved by the fibrational logical relation. 
We now move from the translation of interpretations to that of weakest preconditions.
This translation requires an additional ingredient: an \emph{inference query} (see~\cref{def:inferenceMap}), which is a map from the semantic domain of the source program to that of the target program.
\begin{definition}[inference query]
\label{def:inferenceMap}
  Let $T, U$ be strong monads, $\alpha\colon T\Rightarrow U$ be a monad morphism, and $\tau^T\colon T(\Omega^T)\rightarrow \Omega^T$ and $\tau^U\colon U(\Omega^U)\rightarrow \Omega^U$ be two Eilenberg-Moore algebras.
  An \emph{inference query} (for $\alpha$, $\tau^T$, and $\tau^U$) is a homomorphism $q\colon \Omega^T\rightarrow \Omega^U$ such that
    $q \circ \tau^T = \tau^U \circ U(q)\circ \alpha_{\Omega^T}$.
\end{definition}
Our definition of inference queries is indeed inspired by~\cite{WatanabeJRH25}, but it differs in two key aspects: 1)
we exploit strong monad structures to establish correctness in a \emph{compositional} manner; and
2) we encode specifications (e.g., DFAs) as strong monad morphisms, whereas their approach treats specifications as models akin to system models.

The following theorem establishes the correctness of the translation of effectful programs in view of weakest preconditions.
\begin{theorem}[correctness]
\label{thm:correctProd}
  Assume the following data:
  \begin{itemize}
    \item Strong monads $T, U$ and a strong monad morphism $\alpha\colon T\Rightarrow U$,
    \item Eilenberg-Moore algebras $\tau^T\colon T(\Omega^T)\rightarrow \Omega^T$ and $\tau^U\colon U(\Omega^U)\rightarrow \Omega^U$ and an inference query $q\colon \Omega^T\rightarrow \Omega^U$.
  \end{itemize}
  Let $\mathcal{A}$ be a $\lambda_c$-signature, $\Gamma\vdash M\colon \bt$ be a well-typed term such that all types in $\Gamma$ and $\bt$ are ground types, and let $Q\colon \itp{\mathcal{A}}{\bt}{T}\rightarrow \Omega^{T}$ be a post condition.  Then the following equality holds:
  \begin{equation}
    \label{eq:correctSynch}
    q\circ \big(\wpcond{M}{\tau^T}{Q}{\mathcal{A}}\big) = \wpcond{M}{\tau^U}{q\circ Q}{\mapLambdac{\mathcal{A}}{\alpha}}.
  \end{equation}
\end{theorem}
\begin{proof}
  We prove the statement as follows:
  \begin{align*}
    & q\circ (\wpcond{M}{\tau^T}{Q}{\mathcal{A}}) = q\circ \tau^T\circ T(Q)\circ \itp{\mathcal{A}}{\Gamma\vdash M\colon \bt}{T}\\
     &= \tau^U \circ U(q) \circ \alpha_{\Omega^T}\circ T(Q) \circ  \itp{\mathcal{A}}{\Gamma\vdash M\colon \bt}{T} = \tau^U \circ U(q\circ Q) \circ \alpha_{\itp{\mathcal{A}}{\bt}{T}} \circ \itp{\mathcal{A}}{\Gamma\vdash M\colon \bt}{T}\\
     &= \tau^U \circ U(q\circ Q) \circ \itp{\mapLambdac{\mathcal{A}}{\alpha}}{\Gamma\vdash M\colon \bt}{U}= \wpcond{M}{\tau^U}{q\circ Q}{\mapLambdac{\mathcal{A}}{\alpha}}.
  \end{align*}
  The second and fourth equality hold by~\cref{def:inferenceMap} and~\cref{lem:monadTransGroundTypes}, respectively.  
\end{proof}

\begin{example}[angelic nondeterminism]
  \label{ex:synchAnon}
Here, we provide an explanation informally; 
we focus only on illustrating the types of data required to apply \cref{thm:correctProd} for synchronisation, along with their intuitions.
The formal definitions will appear in~\cref{subsec:emptinessCheckingAcceptingTraces}. 
We fix a specification as a DFA with a set of states $Y$.
For angelic nondeterminism, we adopt a strong monad morphism $\alpha$ of the form 
\begin{displaymath}
\alpha\colon \hoaremonad(\_\times A^{\ast}) \Rightarrow \hoaremonad(\_\times Y)^Y, 
\end{displaymath}
where $\hoaremonad$ is a strong monad representing angelic nondeterminism (the extended Hoare powerdomain monad that includes the emptyset), $A$ is the alphabet, and $Y$ is the set of states of the DFA.
To define weakest preconditions, we consider Eilenberg-Moore algebras in the following form: 
\begin{itemize}
  \item for the source monad $T \coloneqq \hoaremonad(\_\times A^{\ast})$, $\tau^{T}\colon \hoaremonad\big(\hoaremonad(A^{\ast}) \times A^{\ast}\big)\rightarrow \hoaremonad(A^{\ast})$, and
  \item for the target monad $U \coloneqq \hoaremonad(\_\times Y)^Y$, $\tau^{U}\colon \hoaremonad(\boolsets^Y\times Y)^Y\rightarrow \boolsets^Y$.
\end{itemize} 
Roughly speaking, a truth-value in $\hoaremonad(A^{\ast})$ denotes a set of traces of a program, and a truth-value in $\boolsets^Y$ denotes whether there is an accepting trace from each initial state $y\in Y$; in this sense, our focus is on angelic nondeterminism (or may semantics).
Then the inference query $q\colon \hoaremonad(A^{\ast})\rightarrow \boolsets^Y$ translates a set $L$ of traces to a function $q(L)\colon Y\rightarrow \boolsets$ such that $q(L)(y) = \top$ iff there is an accepting trace in $L$ from $y$. 
Under this setting with $Q\colon \itp{\mathcal{A}}{\bt}{\hoaremonad(\_\times A^{\ast})}\rightarrow \hoaremonad(A^{\ast})$ defined by a constant function to the singleton $\{\epsilon\}$,  
\cref{eq:correctSynch} states that the temporal problem (LHS of \cref{eq:correctSynch}) that asks  whether there is an accepting trace
is equivalent to the weakest precondition (RHS of \cref{eq:correctSynch}) of the post-condition $q \circ Q\colon \itp{\mathcal{A}}{\bt}{\hoaremonad(\_\times Y)^Y} (= \itp{\mathcal{A}}{\bt}{\hoaremonad(\_\times A^{\ast})}) \rightarrow \boolsets^Y$.
\end{example}

\begin{example}[probability]
  \label{ex:probSynchSec4}
  We continue~\cref{ex:overviewDenotation}. Let $V$ denote the strong monad $\subdist(\_\times \Ycoin)^{\Ycoin}$.
  We can construct a strong monad morphism $\alpha\colon \subdist(\_\times \fstringscoin)\Rightarrow V$ from the DFA $D$ (provided in~\cref{subsec:ovSafety}), and  the interpretation $\aitp{\mapLambdac{\mathcal{A}}{\alpha}}{N}{V}\colon\allowbreak \{\star\}\rightarrow V(\{\star\})$  of  $N$ is a function such that 
  \begin{align*}
    \aitp{\mapLambdac{\mathcal{A}}{\alpha}}{N}{V}(\star)(y_1)(\star, y_2) = 1/2, \quad  \aitp{\mapLambdac{\mathcal{A}}{\alpha}}{N}{V}(\star)(y_1)(\star, y_1) = 1/2,
  \end{align*}
  because it raises $\heads$ with the probability $1/2$, which changes the state $y_1$ into the state $y_2$ by the synchronisation (and similarly for $\heads$). 
  The post-condition of interest is $q\circ Q\colon \{\star\}\rightarrow [0, 1]^{\Ycoin}$ such that $(q\circ Q)(\star)(y_1) = 0$ and $(q\circ Q)(\star)(y_2) = 1$, which assigns $1$ to accepting states and $0$ to non-accepting states. 
Recall that the weakest pre-condition of $N$ with the post-condition $q\circ Q$ and an Eilenberg-Moore algebra $\tau^V\colon \subdist([0, 1]^{\Ycoin}\times \Ycoin)^{\Ycoin}\rightarrow [0, 1]^{\Ycoin}$, which computes the expectation associated with each $y\in \Ycoin$, is given by 
  \begin{align*}
    \wpcond{M}{\tau^{V}}{q\circ Q}{\mapLambdac{\mathcal{A}}{\alpha}} \defeq \tau^{V} \circ V(q\circ Q)\circ \aitp{\mapLambdac{\mathcal{A}}{\alpha}}{N}{V}\colon \{\star\}\rightarrow [0, 1]^{\Ycoin},
  \end{align*}
  and we can see that $f\defeq V(q\circ Q)\circ \aitp{\mapLambdac{\mathcal{A}}{\alpha}}{N}{V}$ is a function such that $f(\star)(y_1)\big((q\circ Q)(\star), y_1\big) = 1/2$ and $f(\star)(y_1)\big((q\circ Q)(\star), y_2\big) = 1/2$. 
  Finally, the value $\tau^V(f)(y_1)\in [0, 1]$, which is the weakest pre-condition $\big(\wpcond{M}{\tau^{V}}{q\circ Q}{\mapLambdac{\mathcal{A}}{\alpha}}\big)(\star)(y_1)$, is given as follows: 
  \begin{align*}
    & \, f(\star)(y_1)\big((q\circ Q)(\star), y_1\big)\cdot (q\circ Q)(\star)(y_1) + f(\star)(y_1)\big((q\circ Q)(\star), y_2\big)\cdot (q\circ Q)(\star)(y_2)\\
   =&\, 1/2\cdot 0 + 1/2 \cdot 1 = 1/2.
  \end{align*}
  Indeed, this is the probability of accepting traces generated by $N$ from $y_1$, since $\heads$ is the unique accepting trace, which is generated by $N$ with the probability $1/2$.
\end{example}

\subsection{Generic Framework for Synchronisations} 
\label{sec:genericSynchronisation}
To ensure the correctness of our synchronisation,~\cref{thm:correctProd} requires an inference query $q$ that, by definition, satisfies the condition  $q \circ \tau^T = \tau^U \circ U(q)\circ \alpha_{\Omega^T}$ 
where $\tau^T$ and $\tau^U$ are Eilenberg-Moore algebras and $\alpha$ is a monad morphism.
In this section, we present a generic construction---referred to as the \emph{product situation}---that produces data satisfying this condition by design.
Indeed, all of the instances discussed in~\cref{sec:caseStudy} are derived from this construction.

\begin{definition}[product situation]
  We call a tuple $(T, Z, Y, \tau^{\_ \times Z}, \tau^T, \alpha, q)$ \emph{product situation} if the tuple satisfies the following conditions:
  \begin{itemize}
    \item $T$ is a strong monad, $\_\times Z$ is a writer (or action) monad, and $Y$ is an object. 
    \item $\tau^{\_\times Z}\colon \Omega^{\_\times Z}\times Z\rightarrow \Omega^{\_\times Z}$ and $\tau^T\colon T(\Omega^T)\rightarrow \Omega^T$ are Eilenberg-Moore algebras.
    \item $\alpha$ is a strong monad morphism $\alpha\colon \_\times Z \Rightarrow (\_\times Y)^Y$. 
    \item $q\colon \Omega^{\_\times Z}\rightarrow (\Omega^T)^Y$ is an inference query for the strong monads $(\_)\times Z$ and $(\_\times Y)^Y$ with the strong monad morphism $\alpha$, the Eilenberg-Moore algebra $\tau^{\_\times Z}$, and the Eilenberg-Moore algebra $\tau^{(\_\times Y)^Y}\colon \big((\Omega^T)^Y\times Y\big)^Y\rightarrow (\Omega^T)^Y$ defined by $(\ev{\Omega^T}{Y})^Y$.
  \end{itemize}
  We remark that such a monad morphism $\alpha$ bijectively corresponds to a morphism $Z \times Y \to Y$ satisfying certain coherence conditions; see \cref{ap:st_monad_morphism_bij}.
\end{definition}
A product situation provides the data for synchronisation by way of \emph{deterministic} higher-order programs, which are much simpler to handle than non-deterministic or probabilistic programs considered in our case studies.
The tuple $(Z, Y, \tau^{\_ \times Z}, \tau^{(\_\times Y)^Y}, \alpha, q)$ specifies the data for synchronisation of deterministic higher-order programs, and the Eilenberg-Moore algebra $\tau^T$
determines how to resolve the effect $T$, such as non-determinism or probability.
For example, $\tau^T$ takes the disjunction (or angelic choice) or the expectation, depending on $T$.

We then construct the data required for \cref{thm:correctProd} from a product situation.
This construction yields a translation from a source $T$-effectful program to a synchronised $U$-effectful program.
The correctness of this synchronisation is guaranteed by \cref{thm:correctProd} for weakest pre-conditions.
\begin{proposition}
  \label{prop:liftingCorrectness}
  Let $(T, Z, Y, \tau^{\_ \times Z}, \tau^T, \alpha, q)$ be a product situation.
  The following data satisfies the assumption of~\cref{thm:correctProd}:
  \begin{itemize}
    \item Strong monads $T(\_\times Z)$, $T(\_\times Y)^Y$, and a strong monad morphism $\beta\colon T(\_\times Z)\Rightarrow T(\_\times Y)^Y$ defined by $\beta_X\defeq \natu{T, Y}_{X\times Y}\circ T(\alpha_X)$.
    \item Eilenberg-Moore algebras  $\tau^{T(\_\times Z)}\colon T\big(T(\Omega^{\_\times Z})\times Z\big)\rightarrow T(\Omega^{\_\times Z})$  and $\tau^{T(\_\times Y)^Y}\colon T\big((\Omega^T)^Y\times Y\big)^Y\rightarrow (\Omega^T)^Y$  defined by
  \begin{align*}
  \tau^{T(\_\times Z)}&\defeq \mu^{T}_{\Omega^{\_\times Z}} \circ T^2\big(\tau^{\_\times Z}\big) \circ T\big(\rstrength{}{\Omega^{\_\times Z},Z}\big), \\
  \tau^{T(\_\times Y)^Y} &\defeq (\tau^T)^Y\circ T\big(\ev{\Omega^T}{Y}\big)^Y,
  \end{align*}
and an inference query $q^T\colon T\big(\Omega^{\_\times Z}\big)\rightarrow \big(\Omega^T\big)^Y$ defined by
$q^T\defeq \big(\tau^{T}\big)^Y\circ \natu{T, Y}_{\Omega^T}\circ T(q)$.
  \qed
  \end{itemize}
\end{proposition}
See~\cref{sec:omitDefProofProd} for the proof.
Here, 
the strong monad structure of $T(\_\times Z)$ is the one induced by the strength as a distributive law of $T$ over $\_\times Z$, and that of $T(\_\times Y)^Y$ is written in~\cref{def:effectStateMonad}.

\cref{prop:liftingCorrectness} and \cref{thm:correctProd} immediately imply that for each well-typed term $\Gamma \vdash M\colon \bt$ and post condition $Q\colon \itp{\mathcal{A}}{\bt}{T(\_ \times Z)}\rightarrow T(\Omega^{\_\times Z})$, the synchronisation is correct, that is, 
\begin{equation*}
  q^T\circ (\wpcond{M}{\tau^{T(\_ \times Z)}}{Q}{\mathcal{A}}) = \wpcond{M}{\tau^{T(\_\times Y)^Y}}{q^T\circ Q}{\mapLambdac{\mathcal{A}}{\beta}}.
\end{equation*}

\begin{example}[angelic nondeterminism]
  \label{ex:angndPS}
  We continue~\cref{ex:synchAnon}. The data used there can be derived from a product situation. 
  The product situation consists of 
  a tuple $(A^\ast, Y, \tau^{\_ \times A^{\ast}}, \alpha', q')$, 
  for synchronisation of deterministic programs and an Eilenberg-Moore algebra $\tau^{\hoaremonad}$ that resolves non-determinism by taking the disjunction.
  Here, note that
  $\tau^{\_ \times A^{\ast}}\colon A^{\ast}\times A^{\ast}\rightarrow A^{\ast}$,
  $\alpha'\colon \_\times A^{\ast} \Rightarrow (\_\times Y)^Y$, 
  and $q'\colon A^{\ast}\rightarrow \boolsets^Y$.
  We provide the complete definitions of these data in~\cref{subsec:emptinessCheckingAcceptingTraces}. 
\end{example}

\begin{example}[probability]
In the probabilistic setting, the tuple $(A^\ast, Y, \tau^{\_ \times A^{\ast}}, \alpha', q')$ used for synchronising deterministic programs is the same as that for angelic nondeterminism,
except for the inference query $q'\colon A^{\ast} \rightarrow [0, 1]^Y$.
The Eilenberg-Moore algebra $\tau^{\subdist}\colon \subdist([0, 1]) \rightarrow [0, 1]$ computes the expected value, defined by $\tau^{\subdist}(\nu) \defeq \int_p p  d\nu$.
We formally define them within a domain-theoretic model in~\cref{subsec:probAcceptingTrace}.

\end{example}

 \section{SPS Transformation of $\lambda_c$-Calculus}
 \label{sec:uncurrying}

 We now proceed to the second step in~\cref{fig:outline}: the SPS transformation of effectful higher-order programs. 
 The idea of SPS transformation is that we decompose the effect $T(\_\times Y)^Y$ into two effects $T$ and $(\_\times Y)^Y$, and regard the state monad $(\_\times Y)^Y$ as a pure computation by syntactically encoding these computation of states. 

 In this section, we fix 
a strong monad $T$ and the strong monad $S\defeq T(\_\times Y)^Y$ (see \cref{def:effectStateMonad}).
We also fix
a $\lambda_c(\Sigma)$-structure $\sstructure = (\mathbb{C}, S, A, a)$.
We introduce a  new base type $b^Y$ to represent the current state of the specification as a term of the type $b^Y$.  
To reflect this, we expand $B$ by $B + \{b^Y\}$
and $A$ by $A(b^Y) \defeq Y$.

\begin{definition}[SPS transformation] \label{def:uncurry}
  For each type $\bt$, the \emph{SPS transformation} $\uncurry{\bt}$ of $\bt$ is defined by $\bt' \times b^Y$ where $\primetrans{\bt}$ is
  recursively defined as follows.
  \begin{align*}
  &\primetrans{\bt} \coloneqq \bt \text{ for each }\bt \in \{b, \mathbf{1}, \mathbf{0}\}, \quad 
  \primetrans{(\bt_1 \times \bt_2)} \coloneqq \primetrans{\bt_1}\times \primetrans{\bt_2}, \quad
  \primetrans{(\bt_1 + \bt_2)} \coloneqq \primetrans{\bt_1}+ \primetrans{\bt_2}, \\
  &\primetrans{(\bt_1 \rightarrow \bt_2)} \coloneqq \big((\primetrans{\bt_1} \times b^Y) \rightarrow (\primetrans{\bt_2} \times b^Y)\big).
  \end{align*}
Note that for each ground type $\bt\in \gtype$, we have $\primetrans{\bt} = \bt$, which implies $\uncurry{\bt} = \bt\times b^Y$.

For each context $\Gamma \defeq x_1\colon \bt_1,\dots,x_m\colon \bt_m$, the \emph{SPS transformation} $\uncurry{\Gamma}$ is defined by $\uncurry{\Gamma}\coloneqq x_1\colon \primetrans{\bt_1},\dots,x_m\colon \primetrans{\bt}_m, y\colon b^Y$ where $y$ is a distinguished variable.
  The \emph{SPS transformation} $\uncurry{M}$ of a term $M$ is inductively defined in~\cref{fig:uncurrying}.
\end{definition}

We can see that the SPS transformation is well-defined in the following sense: 
\begin{lemma}
  If $\Gamma\vdash M\colon \bt$ is a well-typed term,
  then $\uncurry{\Gamma}\vdash \uncurry{M}\colon \uncurry{\bt}$ is also a well-typed term.
\end{lemma}
\begin{proof}
  By the straightforward structural induction.
\end{proof}

\begin{figure}[t]
  \small
  \textbf{Term}
  \begin{align*}
    &\uncurry{x} \defeq (x, y),\quad  \uncurry{(c\ M)} \defeq c\ \uncurry{M}, \quad \uncurry{(e\ M)} \defeq e\ \uncurry{M}, \quad \uncurry{()}\defeq ((), y),\\
    &\uncurry{(M, N)}\defeq \Big(\lambda z. \big((\pi_1 z, \pi_1 \pi_2 z), \pi_2 \pi_2 z \big) \Big)\ \Big(\big(\lambda z. (\pi_1 z, (\lambda y. \uncurry{N})(\pi_2 z) )\big) (\uncurry{M})\Big), \\
    &\uncurry{(\pi_1\ M)} \defeq  \big(\lambda z. (\pi_1 \pi_1 z, \pi_2 z)\big)(\uncurry{M}),  \quad \uncurry{(\pi_2\ M)} \defeq  \big(\lambda z. (\pi_2 \pi_1 z, \pi_2 z)\big)(\uncurry{M}),\\
    &\uncurry{\delta(M)} \defeq \big(\lambda z. (\delta(\pi_1 z), \pi_2 z)\big) (\uncurry{M}), \quad \uncurry{(\iota_1 \ M)} \defeq \big(\lambda z. (\iota_1 \pi_1 z, \pi_2 z)\big)(\uncurry{M}),\\
    &\uncurry{(\iota_2 \ M)} \defeq \big(\lambda z. (\iota_2 \pi_1 z, \pi_2 z)\big)(\uncurry{M}),\\
    &\uncurry{\big(\delta(M, x_1\colon \bt_1.\ M_1, x_2\colon \bt_2.\ M_2)\big)}\defeq  \delta(N, z\colon {\uncurry{\bt_1}}.\ \uncurry{M_1}[\pi_1 z / x_1, \pi_2 z / y],\, z\colon \uncurry{\bt_2}.\ \uncurry{M_2}[\pi_1 z / x_2, \pi_2 z / y]), \\
    &\text{where } N \defeq \big(\lambda z.~\delta(\pi_1 z,\, x_1\colon \primetrans{\bt_1}.~\iota_1(x_1, \pi_2 z),\, x_2\colon \primetrans{\bt_2}.~\iota_2(x_2, \pi_2 z))\big)(\uncurry{M}),\\
    &\uncurry{(\lambda x\colon \bt_1.\ M)}\defeq \big(\lambda z\colon \uncurry{\bt_1}.\ \uncurry{M}[\pi_1 z/x, \pi_2 z/y],\, y\big),\\
    &\uncurry{(M\ N)}\defeq \Big(\lambda z. \big(\pi_1 z\big)\big((\lambda y. \uncurry{N})(\pi_2 z) \big)\Big)\big(\uncurry{M}\big).
  \end{align*}
  \caption{The SPS transformation for terms.}
  \label{fig:uncurrying}
\end{figure}

Our goal in this section is to show that the weakest pre-condition of the source program $M$ corresponds to that of its SPS transformed term $\uncurry{M}$, as~\cref{eq:uncurryEx}.

Let us introduce $\lambda_c$-signature and structure to interpret the transformed term as follows.
We define a \emph{$\lambda_c$-signature} $\uncurry{\Sigma} \defeq (B, K, E, \uncurry{\arfunc(\_)}, \uncurry{\carfunc(\_)})$
where $\uncurry{\arfunc(\_)}$ and $\uncurry{\carfunc(\_)}$ are defined by $\uncurry{\arfunc(c)}\defeq \arfunc(c)\times b^Y$ and $\uncurry{\carfunc(c)}\defeq \carfunc(c)\times b^Y$, and $\uncurry{\arfunc(e)}\defeq \arfunc(e)\times b^Y$ and $\uncurry{\carfunc(e)}\defeq \carfunc(e)\times b^Y$ for each $c\in K$ and $e\in E$. 
We also define a \emph{$\lambda_c(\uncurry{\Sigma})$-structure} $\tstructure \coloneqq (\mathbb{C}, T, A, \uncurry{a})$ where
$\uncurry{a}(c) \defeq a(c) \times \id_Y$ and $\uncurry{a}(e)$ is the generic effect $\itp{\mathcal{A}}{\uncurry{\arfunc(e)}}{T} \to T\itp{\mathcal{A}}{\uncurry{\carfunc(e)}}{T}$
given
as the transpose ${a(e)}^\dagger$ of $a(e)\colon \itp{\mathcal{A}}{\arfunc(e)}{S} \to S\itp{\mathcal{A}}{\carfunc(e)}{S}$.

To establish the correspondence of weakest pre-conditions between $M$ and $\uncurry{M}$, we prepare a family of isomorphisms for interpretations of types through the SPS transformation.

\begin{definition} \label{def:rho}
  For any type $\bt$, 
  we define the isomorphism $\rho_{\bt}\colon \itp{\mathcal{A}_S}{\bt}{S} \rightarrow \itp{\mathcal{A}_T}{\primetrans{\bt}}{T}$ for \emph{types}, where $\bt'$ is defined in \cref{def:uncurry}, as follows:
  1) if $\bt = \mathbf{0}$, $\bt = \mathbf{1}$, or $\bt \defeq b$, then the map $\rho_{\bt}$ is the identity;
  2) if $\bt = \bt_1\times \bt_2$, then $\rho_{\bt} \defeq \rho_{\bt_1}\times \rho_{\bt_2}$;
  3) if $\bt = \bt_1 + \bt_2$, then $\rho_{\bt} \defeq \rho_{\bt_1} + \rho_{\bt_2}$; and
  4) if $\bt = \bt_1\rightarrow \bt_2$, then $\rho_{\bt}$ is defined by $\rho_{\bt}\defeq \cong \circ S(\rho_{\bt_2})^{\rho_{\bt_1}^{-1}}$,  where $\cong$ is the isomorphism $\cong\colon \big(T(\itp{\mathcal{A}_T}{\primetrans{\bt_2}}{T}\times Y)^{Y}\big)^{\itp{\mathcal{A}_T}{\primetrans{\bt_1}}{T}}\rightarrow \big(T(\itp{\mathcal{A}_T}{\primetrans{\bt_2}}{T}\times Y)\big)^{\itp{\mathcal{A}_T}{\primetrans{\bt_1}}{T}\times Y}$.

For any $\Gamma \defeq x_1\colon \bt_1,\dots, x_m\colon \bt_m$, we define the isomorphism $\rho_{\Gamma}$ for \emph{contexts} by $\rho_{\Gamma}\defeq \rho_{\bt_1\times\dots \times  \bt_m}$.
\end{definition}

We then present a key lemma to show the correctness of the SPS transformation as follows: 
it says that the interpretation $\itp{\mathcal{A}}{\uncurry{M}}{T}$ of the transformed term is \emph{almost} the transpose of the interpretation $\itp{\mathcal{A}}{M}{S}$ of the source program. 
\begin{lemma}
  \label{lem:keyLemma}
  Let $\Gamma\vdash M\colon \bt$ be a well-typed term. Then we have
  \begin{equation}\label{eq:uncurryTransKeyLemma}
  \itp{\mathcal{A}}{\uncurry{M}}{T} = (S(\rho_{\bt})\circ  \itp{\mathcal{A}}{M}{S}\circ \rho^{-1}_{\Gamma})^\dagger.
  \end{equation}
  In particular,  if all types in $\Gamma$ and $\bt$ are ground types, then we have $\itp{\mathcal{A}}{\uncurry{M}}{T} = \itp{\mathcal{A}}{M}{S}^{\dagger}$. 
  \qed
\end{lemma}
We prove this key lemma in~\cref{sec:fibrationalLambdac}. 
Finally, we have the following main theorem, the correctness of the SPS transformation. 
\begin{theorem}[correctness]
  \label{thm:mainwithoutrecursion}
  Let $\tau^T\colon T(\Omega) \to \Omega$ be an Eilenberg-Moore algebra.
  Given a well-typed term $\Gamma\vdash M\colon \bt$, and a post condition $Q\colon \itp{\mathcal{A}}{\bt}{S}\rightarrow \Omega^{Y}$,  the following equality holds:
  \[
  \big(\wpcond{M}{\rtanmod{T(\_\times Y)^Y}{\tau}}{Q}{\mathcal{A}}\circ \rho^{-1}_{\Gamma}\big)^{\dagger} =  \wpcond{\uncurry{M}}{\tau^T}{(Q\circ \rho^{-1}_{\bt})^{\dagger}}{\uncurry{\mathcal{A}}}.
  \]
  Here, $\tau^{T(\_ \times Y)^Y}$ is defined as in \cref{prop:liftingCorrectness}.

  In particular, if all types in $\Gamma$ and $\bt$ are ground types, then it immediately induces the following:
  \begin{equation}
    \label{eq:thmuncurry}
    \big(\wpcond{M}{\rtanmod{T(\_\times Y)^Y}{\tau}}{Q}{\mathcal{A}}\big)^{\dagger} =  \wpcond{\uncurry{M}}{\tau^T}{Q^{\dagger}}{\uncurry{\mathcal{A}}}.
  \end{equation}
\end{theorem}
\begin{proof}

We can see that the following diagram commutes:

  \adjustbox{scale=0.9,center}{
      \begin{tikzcd}
        S(\itp{\mathcal{A}}{\primetrans{\bt}}{T})\times Y \arrow[r, "S(\rho^{-1}_{\bt})\times \id"] \arrow[d, "\evsyb"] & S(\itp{\mathcal{A}}{\bt}{S})\times Y \arrow[rr, "S(Q)\times \id"] \arrow[d, "\evsyb"]&& S(\Omega^Y)\times Y \arrow[rr, "T(\ev{\Omega}{Y})^Y\times \id"] \arrow[d, "\evsyb"]&& T(\Omega)^Y\times Y \arrow[r, "\tau^Y\times \id"] \arrow[d, "\evsyb"]&  \Omega^Y\times Y \arrow[d, "\evsyb"]\\
        T(\itp{\mathcal{A}}{\primetrans{\bt}}{T}\times Y) \arrow[r, "T(\rho^{-1}_{\bt}\times \id)"] & T(\itp{\mathcal{A}}{\bt}{S}\times Y) \arrow[rr, "T(Q\times \id)"] && T(\Omega^Y\times Y) \arrow[rr, "T(\evsyb)"] && T(\Omega) \arrow[r, "\tau"] & \Omega
      \end{tikzcd}
  }
  By~\cref{lem:keyLemma}, we can show the desired equality.
\end{proof}

By~\cref{thm:mainwithoutrecursion}, we reduce temporal verification problems to rather simple problems, including reachability or reachability probabilities, which is summarised in~\cref{fig:list_of_coalgebraic_inference}. 

\begin{example}[angelic nondeterminism]
  \label{ex:uncurryAnon}
  We continue~\cref{ex:synchAnon}. By the RHS of~\cref{eq:correctSynch}, the post-condition we have to consider is $q\circ Q\colon \itp{\mathcal{A}}{\bt}{S}\rightarrow \boolsets^Y$. 
  By~\cref{eq:thmuncurry}, to solve the target verification problem it suffices to compute $\wpcond{\uncurry{M}}{\tau^{\hoaremonad}}{(q\circ Q)^{\dagger}}{\uncurry{\mathcal{A}}}$ for $\Gamma\vdash M\colon \bt$ such that all types in $\Gamma$ and $\bt$ are ground types, where $(q\circ Q)^{\dagger}\colon \itp{\mathcal{A}}{\bt}{S}\times Y\rightarrow \boolsets$ is the function such that $(q\circ Q)^{\dagger}(x, y) = \top$ iff $y$ is an accepting state of the DFA. 
  The Eilenberg-Moore algebra $\tau^{\hoaremonad}$ is precisely the one introduced in~\cref{ex:angndPS}, thus $\wpcond{\uncurry{M}}{\tau^{\hoaremonad}}{(q\circ Q)^{\dagger}}{\uncurry{\mathcal{A}}}$ can be thought as a standard reachability property.
\end{example}

\section{Recursion}
\label{sec:recursion}
In this section, we extend our framework for $\lambda_c$-calculus by including recursive functions. 
The resulting calculus is called \emph{$\lambdacfix$-calculus}.
Following~\cite{Katsumata13}, we assume that the base category $\mathbb{C}$ is an \emph{$\omegaCPO$-enriched bi-CCC}, that is, each homset is an $\omega$CPO (that may not have a least element), and compositions, tupling, cotupling, and currying of the bi-CC structure are all $\omega$-continuous. 
We write  $\mathbb{C}_0$ for the underlying CCC of $\mathbb{C}$, which forgets the $\omega$CPO-structure on each homset.

\begin{definition}[\cite{Katsumata13}]
  A \emph{pseudo-lifting strong monad} $T$ on  $\mathbb{C}$ is a strong monad on  $\mathbb{C}_0$ such that there is a generic effect $\leastelement{T}{0}\colon 1\rightarrow T(0)$, and for each $X \in \mathbb{C}$, $\leastelement{T}{X}\defeq T(!_X) \circ \leastelement{T}{0}$ is the least element in $\mathbf{C}(1, TX)$, where $0$ and $1$ are initial and terminal objects, respectively. 
\end{definition}
Note that 
for each strong monad morphism
$\alpha\colon T\Rightarrow S$ on $\mathbb{C}_0$,
any morphism $\leastelement{S}{X}$ that is induced by the generic effect $\leastelement{S}{0} \defeq \alpha\circ \leastelement{T}{0}$ is the least element in $\mathbb{C}(1, SX)$; see~\cite{Katsumata13}.

\begin{definition}[\cite{Katsumata13}] \label{def:lambdacfix}
  An \emph{$\lambdacfix(\Sigma)$-structure} is a 
  tuple $\mathcal{A} \defeq (\mathbb{C}, T, A, a)$ where $T$ is a pseudo-lifting strong monad on $\mathbb{C}$ and $(\mathbb{C}_0, T, A, a)$ is a $\lambda_c(\Sigma)$-structure.
\end{definition}

With a $\lambdacfix(\Sigma)$-structure, we add  the following recursion rule to the typing rules. 
\begin{mathpar}
  \inferrule*[Left=$\rec$]{
    \Gamma,\, f\colon \bt_1\rightarrow \bt_2,\, x\colon \bt_1 \vdash M\colon \bt_2
	}{
    \Gamma \vdash \mu f x.\, M\colon \bt_1\rightarrow \bt_2
	}
\end{mathpar}

Following~\cite{Katsumata13}, we define the interpretation of recursion as a \emph{stable uniform call-by-value fixpoint
operator}~\cite{HasegawaK02}. Note that stable uniform call-by-value fixpoint operators bijectively correspond to \emph{uniform $T$-fixpoint operators}~\cite{SimpsonP00}; see~\cite{HasegawaK02} for details. 
\begin{definition}[parametrised fixed point~\cite{SimpsonP00,HasegawaK02,Katsumata13}]
  Given $f\colon X\times T(Z)\rightarrow T(Z)$, its \emph{(parametrised) fixed point} $\fixop{X}{Z}{f}\colon X\rightarrow T(Z)$ is given by $\pi_2\circ \bigvee_{n\in \nat} \langle\pi_1, f \rangle^{n}\circ \langle \id_X, \leastelement{T}{Z} \rangle$.
\end{definition}
\begin{definition}[interpretation~\cite{HasegawaK02,Katsumata13}]
  Given a well-typed term $\Gamma \vdash \mu f x.\, M\colon \bt_1\rightarrow \bt_2$, the \emph{interpretation}
  $\itp{\mathcal{A}}{\mu f x.\, M}{T}$ is given by
  \[
    \eta^T_{T(\itp{\mathcal{A}}{\bt_2}{T})^{\itp{\mathcal{A}}{\bt_1}{T}}}\circ  \cal{T}{\itp{\mathcal{A}}{\bt_1}{T}}{\itp{\mathcal{A}}{\bt_2}{T}}\circ \fixop{\itp{\mathcal{A}}{\Gamma}{T}}{\itp{\mathcal{A}}{\bt_1\rightarrow \bt_2}{T}}{\itp{\mathcal{A}}{\lambda x.\ M}{T}\circ \id_{\itp{\mathcal{A}}{\Gamma}{T}}\times \cal{T}{\itp{\mathcal{A}}{\bt_1}{T}}{\itp{\mathcal{A}}{\bt_2}{T}}}.
  \]
  Recall that $\cal{T}{X}{Z}$ is the canonical Eilenberg-Moore algebra $\cal{T}{X}{Z}\colon T\big(T(Z)^X\big)\rightarrow T(Z)^X$.

\end{definition}

We then define our SPS transformation of the recursive function $\mu f x.\, M$. 
\begin{definition}[SPS transformation: recursion]
  Given a well-typed term $\Gamma \vdash\mu f x.\, M\colon \bt_1\rightarrow \bt_2$, we define the \emph{SPS transformation} $\uncurry{(\mu f x. M)}$ by
  $(\mu f z.\, \uncurry{M}[\pi_1 z/x, \pi_2 z/y], y)$.
\end{definition}

  We now extend the results for $\lambda_c(\Sigma)$-calculus to $\lambdacfix(\Sigma)$-calculus.
  The next lemma corresponds to \cref{lem:keyLemma} in the setting with recursive functions.
\begin{lemma}
  \label{lem:keyrecursion}
  Let $\Gamma \vdash M \colon \bt$ be a well-typed term in $\lambdacfix(\Sigma)$-calculus.
  Then \cref{eq:uncurryTransKeyLemma} in~\cref{lem:keyLemma} holds.
  \qed
\end{lemma}

\begin{theorem}
  \label{thm:extendedMain}
  The extensions of~\cref{thm:correctProd,thm:mainwithoutrecursion} that allow recursion hold. 
\end{theorem}
\begin{proof}
  It is known that~\cref{lem:monadTransGroundTypes} can be extended for recursion (see~\cite{Katsumata13,Kura2023}), which implies the extension of~\cref{thm:correctProd} that allows recursion by the same argument. 
  We can easily see that \cref{lem:keyrecursion} implies the extension of~\cref{thm:mainwithoutrecursion} that allows recursion as well. 
\end{proof}
As with~\cref{lem:keyLemma},
the proof of~\cref{lem:keyrecursion} is deferred to \cref{sec:fibrationalLambdacfix}.

\section{A Fibrational Account for SPS transformation}
\label{sec:fibrationalapproach}
In this section, we aim to provide a proof of our key lemmas---\cref{lem:keyLemma,lem:keyrecursion}---within a fibrational framework.
These lemmas play a crucial role in proving~\cref{thm:mainfib,thm:mainfibfix}, which justify the SPS transformation at the level of weakest pre-conditions.

Our approach builds on the fibrational logical relation~\cite{hermidaThesis, Katsumata13}, which is a categorical formalization of logical relations.
It enables us to lift the structures arising in the interpretation of types and terms to relational structures, offering a clear way to prove~\cref{lem:keyLemma,lem:keyrecursion}.
To lift these structures, we adopt a fibration with bi-CCC structures, as introduced in~\cite{Katsumata13}.

\begin{definition}
  We say that $p\colon \mathbb{E} \to \mathbb{C}$ is a \emph{fibration for logical relations}
  if it satisfies the following conditions:
  (i) both $\mathbb{E}$ and $\mathbb{C}$ are bi-CCC,
  (ii) $p$ is a partial order bifibration with fibrewise small products, and
  (iii) $p$ strictly preserves the bi-CCC structure.
\end{definition}
It is immediate that such fibrations are faithful since they are partial order bifibrations.
For functors $p\colon \mathbb{E} \to \mathbb{B}$ and $q\colon \mathbb{F} \to \mathbb{C}$,
we say that a functor $\dot{F}\colon \mathbb{E} \to \mathbb{F}$ is a \emph{lifting of a functor $F\colon \mathbb{B} \to \mathbb{C}$}
if it satisfies $q \circ \dot{F} = F \circ p$.

\subsection{For $\lambda_c$-calculus} \label{sec:fibrationalLambdac}
\newcommand{\natdisc}[1]{{d_{#1}}} 
\newcommand{\natcoprod}[1]{{g_{#1}}} 
We begin by extending the interpretation $\itp{A}{\_}{T}$ to a new interpretation $\litp{A}{\_}{T}{L}$, in which the interpretation of contexts and types are wrapped by a functor $L$.
This allows us to obtain the two interpretations $\itp{A}{M}{S}$ and $\itp{A}{\uncurry{M}}{T}$, which appear in \cref{lem:keyLemma},
as instances of the interpretation $\litp{A}{M}{T}{L}$
by varying parameters $T$ and $L$.
\begin{definition}[extended $\lambda_c(\Sigma)$-structure and $\mathcal{A}\litp{\mathcal{A}}{\_}{T}{L}$] \label{def:eval_L}
  Let $L$ be a strong endo functor $(\mathbb{C}, L, \lstrength{L}{})$ with natural transformations 
  $\natcoprod{}\colon L \circ (+) \Rightarrow (+) \circ (L\times L)$ and
  $\natdisc{}\colon L \Rightarrow \id$.
  An \emph{extended $\lambda_c(\Sigma)$-structure} for the functor $L$ is a tuple $(\mathbb{C}, T, A, a)$ defined 
  in the same way as $\lambda_c(\Sigma)$-structures (see \cref{def:lambdac}),
  except for the definition of the interpretation of types and the assignment $a$.
  The \emph{interpretation} 
  $\mathcal{A}\litp{\mathcal{A}}{\_}{T}{L}\colon \type \to \mathbb{C}$ of types, or shortly written as $\litp{\mathcal{A}}{\_}{T}{L}$, is defined in the same way as $\aitp{\mathcal{A}}{\_}{T}$ (see \cref{def:lambdac}), except that the function types are interpreted as follows:
  $\litp{\mathcal{A}}{\bt_1 \to \bt_2}{T}{L} = (TL\litp{A}{\bt_2}{T}{L})^{L\litp{A}{\bt_1}{T}{L}}$.
 For a context $\Gamma = x_1\colon \bt_1, \cdots, x_m\colon \bt_m$,
 we define $\litp{\mathcal{A}}{\Gamma}{T}{L} \coloneqq \prod_i \litp{\mathcal{A}}{\bt_i}{T}{L}$.
 For each $c \in K$ and $e \in E$, the \emph{assignments} are defined as morphisms
  $a(c)\colon L\litp{A}{\arfunc(c)}{T}{L} \to L\litp{A}{\carfunc(c)}{T}{L}$
  and
  $a(e)\colon L\litp{A}{\arfunc(e)}{T}{L} \rightarrow TL\litp{A}{\carfunc(e)}{T}{L}$.

Given a $\lambda_c(\Sigma)$-structure $\mathcal{A}$,
we define
the \emph{interpretation} $\mathcal{A}\litp{\mathcal{A}}{M}{T}{L}\colon L\litp{A}{\Gamma}{T}{L} \to TL\litp{A}{\bt}{T}{L}$ (or shortly, $\litp{\mathcal{A}}{M}{T}{L}$) of a well-typed term $\Gamma \vdash M\colon \bt$
by the following:
\begin{align*}
  \mathcal{A}\litp{A}{x}{T}{L} &\coloneqq \eta_{L\litp{\tstructure}{\bt_i}{T}{L}} \circ L\pi_i,
  \quad \text{where }\Gamma_i = (x\colon \bt), \\
  \mathcal{A}\litp{A}{(M, N)}{T}{L} &\coloneqq \mu \circ T^2(\lstrength{L}{\litp{A}{\bt_1}{T}{L}, \litp{A}{\bt_2}{T}{L}}) \circ T(\lstrength{T}{\litp{A}{\bt_1}{T}{L}, \litp{A}{\bt_2}{T}{L}}) \circ T(\id \times \mathcal{A}\litp{A}{N}{T}{L}) \\ 
  &\qquad \circ T(\langle \pi_1 \natdisc{\litp{A}{\bt_1}{T}{L} \times \litp{A}{\bt_2}{T}{L}}, L\pi_2 \rangle \circ \rstrength{L}{\litp{A}{\bt_1}{T}{L}, \litp{A}{\Gamma}{T}{L}})  \circ \rstrength{T}{L\litp{A}{\bt_1}{T}{L}, \litp{A}{\Gamma}{T}{L}} \circ \langle \mathcal{A}\litp{\mathcal{A}}{M}{T}{L}, \natdisc{\litp{A}{\Gamma}{T}{L}}\rangle, \\
  &\qquad \text{where }\Gamma\vdash M\colon \bt_1\text{, and } \Gamma\vdash N\colon \bt_2.
\end{align*}
See \cref{ap:litp} for the other terms.
Note that when $L=\id$ (with trivial components $\lstrength{L}{}, g, d$), this extended interpretation coincides with the original one, i.e., $\aitp{\mathcal{A}}{M}{T} = \mathcal{A}\litp{\mathcal{A}}{M}{T}{\id}$.
\end{definition}

The following proposition shows that the interpretation $\aitp{\tstructure}{\uncurry{M}}{T}$ of the SPS transformation of a well-typed term $M$ in a $\lambda_c(\Sigma)$-structure $\tstructure$ can be expressed as the extended interpretation $\litp{\tstructure}{M}{T}{L}$ for a suitable functor $L$.
Throughout this section, we assume the setting of  \cref{sec:uncurrying}.
Recall that in~\cref{sec:uncurrying},
given a $\lambda_c(\Sigma)$-structure $\sstructure = (\mathbb{C}, S \coloneqq T(\_ \times Y)^Y, A, a)$,
a $\lambda_c(\uncurry{\Sigma})$-structure $\tstructure = (\mathbb{C}, T, A, \uncurry{a})$ is defined in \cref{def:uncurry}.
\begin{proposition} \label{prop:trans}
  The following statements hold.
  \begin{enumerate}
    \item 
  For each  type $\bt$ and context $\Gamma$,
  we have
  $\itp{\mathcal{A}}{\uncurry{\bt}}{T} = \litp{\mathcal{A}}{\bt}{T}{\_ \times Y} \times Y$ and $\itp{\tstructure}{\uncurry{\Gamma}}{T} = \litp{\tstructure}{\Gamma}{T}{\_ \times Y} \times Y$.
  \item
  The $\lambda_c(\uncurry{\Sigma})$-structure $\tstructure$ is an extended $\lambda_c(\Sigma)$-structure for the functor $(\_ \times Y) \colon \mathbb{C} \to \mathbb{C}$.
  \item
  For any well-typed term $\Gamma \vdash M\colon \bt$, the SPS transformation $\uncurry{M}$ defined in
  \cref{def:uncurry} 
  satisfies
  $\aitp{\tstructure}{\uncurry{M}}{T} = \tstructure\litp{\tstructure}{M}{T}{\_ \times Y}$.
  \end{enumerate}
\end{proposition}
\begin{proof}
  By induction on the structure of terms. 
  See \cref{ap:transproof} for the detailed proof.
\end{proof}

\newcommand{\yoneda}{\mathbf{y}}
Finally, we reformulate \cref{lem:keyLemma} within a fibrational framework.
To this end,
we define an extended $\lambda_c(\Sigma)$-structure
$\sstructure \times \tstructure$ for the functor $L \coloneqq \id \times ((\_) \times Y)\colon \mathbb{C}^2 \to \mathbb{C}^2$ to be the product structure $(\mathbb{C}^2, S \times T, \langle A, A\rangle, a \times \uncurry{a})$.
Under this structure, a well-typed term $M$ is interpreted as $\litp{A}{M}{S \times T}{L} = \itp{\mathcal{A}}{M}{S} \times \litp{\mathcal{A}}{M}{T}{\_ \times Y}$ where $\litp{\mathcal{A}}{M}{T}{\_ \times Y} = \itp{\mathcal{A}}{\uncurry{M}}{T}$ by \cref{prop:trans}.
Moreover, the isomorphism $\rho_\bt$ defined in~\cref{def:rho} gives an isomorphism $\itp{\mathcal{A}_S}{\bt}{S} \cong \litp{\mathcal{A}_T}{\bt}{T}{\_ \times Y} \ (= \itp{\mathcal{A}_T}{\primetrans{\bt}}{T})$.
Here we assume that $\mathbb{C}$ is small.
We introduce a fibration $p_D\colon \mathbb{K}_D \to \mathbb{C}^2$ by the following change-of-base diagram:
\begin{displaymath}
  \xymatrix{
    \mathbb{K}_D \pullbackmark{1, 0}{0, 1}\ar[d]_{p_D} \ar[r] &\mathbf{Sub}([\mathbb{C}^\op, \sets]) \ar[d] \\
    \mathbb{C}^2 \ar[r]^-{D} &[\mathbb{C}^\op, \sets].
  }
\end{displaymath}
along the functor $D\colon \mathbb{C}^2 \to [\mathbb{C}^\op, \sets]$ defined by
$D(I, J) = \yoneda I \times  (\yoneda J \circ (\_ \times Y))$ for $(I, J) \in \mathbb{C}^2$ and
$D(f, g) = (f \circ \_) \times (g \circ \_)$ for $(f, g)\colon (I, J) \to (I', J')$
where 
$\yoneda\colon \mathbb{C} \to [\mathbb{C}^\op, \sets]$ is the Yoneda embedding.
The right vertical functor in the diagram is the subobject fibration.
Since the functor $D$ preserves finite products, it is a fibration for logical relations by \cite[Example 3.3]{Katsumata13}.

To establish the correspondence between the interpretations $\itp{\mathcal{A}}{M}{S}$ and $\litp{\mathcal{A}}{M}{T}{\_ \times Y}$, we introduce two functors $V, C\colon \type \to \mathbb{K}_{D}$, which can be seen as logical relations.
For each type $\bt$ and object $H \in \mathbb{C}$, we define
\begin{align*}
V(\bt)(H) &= \big\{(f, \rho_\bt \circ f \circ \pi_1) \ \big|\  f \in \mathbb{C}\big(H, \itp{\mathcal{A}}{\bt}{S}\big)\big\} \subseteq \mathbb{C}\big(H, \itp{\mathcal{A}}{\bt}{S}\big) \times \mathbb{C}\big(H \times Y, \litp{A}{\bt}{T}{\_ \times Y}\big), \\
C(\bt)(H) &\coloneqq \big\{\big(f, T(\rho_\bt \times \id_Y) \circ f^\dagger\big) \ \big|\  f \in \mathbb{C}\big(H, S\itp{\mathcal{A}}{\bt}{S}\big)\big\} \subseteq \mathbb{C}\big(H, S\itp{\mathcal{A}}{\bt}{S}\big) \times \mathbb{C}\big(H \times Y, T(\litp{A}{\bt}{T}{\_ \times Y} \times Y)\big),
\end{align*}
and for each $h\colon H \to H'$ in $\mathbb{C}$, both $V(\bt)(h)$ and $C(\bt)(h)$ are defined as $(\_ \circ h) \times (\_ \circ h \times \id_Y)$.
Letting 
$\dot{L}\colon \mathbb{K}_{D} \to \mathbb{K}_{D}$ be a lifting of $L$ defined by $\dot{L}(P) = \{(f, \langle g, \pi_2 \rangle) \mid (f, g) \in P\}$,
the functors $\dot{L}V$ and $C$ are liftings of the functors
$\litp{A}{\_}{S \times T}{L}$ and $(S \times T)\litp{A}{\_}{S \times T}{L}$, respectively.
We extend the definitions of $V$ and $C$ to
contexts $\Gamma \defeq x_1\colon \bt_1,\dots,x_m\colon \bt_m$ by 
identifying $\Gamma$ with $\prod_{i=1}^m x_i$.
Note that $V$ preserves the (monoidal) product $\prod$ on $\type$: 
$V(\Gamma) = \prod_{i=1}^m V(\bt_i)$ and $V(\mathbf{1})$ is the terminal object in $\mathbb{K}_D$.

\begin{theorem} \label{thm:mainfib}
  For a well-typed term $\Gamma \vdash M\colon \bt$, 
  there exists a morphism $\dot{L}V(\Gamma) \to C(\bt)$
  above 
  $\litp{\mathcal{A}}{M}{S \times T}{L}\colon L\litp{\mathcal{A}}{\Gamma}{S \times T}{L} \to (S \times T)L\litp{\mathcal{A}}{\bt}{S \times T}{L}$ along the fibration $p_D$.
  \qed
\end{theorem}
See \cref{ap:proof_liftlambdac} for the proof.
In~\cite{Katsumata13}, Katsumata also show a correspondence between two interpretations under different computational effects, with a fibrational framework. 
He constructs logical relations as interpretations of types by using the so-called \emph{categorical $\top\top$-lifting}~\cite{Katsumata05}, and utilizes the structures used in interpretations to prove the correspondence.
In contrast, we define logical relations as just functors. Although these logical relations are not constructed via bi-CCC structures or monad structures, they still satisfy properties similar to those of logical relations in~\cite{Katsumata13}; see \cref{lem:liftcv}.

Since $(\id, \rho_\Gamma \times \id)\in \dot{L}V(\Gamma)(\itp{\mathcal{A}}{\Gamma}{S})$ holds, 
such a morphism $\dot{L}V(\Gamma) \to C(\bt)$ above $\litp{\mathcal{A}}{M}{S \times T}{L}$ discussed in \cref{thm:mainfib} induces 
$(\itp{\mathcal{A}}{M}{S}, \itp{\mathcal{A}}{\uncurry{M}}{T} \circ (\rho_\Gamma \times \id)) \in C(\bt)(\itp{\mathcal{A}}{\Gamma}{S})$,
which is equivalent to 
\cref{eq:uncurryTransKeyLemma} in \cref{lem:keyLemma}.
The converse direction also holds, that is, \cref{eq:uncurryTransKeyLemma} implies the existence of a morphism $\dot{L}V(\Gamma) \to C(\bt)$ above $\litp{\mathcal{A}}{M}{S \times T}{L}$.

\subsection{For $\lambdacfix$-calculus} \label{sec:fibrationalLambdacfix}
We now extend the fibrational approach for the $\lambdacfix$-calculus introduced in \cref{sec:recursion}.
We assume that the base category $\mathbb{C}$ is an \emph{$\omegaCPO$-enriched bi-CCC}.
As in \cref{sec:fibrationalLambdac}, we first extend the interpretation $\itp{A}{\_}{T}$.

\begin{definition}[extended $\lambdacfix$-structure]
  Let $L$ be a strong functor $(\mathbb{C}_0, L, \lstrength{L}{})$ with natural transformations 
  $\natcoprod{}\colon L \circ (+) \Rightarrow (+) \circ (L\times L)$ and
  $\natdisc{}\colon L \Rightarrow \id$.
  An \emph{extended $\lambdacfix(\Sigma)$-structure} for the functor $L$
  is a tuple $(\mathbb{C}, T, A, a)$ where $T$ is a pseudo-lifting strong monad on $\mathbb{C}$ and $(\mathbb{C}_0, T, A, a)$ is an extended $\lambdacfix(\Sigma)$-structure.
\end{definition}
Given an extended $\lambdacfix(\Sigma)$-structure $\mathcal{A}$, 
the interpretation $\litp{\mathcal{A}}{M}{T}{L}$ 
of a well-typed term $\Gamma \vdash M\colon \bt$ is defined 
in the same way as for an extended $\lambda_c(\Sigma)$-structure, with the following additional rule to interpret recursive terms:
\begin{align*}
  \litp{A}{\mu fx.~M}{T}{L} &\coloneqq 
    \eta^T_{L\big(TL(\litp{\mathcal{A}}{\bt_2}{T}{L})^{L\litp{\mathcal{A}}{\bt_1}{T}{L}}\big)}\circ  L\cal{T}{L\litp{\mathcal{A}}{\bt_1}{T}{L}}{L\litp{\mathcal{A}}{\bt_2}{T}{L}}\circ L\fixop{\litp{\mathcal{A}}{\Gamma}{T}{L}}{\litp{\mathcal{A}}{\bt_1\rightarrow \bt_2}{T}{L}}{
      \eta \circ (\litp{\mathcal{A}}{M}{T}{L} \circ \lstrength{L}{})^\dagger \circ \id \times \cal{T}{L\litp{\mathcal{A}}{\bt_1}{T}{L}}{L\litp{\mathcal{A}}{\bt_2}{T}{L}}}.
\end{align*}
Again, it is consistent with $\itp{\mathcal{A}}{\mu fx.~M}{T}$ in \cref{def:lambdacfix} when $L = \id$, i.e.~$\itp{\mathcal{A}}{\mu fx.~M}{T} = \litp{\mathcal{A}}{\mu fx.~M}{T}{\id}$.

Then we can extend the results \cref{prop:trans} and \cref{thm:mainfib} to the setting with recursive functions.
See \cref{ap:proof_transLambdacfix} and \cref{ap:proof_mainfibfix} for the proofs.
\begin{proposition} \label{prop:transLambdacfix}
  The $\lambdacfix(\uncurry{\Sigma})$-structure $\tstructure$ (see \cref{sec:uncurrying} and \cref{sec:recursion}) can be seen as an extended $\lambdacfix(\Sigma)$-structure for the functor $(\_ \times Y) \colon \mathbb{C} \to \mathbb{C}$.
  Moreover, for any well-typed term $\Gamma \vdash \mu f x.~M\colon \bt$, $\aitp{\tstructure}{\uncurry{(\mu fx.~M)}}{T} = \tstructure\litp{\tstructure}{\mu fx.~M}{T}{\_ \times Y}$ holds.
  \qed
\end{proposition}

\begin{theorem} \label{thm:mainfibfix}
  The statement of~\cref{thm:mainfib} holds for the extended $\lambdacfix(\Sigma)$-structure $\sstructure \times \tstructure$.
  \qed
\end{theorem}

\section{Case Studies}
\label{sec:caseStudy}

Finally, we present our case studies for our framework. 
These case studies include temporal verification problems for probabilistic programs and angelic non-deterministic programs. 
We use the category $\dCPO$ of directed complete partial orders as the base category for all case studies; note that we can always take small subcategories of $\dCPO$ with a suitalby chosen Grothendieck universe.   
All case studies consider only terminating traces; the term "terminating" is omitted throughout.

We first introduce a simple temporal specification, called \emph{coalgebraic DFA}, before moving to the case studies.

\begin{definition}[coalgebraic DFA]
  A \emph{coalgebraic DFA} is a function $d\colon Y\rightarrow Y^A\times \boolsets$, where 
  $Y$ is a finite set of  \emph{states}, $A$ is a finite set of \emph{characters}, and $\boolsets$ is the set of Booleans $\{\top, \bot\}$.
\end{definition}
Later, we regard $d$ as a morphism in $\dCPO$ with the discrete ordering. 

We define the \emph{semantics} and \emph{acceptance condition} of coalgebraic DFAs in a manner analogous to the standard concrete definition.

\begin{definition}[semantics of coalgebraic DFA]
  Given a coalgebraic DFA $d\colon Y\rightarrow Y^A\times \boolsets$,
  we recursively define a \emph{coalgebra} $\widehat{d}\colon Y\rightarrow \gendfa{Y}$ by the run of $d$, that is, for each $y\in Y$, and $w\in A^{\ast}$, 
  \begin{align*}
    \widehat{d}(y)(w) \defeq \begin{cases*}
      y &\text{ if  $w = \epsilon$, }\\
      \widehat{d}\big(d(y)(a)\big)(w') &\text{ if $w = a\cdot w'$ for some $a\in A$.}
    \end{cases*}
  \end{align*}
  We call $\widehat{d}$ the \emph{semantics} of $d$.
  Given $w\in A^{\ast}$ and $y\in Y$, we say that 
  $y$ is \emph{accepting} if $(\pi_2\circ d)(y) = \top$, and
  $w$ is \emph{accepting} in $d$ from $y$ if $\widehat{d}(y)(w)$ is accepting.
\end{definition}

With the semantics of coalgebraic DFAs, we define a strong monad morphism that is an ingredient of synchronisation with coalgebraic DFAs. 

\begin{definition}[strong monad morphism induced by $d$]
  \label{def:monadMorDFA}
  Given a coalgebraic DFA $d\colon Y\rightarrow Y^A\times \boolsets$,
  we define a \emph{strong monad morphism} $\alpha^d\colon \_\times A^{\ast}\Rightarrow (\_\times Y)^{Y}$ as
  \begin{align*}
    \alpha^d_X(x, w)(y)\defeq \big(x, \widehat{d}(y)(w)\big), \text{ for each $X$.}
  \end{align*}
\end{definition}
See~\cref{subsec:omitproofEx1} for the proof that shows $\alpha$ is indeed a strong monad morphism.

We fix a DFA $d\colon Y\rightarrow Y^A \times \boolsets$ as a specification DFA for the case studies in~\cref{subsec:probAcceptingTrace,subsec:expectedRewardAcceptingTrace,subsec:emptinessCheckingAcceptingTraces}.

\subsection{Probability of Accepting Traces}
\label{subsec:probAcceptingTrace}

We first study the temporal verification of a higher-order probabilistic program, which concerns computing the probability of accepting (terminating) traces generated by the program.
To this end, we instantiate our framework in the category $\dCPO$ of dcpos, equipped with a commutative probability monad $\comProbPowerdommonad$~\cite{JiaLMZ21}; see~\cref{subsec:preliminaries} for background on $\comProbPowerdommonad$.
To apply our framework, it suffices to identify an appropriate product situation that yields the desired synchronisation.

  We employ a product situation $(\comProbPowerdommonad, (A^{\ast}, =), (Y, =), \tau^{\_ \times A^*}, \tau^\comProbPowerdommonad, \alpha^d, q)$ on $\dCPO$ where:
  \begin{itemize}
    \item $\comProbPowerdommonad$ is the commutative probabilistic monad, and $\_\times A^{\ast}$ is the writer monad induced by the monoid structure on $A^\ast$. 
    \item $\tau^{\_ \times A^{\ast}}\colon A^{\ast}\times A^{\ast}\rightarrow A^{\ast}$ is the multiplication of $\_\times A^{\ast}$, and $\tau^{\comProbPowerdommonad}\colon \comProbPowerdommonad([0, 1])\rightarrow [0, 1]$ is the expectation $\tau^{\comProbPowerdommonad}(\nu)\defeq \int p  d\nu$ of the identity function $p\colon [0, 1]\rightarrow [0, 1]$, where the order on $[0, 1]$ is the standard total order.
    \item $q$ is the inference query $q\colon A^{\ast}\rightarrow [0, 1]^{Y}$ defined by $q(w)(y) \defeq 1$ if $w$ is accepting in $d$ from $y$, and  $q(w)(y) \defeq 0$ otherwise.
  \end{itemize}
  Note that we write $(X, =)$ for the discrete dcpo $X$. 
Then~\cref{prop:liftingCorrectness} yields an Eilenberg-Moore algebra $\tau^{\comProbPowerdommonad(\_\times A^{\ast})}\colon \comProbPowerdommonad\big(\comProbPowerdommonad(A^{\ast})\times A^{\ast}\big)\rightarrow \comProbPowerdommonad(A^{\ast})$ and an inference query $q^{\comProbPowerdommonad}\colon \comProbPowerdommonad(A^{\ast})\rightarrow [0, 1]^{Y}$. 
We now formally define the target problem as follows. 
\vspace{0.1cm}
\begin{mdframed}
{\bf Problem: Probability of Accepting Traces.}
 \quad
\noindent 
Let $\Gamma\vdash M\colon \bt$ be a $\comProbPowerdommonad(\_\times A^{\ast})$-effectful well-typed term  such that all types in $\Gamma$ and $\bt$ are ground types.
Given $x\in \itp{\mathcal{A}}{\Gamma}{\comProbPowerdommonad(\_\times A^{\ast})}$ 
and an initial state $y\in Y$, compute:
\begin{align*}
  \Big(q^{\comProbPowerdommonad}\circ \big(\wpcond{M}{\tau^{\comProbPowerdommonad(\_\times A^{\ast})}}{Q}{\mathcal{A}}\big)\Big)(x)(y)\in [0, 1],
\end{align*}
where $Q\colon \itp{\mathcal{A}}{\bt}{\comProbPowerdommonad(\_\times A^{\ast})}\rightarrow \comProbPowerdommonad(A^{\ast})$ is the constant function
assigning
the Dirac valuation $\eta_{A^{\ast}}(\epsilon)$ (see \cref{def:dirac}).
Note that
the weakest pre-condition $\big(\wpcond{M}{\tau^{\comProbPowerdommonad(\_\times A^{\ast})}}{Q}{\mathcal{A}}\big)(x)\in \comProbPowerdommonad(A^{\ast})$ computes the subdistribution of traces,
and $q^{\comProbPowerdommonad}\colon\comProbPowerdommonad(A^{\ast})\rightarrow [0, 1]^{Y}$ computes the expectation of accepting traces.
\end{mdframed}

\begin{example}
  \label{ex:hoRW}
  Consider the following simple higher-order random walk:
  \begin{equation*}
    \letrec{\randomwalk}{f\ x}{\ifexpr{x\geq 0}{\probbranch{(\randomwalk\ g\ f(x+2))^{\up}}{1/2}{(\randomwalk\ g\ f(x-2))^{\down}}}{()}}{\randomwalk\ (\lambda z. z)\ 1},
  \end{equation*}
  where $f\colon\typeReal\rightarrow \typeReal$, $x\colon \typeReal$, and $g\defeq \lambda z. (fz - 1)$; 
   following~\cite{KuraUnno2024}, we interpret the type $\typeReal$ by setting $\itp{\mathcal{A}}{\typeReal}{\comProbPowerdommonad} \defeq (\real, =)$.
Here, the variable $x$ represents the current position, and the program terminates when the position becomes strictly negative.
During execution, the position is updated probabilistically; the update rule is determined by a function $f$, which is itself defined recursively.
The program returns $\up$ if it attempts to increase the position---specifically, by calling the function $f$ with the argument $x + 2$---and returns $\down$ otherwise.
We are interested in the termination probability of the program under the condition that it attempts to increase the position at least twice.
  See~\cref{subsec:omitSpec} for the precise specification. 
  The SPS transformation of the program is 
  \begin{equation*}
    \letrec{\randomwalk}{f\ x\ y}{\ifexpr{x \geq 0}{N}{((), y)}}{\randomwalk\ (\lambda z. z)\ 1\ y},
  \end{equation*}
  where the term $N$ is given by 
  \begin{equation*}
    N \defeq \probbranch{(\randomwalk\ g\ f(x+2)\ (c^{\up} y))}{1/2}{(\randomwalk\ g\ f(x-2)\ y)},
  \end{equation*}
  with the effect-free constant $c^{\up}$ that changes the state $y$ by reading the character $\up$. 
\end{example}

\begin{remark}[optimisation and selective SPS transformation]
  To obtain a more compact expression and make the inputs more suitable for the solver, we apply several optimisations in our effectful SPS transformation.
  For instance, the effectful SPS transformation $\uncurry{((M)^{\heads})}$ of the event-raising effect used in~\cref{eq:exProbProg} corresponds to 
  $\uncurry{((M)^{\heads})} = \uncurry{\big(\big(e^{\heads}\, ()\big) ;\,  M\big)}  = (\uncurry{M})[y_2/y]$.
In general, such optimisations can be formally justified by a \emph{selective} extension of our effectful SPS transformation, for which a corresponding variant in CPS transformations has been extensively studied in~\cite{SatoUK13,KuraUnno2024,Nielsen01}.
  While implementing such automated optimisations remains as future work, we believe that this is feasible by following a procedure broadly similar to that employed in~\cite{SatoUK13}. 
\end{remark}

By~\cref{prop:liftingCorrectness,thm:extendedMain}, we immediately obtain the following result:

\begin{corollary}
  Given $\Gamma\vdash M\colon \bt$ such that all types in $\Gamma$ and $\bt$ are ground types, the temporal verification problem (LHS) is reduced to the computation of a weakest pre-condition (RHS) as follows:
  \begin{equation}
    \label{eq:correctProbAccTrace}
    \Big(q^{\comProbPowerdommonad}\circ \big(\wpcond{M}{\tau^{\comProbPowerdommonad(\_\times A^{\ast})}}{Q}{\mathcal{A}}\big)\Big)^{\dagger}= \wpcond{\uncurry{M}}{\tau^{\comProbPowerdommonad}}{(q^{\comProbPowerdommonad}\circ Q)^{\dagger}}{\uncurry{\beta^d(\mathcal{A})}},
  \end{equation}
  where $Q\colon \itp{\mathcal{A}}{\bt}{\comProbPowerdommonad(\_\times A^{\ast})}\rightarrow \comProbPowerdommonad(A^{\ast})$ is the 
  constant function assigning the Dirac valuation $\eta_{A^{\ast}}(\epsilon)$, 
  and $\beta^d$ is the strong monad morphism constructed by the product situation. \qed
\end{corollary}
By expanding the definition, the post-condition $(q^{\comProbPowerdommonad}\circ Q)^{\dagger}\colon \itp{\mathcal{A}}{\bt}{\comProbPowerdommonad}\times \itp{\mathcal{A}}{b^Y}{\comProbPowerdommonad}\rightarrow [0, 1]$ is given by $(q^{\comProbPowerdommonad}\circ Q)^{\dagger}(x, y) \defeq  1$ if $y$ is an accepting state, and  $(q^{\comProbPowerdommonad}\circ Q)^{\dagger}(x, y) \defeq  0$ otherwise, for any $x\in \itp{\mathcal{A}}{\bt}{\comProbPowerdommonad}$ and $y\in \itp{\mathcal{A}}{b^Y}{\comProbPowerdommonad} \ (=Y)$. 
Importantly, the (RHS) of~\cref{eq:correctProbAccTrace} is the standard weakest pre-condition for reachability probabilities, for which an automated solver is available~\cite{KuraUnno2024} to compute over approximations. We elaborate on this in~\cref{sec:experiments}.

\subsection{Partial Expected Rewards of Accepting Traces}
\label{subsec:expectedRewardAcceptingTrace}
Next, we incorporate rewards into probabilistic programs and consider \emph{partial} expected rewards of accepting traces. 
Here, ``partial'' indicates that the expectation is taken only over accepting traces; see e.g.~\cite{Baier0KW17,WatanabeJRH25}.

To formalize this setting, we define a product situation 
\begin{displaymath}
(\comProbPowerdommonad(\_\times \nonnegrat), (A^{\ast}, =), (Y, =), \tau^{\_ \times A^\ast}, \tau^{\comProbPowerdommonad(\_\times \nonnegrat)}, \alpha^d, q)
\end{displaymath}
 on $\dCPO$ where:
  \begin{itemize}
    \item $\comProbPowerdommonad(\_\times \nonnegrat)$ is the composite strong monad of the commutative probabilistic monad $\comProbPowerdommonad$ and the writer monad $\_\times \nonnegrat$ given by the monoid structure $(\nonnegrat, +, 0)$ with the discrete order. 
    Also, $\_\times A^{\ast}$ is the strong monad induced by the monoid structure on $A^\ast$. 
    \item $\tau^{\_ \times A^{\ast}}\colon A^{\ast}\times A^{\ast}\rightarrow A^{\ast}$ is the multiplication of $\_\times A^{\ast}$, and $\tau^{\comProbPowerdommonad(\_\times \nonnegrat)}\colon \comProbPowerdommonad([0, 1]\times [0, \infty]\times \nonnegrat)\rightarrow [0, 1]\times [0, \infty]$ is the pair of expectations $\tau^{\comProbPowerdommonad}(\nu)\defeq \big(\int_{p, r, n}  p d\nu, \int_{p, r, n} (p\cdot n + r) d\nu\big)$, where $(p, r, n)$ moves in $[0, 1]\times [0, \infty]\times \nonnegrat$, and the order in $[0, 1]\times [0, \infty]$ is the product of the standard total orders.
    \item $q$ is an inference query $q\colon A^{\ast}\rightarrow ([0, 1]\times [0, \infty])^{Y}$ defined by $q(w)(y) \defeq (1, 0)$ if $w$ is accepting in $d$ from $y$, and  $q(w)(y) \defeq (0, 0)$ otherwise.
  \end{itemize}
See~\cref{sec:omitProofPartialExp} for the proof that shows the data above is indeed a product situation. 
The inference query $q$ may appear unintuitive, as it returns a pair $(p, r) \in [0, 1] \times [0, \infty]$ for each input $(w, y)$.
This is because, in order to compute the partial expected reward iteratively, it is necessary to track not only the expected reward itself but also the reachability probability (see e.g.~\cite{WatanabeJRH25}).
Given this product situation, we formulate the target problem as follows.
\vspace{0.1cm}
\begin{mdframed}
{\bf Problem: Partial Expected Rewards of Accepting Traces.}
 \quad
\noindent 
Let $\Gamma\vdash M\colon \bt$ be a $\comProbPowerdommonad(\_\times A^{\ast}\times \nonnegrat)$-effectful well-typed term  such that all types in $\Gamma$ and $\bt$ are ground types.
Given $x\in \itp{\mathcal{A}}{\Gamma}{\comProbPowerdommonad(\_\times A^{\ast}\times \nonnegrat)}$ and an initial state $y\in Y$, compute:
\begin{align*}
  \Big(\pi_2\circ q^{\comProbPowerdommonad(\_\times\nonnegrat)}\circ \big(\wpcond{M}{\tau^{\comProbPowerdommonad(\_\times A^{\ast}\times \nonnegrat)}}{Q}{\mathcal{A}}\big)\Big)(x)(y)\in [0, \infty],
\end{align*}
where $Q\colon \itp{\mathcal{A}}{\bt}{\comProbPowerdommonad(\_\times A^{\ast}\times \nonnegrat)}\rightarrow \comProbPowerdommonad(A^{\ast}\times \nonnegrat)$ is the constant function assigning the Dirac valuation $\eta_{A^{\ast}\times \nonnegrat}(\epsilon, 0)$.
Note that the weakest pre-condition $\big(\wpcond{M}{\tau^{\comProbPowerdommonad(\_\times A^{\ast}\times \nonnegrat)}}{Q}{\mathcal{A}}\big)(x)\in \comProbPowerdommonad(A^{\ast}\times \nonnegrat)$ computes the subdistribution of traces and their cumulative rewards,
and $q^{\comProbPowerdommonad(\_\times\nonnegrat)}\colon\comProbPowerdommonad(A^{\ast}\times \nonnegrat)\rightarrow ([0, 1]\times [0, \infty])^{Y}$ computes the pair of the probability and (partial) expected cumulative reward of accepting traces.
\end{mdframed}

\begin{example}
  \label{ex:rewrdFO}
  Consider the following program: 
  \begin{equation*}
    \letrec{\gainreward}{x}{\exprobbranch{\big(\exprobbranch{{(\gainreward\ ())}^{\fail}}{1/2}{1}{{(\gainreward\ ())}^{\success}}\big)}{3/4}{0}{()}}{\gainreward\ ()},
  \end{equation*}
  where the generic effect $\exprobbranch{\_}{p}{r}{\_} \colon \unit \rightarrow \unit + \unit$ models the probabilistic branch that selects each branch with the probability $p$ and $1 - p$, respectively, while accumulating a reward of $r$ regardless of the chosen branch.
  The alphabet is $\{\fail, \success\}$.
  We then require that the program terminates without encountering any failures, denoted by $\fail$. This specification is illustrated as follows:
  
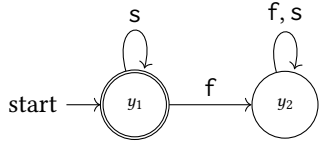
\begin{equation*}
  \label{eq:exProbSpec}
    \begin{tikzpicture}
          \node[state, initial, accepting] (y1) at (0, 0) {\tiny $y_1$};
          \node[state] (y2) at (2, 0) {\tiny $y_2$};
          \draw[->] (y1) to node[pos=0.5, inner sep=3pt, above] {$\fail$} (y2);
          \draw[->] (y1) edge [loop above] node[pos=0.5, inner sep=3pt, above] {$\success$} (y1);
          \draw[->] (y2) edge [loop above] node[pos=0.5, inner sep=3pt, above] {$\fail, \success$} (y2);
    \end{tikzpicture}
  \end{equation*}
  We then have the following term by the SPS-transformation.
  \begin{equation*}
    \letrec{\gainreward}{x\ y}{\exprobbranch{\big(\exprobbranch{\gainreward\ ((), y_2)}{1/2}{1}{\gainreward\ ((), y)}\big)}{3/4}{0}{((), y )}}{\gainreward\ ((), y)}.
  \end{equation*}
\end{example}

We immediately obtain the following corollary:
\begin{corollary}
  Given  $\Gamma\vdash M\colon \bt$ such that all types in $\Gamma$ and $\bt$ are ground types, the temporal verification problem (LHS) is reduced to the computation of a weakest pre-condition (RHS) as follows:
  \begin{align*}
    &\Big(q^{\comProbPowerdommonad(\_\times \nonnegrat)}\circ \big(\wpcond{M}{\tau^{\comProbPowerdommonad(\_\times A^{\ast}\times \nonnegrat)}}{Q}{\mathcal{A}}\big)\Big)^{\dagger}\\
     = &\wpcond{\uncurry{M}}{\tau^{\comProbPowerdommonad(\_\times \nonnegrat)}}{(q^{\comProbPowerdommonad(\_\times \nonnegrat)}\circ Q)^{\dagger}}{\uncurry{\beta^d(\mathcal{A})}},
  \end{align*}
  where $Q\colon \itp{\mathcal{A}}{\bt}{\comProbPowerdommonad(\_\times \nonnegrat)}\rightarrow \comProbPowerdommonad(A^{\ast}\times \nonnegrat)$ is the constant function assigning the Dirac valuation $\eta_{A^{\ast}\times \nonnegrat}(\epsilon, 0)$,
  and $\beta^d$ is the strong monad morphism constructed by the product situation. \qed
\end{corollary}
We note that the post condition $(q^{\comProbPowerdommonad(\_\times  \nonnegrat)}\circ Q)^{\dagger}\colon \itp{\mathcal{A}}{\bt}{\comProbPowerdommonad(\_\times \nonnegrat)}\times \itp{\mathcal{A}}{b^Y}{\comProbPowerdommonad(\_\times \nonnegrat)}\rightarrow [0, 1]\times [0,\infty]$ maps $(x, y)$ to $(1, 0)$
if $y$ is an accepting state, and $(0, 0)$ otherwise. 
We verify two examples, including~\cref{ex:rewrdFO}, for partial expected rewards with the solver~\cite{KuraUnno2024} in~\cref{sec:experiments}.

\subsection{Emptiness Checking of Accepting Traces}
\label{subsec:emptinessCheckingAcceptingTraces}
We then move on to non-deterministic programs: we consider angelic non-determinism that tries to satisfy the specification, which is reduced to the emptiness checking of accepting traces. 
We employ a \emph{lower monad} $\hoaremonad$, which is an extended Hoare powerdomain monad including the emptyset (see e.g.~\cite{AbramskyJ95}), to model an angelic non-determinism on the category $\dCPO$ of dcpos; see~\cref{subsec:preliminaries} for the definition. 

We define a product situation $(\hoaremonad, (A^{\ast}, =), (Y, =), \tau^{\_ \times A^\ast}, \tau^{\hoaremonad}, \alpha^d, q)$ on $\dCPO$ as follows:
  \begin{itemize}
    \item $\hoaremonad$ is the lower monad, and $\_\times A^{\ast}$ is the strong writer monad induced by the monoid structure on $A^\ast$.
    \item $\tau^{\_ \times A^{\ast}}\colon A^{\ast}\times A^{\ast}\rightarrow A^{\ast}$ is the multiplication of $\_\times A^{\ast}$, and $\tau^{\hoaremonad}\colon \hoaremonad(\boolsets)\rightarrow\boolsets$ is the disjunction $\tau^{\hoaremonad}(S)\defeq \bigvee S$, where the order on $\boolsets$ is the partial order $\preceq$ such that $\bot \prec \top$.
    \item $q$ is an inference query $q\colon A^{\ast}\rightarrow \boolsets^{Y}$ defined by $q(w)(y) \defeq \top$ if $w$ is accepting in $d$ from $y$, and  $q(w)(y) \defeq \bot$ otherwise.
  \end{itemize}
See~\cref{subsection:omittedProofsEmptinessCheck} for the details. We state the emptiness checking  as follows: 
\vspace{0.1cm}
\begin{mdframed}
{\bf Problem: Emptiness Checking of Accepting Traces.}
 \quad
\noindent 
Let $\Gamma\vdash M\colon \bt$ be a $\hoaremonad(\_\times A^{\ast})$-effectful well-typed term  such that all types in $\Gamma$ and $\bt$ are ground types.
Given $x\in \itp{\mathcal{A}}{\Gamma}{\hoaremonad(\_\times A^{\ast})}$ and an initial state $y\in Y$, compute:
\begin{align*}
  \Big(q^{\hoaremonad}\circ \big(\wpcond{M}{\tau^{\hoaremonad(\_\times A^{\ast})}}{Q}{\mathcal{A}}\big)\Big)(x)(y)\in \boolsets,
\end{align*}
where  $Q\colon \itp{\mathcal{A}}{\bt}{\hoaremonad(\_\times A^{\ast})}\rightarrow \hoaremonad(A^{\ast})$ is the constant function assigning the singleton $\{\epsilon\}$. Note that the weakest pre-condition $\big(\wpcond{M}{\tau^{\hoaremonad(\_\times A^{\ast})}}{Q}{\mathcal{A}}\big)(x)\in \hoaremonad(A^{\ast})$ computes accepting traces,
and $q^{\hoaremonad}\colon\hoaremonad(A^{\ast})\rightarrow \boolsets^{Y}$ checks whether there is an accepting trace.
\end{mdframed}
\begin{example}
  Consider the following non-deterministic program. It continues either writing a file and closing it (by raising $\writech\cdot \closech$), or just writing a file (by raising $\writech$) until it terminates (by raising $\terminatech$).   
  These behaviours are non-deterministic, induced by the generic effect $\probbranch{}{}{}$. 
\begin{equation}
  \label{eq:exEmptinessProg}
  \letrec{\nefunc}{x}{\probbranch{\big(\probbranch{(\nefunc\ ())^{\writech\cdot \closech}}{}{(\nefunc\ ())^{\writech}}\big)}{}{()^{\terminatech}}}{\nefunc\ ()}
\end{equation}
Now we impose the specification to detect a \emph{bad} behaviour, writing files continuously without closing the former one. 
This can be easily described by a DFA, and the emptiness checking asks that whether there is a \emph{bad} finite trace, that is,  a finite trace that is accepted by the DFA, which is indeed the case for the example (e.g. $\writech\cdot \writech\cdot \terminatech$ is a bad trace).   
\end{example}

Similarly, we obtain the following corollary. 
\begin{corollary}
  Given $\Gamma\vdash M\colon \bt$ such that all types in $\Gamma$ and $\bt$ are ground types,  the temporal verification problem (LHS) is reduced to the computation of a weakest pre-condition (RHS) as follows:
  \begin{align*}
   \Big( q^{\hoaremonad}\circ \big(\wpcond{M}{\tau^{\hoaremonad(\_\times A^{\ast})}}{Q}{\mathcal{A}}\big)\Big)^{\dagger} = \wpcond{\uncurry{M}}{\tau^{\hoaremonad}}{(q^{\hoaremonad}\circ Q)^{\dagger}}{\uncurry{\beta^d(\mathcal{A})}},
  \end{align*}
  where $Q\colon \itp{\mathcal{A}}{\bt}{\hoaremonad(\_\times A^{\ast})}\rightarrow \hoaremonad(A^{\ast})$ is the constant function assigning the singleton $\{\epsilon\}$, 
  and $\beta^d$ is the strong monad morphism constructed by the product situation.
  \qed
\end{corollary}

\subsection{Maximising Rewards of Accepting Traces}
\label{subsec:maximumRewardsAcceptingTraces}
Lastly, we present a temporal verification problem with a different specification, called \emph{reward machine}. 
Reward machine has been used for assigning rewards on traces directly, which is particularly useful in a black-box setting (assuming we can only observe traces, not states) that is common in reinforcement learning; see e.g.~\cite{IcarteKVM22,KopruluT23,AlurBBJ22}.
For this, we introduce \emph{coalgebraic reward machines} and their semantics.

\begin{definition}[coalgebraic reward machine]
  A \emph{coalgebraic reward machine} $d$ is a function $d\colon U\rightarrow (U\times \boolsets\times\nonnegreal )^A$, where
  $U$ is a finite set of \emph{states}, $A$ is a finite set of \emph{characters}, $\boolsets$ is the set of Booleans $\{\top, \bot\}$, and $\nonnegreal$ is the set of non-negative real numbers.
\end{definition}
We often regard $d$ as a morphism in $\dCPO$ with the discrete ordering. 

\begin{definition}[semantics of coalgebraic reward machine]
  Given a coalgebraic reward machine $d\colon U\rightarrow (U\times \boolsets\times \nonnegreal)^A$,
  we recursively define a \emph{coalgebra} $\widehat{d}\colon U\times \boolsets\times \nonnegreal\rightarrow (U\times \boolsets\times \nonnegreal)^{A^{\ast}}$ by the run of $d$, that is, for each $w\in A^{\ast}$, $u\in U$, $b\in \boolsets$, and $m\in \nonnegreal$,
  \begin{align*}
    \widehat{d}(u,b, m)(w) \defeq \begin{cases*}
      (u, b, m) &\text{ if  $w = \epsilon$, }\\
      \widehat{d}\Big(\pi_1\big(d(u)(a)\big), \pi_2\big(d(u)(a)\big)+m\Big)(w') &\text{ if $a\in A$ and  $w = a\cdot w'$,}
    \end{cases*}
  \end{align*}
  where $\pi_1$ and $\pi_2$ are the projections such that  $\pi_1\colon U\times \boolsets\times \nonnegreal\rightarrow U\times \boolsets$ and $\pi_2\colon U\times \boolsets\times \nonnegreal\rightarrow\nonnegreal$.
  We call $\widehat{d}$ \emph{semantics} of $d$.   Given $w\in A^{\ast}$ and $(u, b, m)\in U\times \boolsets$, we say that $w$ gains a reward $m'$ in $d$ from $(u, b, m)$ if $\pi_2\big(\widehat{d}(u,b,m)(w)\big) = \top$ and $\pi_3\big(\widehat{d}(u,b,m)(w)\big) = m'$.
\end{definition}

Similar to coalgebraic DFAs, the semantics induces a strong monad morphism. 
\begin{definition}[strong monad morphism induced by $d$]
  Given a coalgebraic reward machine $d\colon U\rightarrow (U\times \boolsets\times \nonnegreal)^A$,
  we define a \emph{strong monad morphism} $\alpha^d\colon \_\times A^{\ast}\Rightarrow (\_\times (U\times \boolsets\times \nonnegreal))^{U\times \boolsets\times \nonnegreal}$ as
  \begin{align*}
    \alpha^d_X(x, w)(u, b, m)\defeq \big(x, \widehat{d}(u, b, m)(w)\big), \text{ for each $X$.}
  \end{align*}
  Note that the order in $U\times \boolsets\times \nonnegreal$ is the discrete order. 
\end{definition}

With this strong monad morphism $\alpha^d$, we employ the product situation  $(\hoaremonad, (A^{\ast}, =), (U\times \boolsets\times \nonnegreal, =), \tau^{\_ \times A^\ast}, \tau^{\hoaremonad}, \alpha^d, q)$ on $\dCPO$ where:
  \begin{itemize}
    \item $\hoaremonad$ is the lower monad and $\_\times A^{\ast}$ is the strong monad induced by the monoid structure.
    \item $\tau^{\_ \times A^{\ast}}\colon A^{\ast}\times A^{\ast}\rightarrow A^{\ast}$ is the multiplication of $\_\times A^{\ast}$, and $\tau^{\hoaremonad}\colon \hoaremonad([0, \infty])\rightarrow[0, \infty]$ is the supremum $\tau^{\hoaremonad}(S)\defeq \bigvee S \ (= \sup S)$, where the order in $[0, \infty]$ is the standard total order.
    \item $q$ is the inference query $q\colon A^{\ast}\rightarrow {[0, \infty]}^{U\times \boolsets\times  \nonnegreal}$ defined by $q(w)(u, b, m) \defeq m'$ if $b'=\top$, and  $q(w)(u, b, m) \defeq 0$ otherwise, where $(u', b', m')\defeq \widehat{d}(u, b, m)(w)$.
  \end{itemize}
See~\cref{subsection:omittedProofsRewardMachine} for details. We define an optimisation problem over accepting traces as follows.
\vspace{0.1cm}
\begin{mdframed}
{\bf Problem: Maximising Rewards of Accepting Traces.}
 \quad
\noindent 
Let $\Gamma\vdash M\colon \bt$ be a $\hoaremonad(\_\times A^{\ast})$-effectful well-typed term $\Gamma\vdash M\colon \bt$ such that all types in $\Gamma$ and $\bt$ are ground types.
Given $x\in \itp{\mathcal{A}}{\Gamma}{\hoaremonad(\_\times A^{\ast})}$ and an initial state $(u, b)\in U\times \boolsets$, compute:
\begin{align*}
  \Big(q^{\hoaremonad}\circ \big(\wpcond{M}{\tau^{\hoaremonad(\_\times A^{\ast})}}{Q}{\mathcal{A}}\big)\Big)(x)(u, b, 0)\in [0, \infty],
\end{align*}
where  $Q\colon \itp{\mathcal{A}}{\bt}{\hoaremonad(\_\times A^{\ast})}\rightarrow \hoaremonad(A^{\ast})$ is the constant function assigning the singleton $\{\epsilon\}$. 
Note that the weakest pre-condition $\big(\wpcond{M}{\tau^{\hoaremonad(\_\times A^{\ast})}}{Q}{\mathcal{A}}\big)(x)\in \hoaremonad(A^{\ast})$ computes accepting traces,
and $q^{\hoaremonad}\colon\hoaremonad(A^{\ast})\rightarrow [0, \infty]^{U\times \boolsets\times \nonnegreal}$ computes the optimal reward of accepting traces.

\end{mdframed}
\vspace{0.1cm}

Again, we derive the following corollary that ensures the correctness of our reduction:

\begin{corollary}
  Given $\Gamma\vdash M\colon \bt$ such that all types in $\Gamma$ and $\bt$ are ground types, the temporal verification problem (LHS) is reduced to the computation of a weakest pre-condition (RHS) as follows:
  \begin{align*}
    \Big( q^{\hoaremonad}\circ \big(\wpcond{M}{\tau^{\hoaremonad(\_\times A^{\ast})}}{Q}{\mathcal{A}}\big)\Big)^{\dagger} = \wpcond{\uncurry{M}}{\tau^{\hoaremonad}}{(q^{\hoaremonad}\circ Q)^{\dagger}}{\uncurry{\beta^d(\mathcal{A})}},
  \end{align*}
  where $Q\colon \itp{\mathcal{A}}{\bt}{\hoaremonad(\_\times A^{\ast})}\rightarrow \hoaremonad(A^{\ast})$ is the constant function assigning the singleton $\{\epsilon\}$,
  and $\beta^d$ is the strong monad morphism constructed by the product situation.
  \qed
\end{corollary}

\section{Implementation}
\label{sec:experiments}
\begin{table}
	\caption{Experimental results.}
	\label{tab:experiment}
	\small
	\begin{tabular}{clc}
		Problem & Benchmark & Time (sec) \\
		\hline
		\multirow{2}{*}{Probability of Accepting Traces} & \texttt{coin\_flip} (the example in~\cref{sec:overview}) & 0.382 \\
		& \texttt{ho\_rw} (\cref{ex:hoRW}) & 26.375 \\
		\hline
		\multirow{2}{*}{Partial Expected Rewards of Accepting Traces} & \texttt{gr} (\cref{ex:rewrdFO}) & 2.322 \\
		& \texttt{ho\_gr} (\cref{ex:horeward}) & 7.508 \\
		\hline
	\end{tabular}
\end{table}
We instantiate our frameworks for probabilistic programs, presented in~\cref{subsec:probAcceptingTrace,subsec:expectedRewardAcceptingTrace}, with the existing solver~\cite{KuraUnno2024}. 
Given a probabilistic program and a DFA, we manually construct the effectful SPS transformed program, and provide its CPS transformed program as an input for the existing solver~\cite{KuraUnno2024}.
Automating these procedures is an orthogonal issue and remains a future work. 
We expect that we can extend and apply existing techniques, such as, efficient handling with administrative redexes or selective CPS transformation; see also~\cite{SatoUK13,SekiyamaU24}. 
The existing solver~\cite{KuraUnno2024} supports the  over-approximations (or upper-bounds) 
 of the reachability probabilities or partial expected rewards of product terms.
Our framework thus supports the computation of over-approximations for temporal verifications that we showed in~\cref{subsec:probAcceptingTrace,subsec:expectedRewardAcceptingTrace}. 
The existing solver~\cite{KuraUnno2024} reduces the problem to a type inference problem of their dependent refinement type system~\cite{KuraUnno2024}. 
The type inference algorithm in~\cite{KuraUnno2024} is automated by reducing the problem to the satisfiability problem of extended \emph{CHC constraints}~\cite{BjornerGMR15} with admissible predicate variables, for which they employ CEGIS-based CHC solving algorithm~\cite{UnnoTK21}.

We conduct a preliminary experiment to show the effectiveness of our approach. 
To the best of our knowledge, no existing tools provide comparable support for our target problems.
We check four instances, the running example in~\cref{sec:overview},~\cref{ex:hoRW},~\cref{ex:rewrdFO}, and an additional benchmark for partial expected rewards (see~\cref{sec:omittedBenchmarks}). 
\cref{tab:experiment} shows the overall result, which is obtained with 12th Gen Intel(R) Core(TM) i7-1270P 2.20 GHz with 32 GB of memory. 
For instance, in the first benchmark \texttt{coin\_flip}, we successfully verify the exact probability $1/7$ as an over-approximation in a reasonable time. 
In the third benchmark \texttt{gr}, we show an over-approximation $0.26$, where the exact partial expected reward is $0.24$.

\section{Related Work}
\label{sec:relatedWork}

\subsection{Categorical Approach to Product Constructions}

Categorical parallel compositions or product constructions have been actively studied, e.g.~\cite{HasuoJS08,TuriP97, GoncharovMSTU23,DAngeloGK0NRW24,Kori0RK24,CirsteaK23,WatanabeJRH25}.
In particular, our approach is inspired by the very recent work by Watanabe et al.~\cite{WatanabeJRH25}, where they provide a simple sufficient condition to ensure the correctness of product constructions for temporal verification of transition systems. 
One of our main novelty is that we define the product constructions \emph{symbolically} and \emph{compositionally}, in other words, we define product constructions on the level of \emph{programs}, not on the level of operational semantics (or transition systems).  
This leads our novel reduction to the computation of weakest pre-condition of (higher-order) programs.
To the best of our knowledge, this is the first work that provides such a generic reduction from temporal verification to computation of weakest pre-conditions for effectful higher-order programs, in which we can apply the existing off-the-shelf tool~\cite{KuraUnno2024} for probabilistic higher-order programs. 
We expect that we can apply the existing tools~\cite{YamadaKSS25,MaillardAAMHRT19,UnnoTK21} for angelic non-deterministic higher-order programs; this remains a future work.

\subsection{Categorical Frameworks for Weakest Pre-Conditions}

There is a line of research on categorical weakest pre-conditions for effectful programs~\cite{GoncharovS13,Hasuo15,HinoKH016,Rauch0S16,Jacobs17}. 
We employ the fibrational weakest pre-condition for effectful programs proposed by Aguirre et al.~\cite{AguirreKK22}.
They define Hoare triples in their fibrational account of weakest pre-conditions, and demonstrate its generality with a couple of examples, including weakest pre-expectation~\cite{McIverM05}, higher moment transformer~\cite{KuraUH19}, and expected runtime transformer~\cite{KaminskiKMO18}.
Recently, Kura and Unno~\cite{KuraUnno2024} develop a fully automated algorithm to compute weakest pre-conditions for probabilistic higher-order programs based on the fibrational weakest pre-condition~\cite{AguirreKK22}; see also~\cite{Kura2023}. 
Our approach also builds on their fibrational treatment of weakest pre-conditions, reducing temporal verification problems to the computation of weakest pre-conditions.
\subsection{Temporal Verification of Effectful Higher-Order Programs}

One promising methodology for temporal verification of effectful higher-order
programs is \emph{higher-order model checking} (HOMC)~\cite{Ong06,Kobayashi13},
which is an extension of model checking~\cite{ClarkeE81,Clarke18,BaierK08} to
higher-order programs. HOMC accommodates both safety and liveness problems with
\emph{branching-time} temporal verification (i.e., it can reason about trees
consisting of observations of interest) and leads to automated temporal
verifiers of higher-order programs~\cite{MuraseT0SU16}.
HOMC has also been extended to support a variety of specific effects, such as
non-determinism~\cite{Kobayashi09}, probability~\cite{KobayashiLG20},
exception~\cite{SatoUK13}, and control effects~\cite{LagoG24,SekiyamaU24}.
\citet{LagoG24} generalize HOMC to generic effects and show that it is possible
to reduce an HOMC problem with generic effects to one without them, via
\emph{continuation-passing style} (CPS) transformation.
By following \citeANP{LagoG24}'s recipe, \citet{SekiyamaU24} implemented an HOMC
tool for programs with generic effects (and control effects).
Their tool checks whether any trace generated by effectful higher-order
programs is accepting, but it does not support probabilistic temporal verification (e.g,
how likely accepting traces are generated).
Thus, it is left open whether \citeANP{LagoG24}'s recipe for HOMC can lead to
implementations of probabilistic temporal verifiers.
In contrast, our general framework, which is based on product constructions, is
proven useful to implement automated algorithms for probabilistic temporal
verification.

\emph{Dijkstra monads}~\cite{SwamyWSCL13} have also been shown to be an
effective approach to temporal verification of effectful higher-order programs.
They can lead to a semi-automated verification method, while HOMC, as well as
our implemented verifiers, are fully automated.
The verification method based on Dijkstra monads can support monadic effects,
including demonic and angelic nondeterminism and interactive
IO~\cite{AhmanHMMPPRS17,MaillardAAMHRT19}.
However, whether they can be applied to probabilistic verification remains
open~\cite{MaillardAAMHRT19}.

Another well-studied approach is \emph{type-and-effect systems} (or effect
systems for
short)~\cite{SkalkaS04,KoskinenT14,Gordon17,Nanjo0KT18,Gordon20,SekiyamaU23,ZhouYDJ24,SekiyamaU25}.
\citet{Gordon17,Gordon21} presents a theoretical general framework for effect
systems tailored to linear-time temporal safety verification (Gordon calls such
effect systems \emph{sequential}). However, it has not been explored whether the
framework allows probabilistic temporal verification and whether it can lead to
automated verification algorithms.

\section{Conclusion}
We propose a denotational product construction for linear-time temporal verification of effectful higher-order programs with recursion. 
We instantiate our framework with probabilistic and angelic non-deterministic programs, and we provide an automated verifier for probabilistic programs on top of~\cite{KuraUnno2024}. 

An important direction for future work is linear-time temporal verification for \emph{non-terminating traces}, which lies outside the scope of the present study.
This is particularly challenging for probabilistic programs, as it requires measure-theoretic foundations to define probability measures over infinite traces (see, e.g.,~\cite{BaierK08}).
It would be exciting to extend our framework to handle $\omega$-regular properties, such as liveness.
As an orthogonal direction, we aim to incorporate continuous distributions as computational effects.
This may be achievable using well-studied strong monads for continuous distributions (see, e.g.,~\cite{VakarKS19,GoubaultLarrecqJT23}).
\label{sec:conclusion}

\bibliographystyle{ACM-Reference-Format}
\bibliography{mybib}
\newpage

\appendix

\section{Omitted Definitions in~\cref{sec:overview}}
\label{sec:omittedDefOverview}
\begin{definition}[strong monad $T(\_\times Y)^Y$]
  \label{def:effectStateMonad}
  Given a strong monad $T$, we define the strong monad $(T(\_\times Y)^Y, \eta, \mu, \lstrength{}{})$ as follows:
  \begin{align*}
    \eta^{}_X &\defeq (\eta^T_{X \times Y})^\dagger, \\
    \mu^{}_X &\defeq \big(\mu^T_{X\times Y}\circ T(\ev{T(X\times Y)}{Y})\big)^Y, \\
    \lstrength{}{X,Z} &\defeq \big(\lstrength{T}{X, Z\times Y} \circ X \times \ev{T(Z\times Y)}{Y} \big)^\dagger.
  \end{align*}

\end{definition}

\section{Omitted Definitions in~\cref{sec:preliminaries} }
\label{sec:appSourceProgram}
\begin{figure}[t]
  \begin{mathpar}
  \inferrule* [Left=$\varintro$]{
    (x, \bt)\in \Gamma
  }{
    \Gamma \vdash x \colon \bt
  }
  \and
  \inferrule* [Left=$\exchange$]{
    \Gamma, x_2\colon \bt_2, x_1\colon \bt_1, \Delta \vdash M\colon \bt
  }{
    \Gamma, x_1\colon \bt_1, x_2\colon \bt_2, \Delta \vdash M\colon \bt
  }
  \and
  \inferrule* [Left=$\conintro$]{
		\Gamma \vdash M \colon \arfunc(c)
	}{
		\Gamma \vdash c\ M \colon \carfunc(c)
	}
	\and
	\inferrule* [Left=$\genintro$]{
		\Gamma \vdash M \colon \arfunc(e)
	}{
		\Gamma \vdash e\ M \colon \carfunc(e)
	}
  \\
  \inferrule*[Left=$\unitprod$]{
	}{
		\Gamma\vdash () \colon \mathbf{1}
	}\and
  \inferrule*[Left=$\prodintro$]{
    \Gamma \vdash M \colon \bt_1\\
    \Gamma \vdash N \colon \bt_2
	}{
		\Gamma\vdash (M, N) \colon \bt_1\times \bt_2
	}
  \and
  \inferrule*[Left=$\projintroone$]{
    \Gamma \vdash M\colon \bt_1\times \bt_2
	}{
    \Gamma \vdash \pi_1\ M\colon \bt_1
	}
  \and
  \inferrule*[Left=$\projintrotwo$]{
    \Gamma \vdash M\colon \bt_1\times \bt_2
	}{
    \Gamma \vdash \pi_2\ M\colon \bt_2
	}
  \\
  \inferrule*[Left=$\unitcoprod$]{
    \Gamma \vdash M\colon \mathbf{0}
	}{
    \Gamma \vdash \delta(M)\colon \bt
	}
  \and
  \inferrule*[Left=$\coprodone$]{
    \Gamma \vdash M\colon \bt_1
	}{
    \Gamma \vdash \iota_1\ M\colon \bt_1 +  \bt_2
	}
  \and
  \inferrule*[Left=$\coprodtwo$]{
    \Gamma \vdash M\colon \bt_2
	}{
    \Gamma \vdash \iota_2\ M\colon \bt_1 +  \bt_2
	}
  \and
  \inferrule*[Left=$\coprodelm$]{
    \Gamma \vdash M\colon \bt_1 +  \bt_2 \\  \Gamma, x_1\colon \bt_1 \vdash M_1\colon \bt \\  \Gamma, x_2\colon \bt_2 \vdash M_2\colon \bt
	}{
    \Gamma \vdash \delta(M, x_1\colon \bt_1.\ M_1, x_2\colon \bt_2.\ M_2)\colon \bt
	}\\
  \inferrule*[Left=$\abst$]{
    \Gamma, x\colon \bt_1 \vdash M\colon \bt_2
	}{
    \Gamma \vdash \lambda x.\ M\colon \bt_1\rightarrow \bt_2
	}
  \and
  \inferrule*[Left=$\appl$]{
    \Gamma \vdash M\colon \bt_1\rightarrow \bt_2
    \\ \Gamma \vdash N\colon \bt_1
	}{
    \Gamma \vdash M\ N\colon \bt_2
	}
\end{mathpar}
\caption{Typing rules of the source program without recursion.}
\label{fig:typingRulesSource}
\end{figure}

\begin{definition}[stable~\cite{FioreS99}]
  \label{def:stableCoproduct}
  A binary coproduct $X_1 \xrightarrow{\iota_1} X \xleftarrow{\iota_2} X_2$ is \emph{stable}   if for any morphism $f\colon Z\rightarrow X$, 
  there are pullbacks $\iota^{\ast}_1 Z$ and  $\iota^{\ast}_2 Z$ of $Z$ along $\iota_{i_1}$ and $\iota_{i_2}$, respectively, and $\iota^{\ast}_1(Z) \xrightarrow{} Z \xleftarrow{} \iota^{\ast}_2(Z)$ is a binary coproduct. 
  A bi-CCC is \emph{stable} if it has stable binary coproducts. 
\end{definition}

We present the typing rules of the source program without recursion in~\cref{fig:typingRulesSource}.

\section{Omitted Definitions and Proofs in~\cref{sec:prodLambdaCalc} }
\label{sec:omitDefProofProd}

\begin{proposition} \label{ap:st_monad_morphism_bij}
  Let $Y \in \mathbb{C}$ be an object and 
  consider a strong writer monad $\_ \times Z\colon \mathbb{C} \to \mathbb{C}$
  induced by
 a monoid object $(Z, e, \cdot)$.
  Then there is a bijective correspondence between 
    a strong monad morphism $\alpha\colon \_\times Z \Rightarrow (\_\times Y)^Y$
    and
    a morphism $f\colon Z \times Y \to Y$ satisfying that (i) $f \circ e \times Y$ is the canonical isomorphism $1 \times Y \cong Y$ and (ii) $f \circ (\cdot) \times Y = \ev{Y}{Y} \circ f^\dagger \times f$.
\end{proposition}
\begin{proof}
  We write $t'\colon \_ \times Y^Y \Rightarrow (\_ \times Y)^Y$ for the natural transformation given by $(\id \times \ev{Y}{Y})^\dagger$.
  For a strong monad morphism $\alpha\colon \_\times Z \Rightarrow (\_\times Y)^Y$,
  we define a morphism $f_\alpha\colon Z \times Y \to Y$ to be the transpose of $Z \cong 1 \times Z \xrightarrow{\alpha_1} (1 \times Y)^Y \cong Y^Y$.
  This $f_\alpha$ satisfies conditions (i) and (ii)
  because 
  \begin{itemize}
    \item 
  (i): 
  \begin{align*}
  &f_\alpha \circ e \times Y  \\
  &= \big(1 \times Y \xrightarrow{e \times Y} Z \times Y \cong 1 \times Z \times Y \xrightarrow{\alpha_1 \times Y} (1 \times Y)^Y \times Y \cong Y^Y \times Y \xrightarrow{\ev{Y}{Y}} Y\big) \\
  &= \big(1 \times Y \cong 1 \times 1 \times Y \xrightarrow{1 \times e \times Y} 1 \times Z \times Y  \xrightarrow{\alpha_1^\dagger} 1 \times Y \cong Y\big) \\
  &= \big(1 \times Y  \cong Y\big) \text{ since }\eta^{(\_ \times Y)^Y} = \alpha_1 \circ \eta^{\_ \times Z}_1.
  \end{align*}
  \item (ii): 
  Since $\alpha$ is a strong monad morphism,
  the following equalities hold.
  \begin{align*}
    &\big(Y^Y \times Z \cong (1 \times Y)^Y \times Z \xrightarrow{\alpha_{(1 \times Y)^Y}} ((1 \times Y)^Y \times Y)^Y \big) \\
    &=\big(Y^Y \times Z \xrightarrow{\alpha_{Y^Y}} (Y^Y \times Y)^Y \cong ((1 \times Y)^Y \times Y)^Y \big) \\
    &=\big(Y^Y \times Z \cong Y^Y \times 1 \times Z \xrightarrow{\alpha_{(Y^Y \times 1)}} (Y^Y \times 1 \times Y)^Y \cong ((1 \times Y)^Y \times Y)^Y \big) \\
    &=\big(Y^Y \times Z \cong Y^Y \times (1 \times Z) \xrightarrow{\id \times \alpha_{1}} Y^Y \times (1 \times Y)^Y \xrightarrow{\lstrength{(\_ \times Y)^Y}{}} (Y^Y \times 1 \times Y)^Y \cong ((1 \times Y)^Y \times Y)^Y \big). \\
  \end{align*}
  The equations above imply $\big(Y^Y \times Z \times Y \cong (1 \times Y)^Y \times Z \times Y  \xrightarrow{\alpha_{(1 \times Y)^Y}^\dagger} (1 \times Y)^Y \times Y \cong Y^Y \times Y\big) = \big(Y^Y \times Z \times Y \cong Y^Y \times (1 \times Z) \times Y   \xrightarrow{\id \times \alpha_1^\dagger} Y^Y \times 1 \times Y \cong Y^Y \times Y\big)$.
  Therefore, it follows that
  \begin{align*}
    &f_\alpha \circ (\cdot) \times Y  \\
    &= \big(Z \times Z \times Y \cong 1 \times Z \times Z \times Y \xrightarrow{1 \times (\cdot) \times Y} 1 \times Z \times Y \xrightarrow{1 \times f_\alpha} 1 \times Y \cong Y\big) \\
    &= \big(Z \times Z \times Y \cong 1 \times Z \times Z \times Y \xrightarrow{(\alpha_1 \circ \mu^{\_ \times Z}_1)^\dagger} 1 \times Y \cong Y\big) \\
    &= \big(Z \times Z \times Y \cong 1 \times Z \times Z \times Y \xrightarrow{(\mu^{(\_ \times Y)^Y} \circ \alpha_{(1 \times Y)^Y} \circ \alpha_1 \times Z)^\dagger} 1 \times Y \cong Y\big) \\
    &= \big(Z \times Z \times Y \cong 1 \times Z \times Z \times Y \xrightarrow{\alpha_1 \times Z \times Y} (1 \times Y)^Y \times Z \times Y \xrightarrow{\alpha_{(1 \times Y)^Y}^\dagger} (1 \times Y)^Y \times Y \xrightarrow{\ev{1 \times Y}{Y}} 1 \times Y\cong Y\big) \\
    &= \big(Z \times Z \times Y \xrightarrow{f_\alpha^\dagger \times \id} Y^Y \times Z \times Y \cong (1 \times Y)^Y \times Z \times Y  \xrightarrow{\alpha_{(1 \times Y)^Y}^\dagger} (1 \times Y)^Y \times Y \cong Y^Y \times Y\xrightarrow{\ev{Y}{Y}} Y\big) \\
    &= \big(Z \times Z \times Y \xrightarrow{f_\alpha^\dagger \times \id} Y^Y \times Z \times Y \cong Y^Y \times (1 \times Z) \times Y   \xrightarrow{\id \times \alpha_1^\dagger} Y^Y \times 1 \times Y \cong Y^Y \times Y\xrightarrow{\ev{Y}{Y}} Y\big) \\
    &= \ev{Y}{Y} \circ f_\alpha^\dagger \times f_\alpha.
  \end{align*}
  \end{itemize}

  For a morphism $f\colon Z \times Y \to Y$ satisfying $f \circ e \times Y$ and $f \circ (\cdot) \times Y = \ev{Y}{Y} \circ f^\dagger \times f$,
  we define a strong monad morphism $\alpha_f\colon \_\times Z \Rightarrow (\_\times Y)^Y$ by
  $(\alpha_f)_X \defeq t'_X \circ (X \times f^\dagger)$.
  It is easy to see that $\alpha_f$ is a natural transformation.
  Moreover, it forms a strong monad morphism because for each $I, J \in \mathbb{C}$, 
  \begin{itemize}
    \item $(\alpha_f)_I \circ \eta^{\_ \times Z} = \eta^{(\_ \times Y)^Y}$ holds since the following diagram commutes.
    \begin{displaymath}
      \xymatrix@R=1.5em@C=1.5em{
        I \ar[r]^-\cong \ar[ddrr]^-{\eta} &I \times 1 \ar[dr]^{I \times \cong^\dagger} \ar[r]^-{I \times e} &I \times Z \ar[d]^{I \times f^\dagger} \\
        & &I \times Y^Y \ar[d]^{t'} \\
        & &(I \times Y)^Y.
      }
    \end{displaymath}
    The top-right triangle commutes by the condition (i). 
    \item 
    We write $g_{X}$ for the isomorphism $1 \times X \cong X$.
    The condition (ii) implies that
    the following diagram commutes:
    \begin{displaymath}
      \xymatrix@R=1.5em@C=1.5em{
        I \times Z \times Z \times Y \ar[dd]^-{I \times f^\dagger \times f^\dagger \times Y} \ar[r]^-\cong &I \times (Z \times Z) \times Y \ar[r]^-{I \times (\cdot) \times Y} \ar[d]^-\cong &I \times Z \times Y \ar[r]^-{I \times f^\dagger \times Y} \ar@/_1em/[rr]_-{I \times f} &I \times Y^Y \times Y \ar[r]^-{I \times \ev{Y}{Y}} &I \times Y \\
        &I \times Z \times (Z \times Y) \ar[d]^-{I \times f^\dagger \times f} \\
        I \times Y^Y \times Y^Y \times Y \ar[dr]^-{t' \times \id} \ar[r]^-{\id \times \ev{Y}{Y}} &I \times Y^Y \times Y \ar[rrr]^-{t' \times Y} \ar[uurrr]^-{I \times \ev{Y}{Y}} & & &(I \times Y)^Y \times Y \ar[uu]^{\ev{I \times Y}{Y}} \\
        &(I \times Y)^Y \times Y^Y \times Y \ar[urrr]_-{\id \times \ev{Y}{Y}}
      }
    \end{displaymath}
    By taking its transpose, we have $(\alpha_f)_I \circ \mu_I^{\_ \times Z} = \ev{I \times Y}{Y}^Y \circ (\alpha_f)_{(I \times Y)^Y} \circ T (\alpha_f)_I$.
    \item 
    Because
    \begin{align*}
    (t'_{I \times J} \circ \cong)^\dagger 
    &= I \times (J \times Y^Y) \times Y \cong I \times J \times Y^Y \times Y \xrightarrow{I \times J \times \ev{Y}{Y}} I \times J \times Y \\
    &= I \times (J \times Y^Y) \times Y  \cong I \times (J \times Y^Y \times Y) \xrightarrow{I \times t'^\dagger}  I \times (J \times Y) \cong I \times J \times Y \\
    &= I \times (J \times Y^Y) \times Y \xrightarrow{I \times t' \times Y} I \times (J \times Y)^Y \times Y \xrightarrow{I \times \ev{J \times Y}{Y}} I \times (J \times Y) \cong I \times J \times Y \\
    &= (\lstrength{(\_ \times Y)^Y}{} \circ I \times t'_{J})^\dagger,
    \end{align*}
    we have
    \begin{align*}
      &(\alpha_f)_{I \times J} \circ \lstrength{\_ \times Z}{} \\
      &= t' \circ \id \times f^\dagger \circ \lstrength{\_ \times Z}{} \\
      &= t' \circ \cong \circ \id \times f^\dagger \text{ where }\cong\colon I \times (J \times Y^Y) \to (I \times J) \times Y^Y \\
      &= \lstrength{(\_ \times Y)^Y}{} \circ I \times t' \circ \id \times f^\dagger \text{ since }t'_{I \times J} \circ \cong \ = \lstrength{(\_ \times Y)^Y}{} \circ I \times t'_{J}.
    \end{align*}
    Therefore, $(\alpha_f)_{I \times J} \circ \lstrength{\_ \times Z}{} = \lstrength{(\_ \times Y)^Y}{} \circ I \times (\alpha_f)_I$ holds.
  \end{itemize}

  These assignment $f_{(\_)}$ and $\alpha_{(\_)}$ gives a bijective correspondence.
\end{proof}

\begin{lemma}
  \label{lem:alphaInvariant}
  Let $(T, Z, Y, \tau^{\_ \times Z}, \tau^T, \alpha, q)$ be a product situation.
  The following diagram commutes:

  \adjustbox{scale=1,center}{
  \begin{tikzcd}
    U\times Z   \arrow[r, "\alpha"] \arrow[d, "U\times \alpha"]& (U\times Y)^Y\\
    U\times Y^Y \arrow[r, "\eta"]& (U\times Y^Y\times Y)^Y \arrow[u, "(\id\times\evsyb)^Y"]
  \end{tikzcd}
  }
\end{lemma}
\begin{proof}
  This follows from the fact that $\alpha = \lstrength{(\_ \times Y)^Y}{} \circ U \times \alpha$.
\end{proof}

\begin{lemma}
  \label{lem:alphaStrengthInvariant}
  Let $(T, Z, Y, \tau^{\_ \times Z}, \tau^T, \alpha, q)$ be a product situation.
  The following diagram commutes:

  \adjustbox{scale=1,center}{
  \begin{tikzcd}
    T(U)\times Z \arrow[d, "\rstrength{T}{}"] \arrow[r, "\alpha"] & \big(T(U)\times Y\big)^Y \arrow[r, "(\rstrength{T}{})^Y"] & T(U\times Y)^Y\\
    T(U\times Z) \arrow[r, "T(\alpha)"]& T\big((U\times Y)^Y\big) \arrow[ur, "u"]
  \end{tikzcd}
}
\end{lemma}
\begin{proof}
 We take the transpose, and show that the following diagram commutes by~\cref{lem:alphaInvariant}:

  \adjustbox{scale=0.8,center}{
    \begin{tikzcd}
      T(U)\times Z\times Y \arrow[rrr, "\alpha\times \id"] \arrow[rd, "\id\times \alpha\times \id"]\arrow[dd, "\rstrength{T}{}\times \id"] &&& \big(T(U) \times Y\big)^Y \times Y \arrow[ddl, "\evsyb"]\arrow[ddd, "(\rstrength{T}{})^{\id}\times \id"]\\
      & T(U)\times Y^Y\times Y \arrow[r, "\eta\times \id"] \arrow[d, "\rstrength{T}{}\times \id"] \arrow[rd, "\id\times \evsyb"] & \big(T(U)\times Y^Y\times Y\big)^Y\times Y\arrow[ru, "(\id\times \evsyb)^{\id}\times \id"]\\
      T(U\times Z)\times Y \arrow[dd, "T(\alpha)\times \id"]\arrow[r,swap, "T(\id\times \alpha)\times \id"] \arrow[rd,swap,"\rstrength{T}{}"]& T(U\times Y^Y)\times Y \arrow[rd,swap,"\rstrength{T}{}"]& T(U)\times Y \arrow[ddr, "\rstrength{T}{}"]\\
      & T(U\times Z\times Y) \arrow[r,swap, "T(\id\times \alpha\times \id)"] \arrow[d, "T(\alpha\times \id)"] & T(U\times Y^Y\times Y) \arrow[dl, ""] \arrow[rd, swap, "T(\id\times \evsyb)"] & T(U\times Y)^Y\times Y \arrow[d, "\evsyb"]\\
      T\big((U\times Y)^Y\big)\times Y \arrow[d, "\eta\times \id"]\arrow[r, "\rstrength{T}{}"]& T\big((U\times Y)^Y\times Y\big) \arrow[rr, "T(\evsyb)"]& & T(U\times Y)\\
      \Big( T\big((U\times Y)^Y\big)\times Y \Big)^Y\times Y \arrow[d, "(\rstrength{T}{})^{\id}\times \id"]\arrow[rr, "\evsyb"]&& T\big((U\times Y)^Y\big)\times Y \arrow[lu, "\rstrength{T}{}"] \arrow[d, "\rstrength{T}{}"] \\
      \Big( T\big((U\times Y)^Y\times Y\big) \Big)^Y\times Y \arrow[rrd, "T(\evsyb)^\id \times \id"]\arrow[rr, "\evsyb"]&& T\big((U\times Y)^Y\times Y\big) \arrow[ruu, swap,"T(\evsyb)"]\\
      && T(U\times Y)^Y\times Y \arrow[uuur, bend right=20, swap, "\evsyb"]
    \end{tikzcd}
  }
\end{proof}

\subsection{Proof of~\cref{prop:liftingCorrectness}}
We first prove that $\tau^{T(\_\times Z)}$ is an Eilenberg-Moore algebra.
We can see that $\tau^{T(\_\times Z)} \circ \eta^{T(\_\times Z)}_{T(\Omega^{\_\times Z})} = \id_{T(\Omega^{\_\times Z})}$ by the following diagram.

\adjustbox{scale=0.8,center}{
  \begin{tikzcd}
    T(\Omega^{\_\times Z}) \arrow[rrr, "T(\eta)"]\arrow[d, "\eta"]&&& T(\Omega^{\_\times Z}\times Z) \arrow[dd, "T(\tau)"] \\
    T(\Omega^{\_\times Z})\times Z \arrow[r, "\rstrength{}{}"]\arrow[d, "\eta"] & T(\Omega^{\_\times Z}\times Z) \arrow[r, "\id"] \arrow[d, "\eta"] & T(\Omega^{\_\times Z}\times Z) \arrow[rd, "T(\tau)"] \arrow[ru, "\id"]\\
    T\big(T(\Omega^{\_\times Z})\times Z \big) \arrow[r, "T(\rstrength{}{})"] & T^2(\Omega^{\_\times Z}\times Z) \arrow[ru, "\mu"] \arrow[r, "T^2(\tau)"]& T^2(\Omega^{\_\times Z}) \arrow[r, "\mu^T"] & T(\Omega^{\_\times Z})
  \end{tikzcd}
}

We can also see that $\tau^{T(\_\times Z)}\circ T(\tau^{T(\_\times Z)}\times Z) = \tau^{T(\_\times Z)}\circ \mu^{T(\_\times Z)}_{\Omega^{\_\times Z}}$ by the following two diagrams.

\adjustbox{scale=0.8,center}{
  \begin{tikzcd}
    T\Big( T\big(T(\Omega)\times Z\big)\times Z \Big) \arrow[rr, "T\big(T(\rstrength{}{})\times Z\big)"] \arrow[d, "T(\rstrength{}{})"] &&  T\big( T^2(\Omega\times Z)\times Z \big)\arrow[d, "T(\rstrength{}{})"]\\
    T^2\big(T(\Omega)\times Z\times Z\big) \arrow[rr, "T^2(\rstrength{}{}\times Z)"]\arrow[d, "\mu"]&&  T^2\big( T(\Omega\times Z)\times Z \big) \arrow[d, "\mu"]\\
    T\big(T(\Omega)\times Z\times Z\big) \arrow[rr, "T(\rstrength{}{}\times Z)"] \arrow[d, "T\big(\mu\big)"]&& T\big( T(\Omega\times Z)\times Z \big) \arrow[rrr, "T\big(T(\tau)\times Z\big)"] \arrow[d, "T(\rstrength{}{})"]&&& T\big(T(\Omega)\times Z\big)\arrow[d, "T(\rstrength{}{})"]\\
    T\big(T(\Omega)\times Z\big)      \arrow[d, "T(\rstrength{}{})"]    && T^2(\Omega\times Z\times Z) \arrow[dll, "T^2(\mu)"] \arrow[d, "\mu"] \arrow[rrr, "T^2\big(T(\tau)\times Z\big)"]&&& T^2(\Omega\times Z) \arrow[dd, "T^2(\tau)"]\\
    T^2(\Omega\times Z) \arrow[d, "T^2(\tau)"]  \arrow[drr, "\mu"]&& T(\Omega\times Z\times Z) \arrow[d, "T(\mu)"] \arrow[dr, "T(\tau\times Z)"]\\
    T^2(\Omega)  \arrow[d, "\mu"]&& T(\Omega\times Z) \arrow[lld, "T(\tau)"] & T(\Omega\times Z) \arrow[llld, bend left=5, "T(\tau)"] && T^2(\Omega) \arrow[llllld, bend left = 10, "\mu"]\\
    T(\Omega)
  \end{tikzcd}
}

and

\adjustbox{scale=0.8,center}{
  \begin{tikzcd}
    && T\big(T^2(\Omega\times Z)\times Z \big) \arrow[rr, "T(T^2(\tau)\times Z)"]&& T\big(T^2\Omega\times Z\big) \arrow[d, "T(\mu\times Z)"]\\
    T\Big( T\big(T(\Omega)\times Z\big)\times Z \Big) \arrow[d, "T(T(\rstrength{}{})\times Z)"]\arrow[rru, "T(T(\rstrength{}{})\times Z)"]&& T\big(T(\Omega\times Z)\times Z \big) \arrow[rr, "T(T(\tau)\times Z)"]&& T(T\Omega\times Z) \arrow[r, "T(\rstrength{}{})"]& T^2(\Omega\times Z)\arrow[ddddddd, "T^2(\tau)"]\\
    T\big( T^2(\Omega\times Z)\times Z \big) \arrow[d, "T(\rstrength{}{})"]\arrow[rr, "T(T^2(\tau)\times Z)"] \arrow[rru, "T(\mu\times Z)"]&& T(T^2\Omega\times Z)\arrow[rru, "T(\mu\times Z)"] \arrow[d, "T(\rstrength{}{})"]\\
    T^2\big( T(\Omega\times Z)\times Z \big)\arrow[d, "\mu"] \arrow[rr, "T^2(T(\tau)\times Z)"] && T^2(T\Omega\times Z)\arrow[d, "T^2(\rstrength{}{})"]\\
    T\big( T(\Omega\times Z)\times Z \big) \arrow[d, "T(T(\tau)\times Z)"]&& T^3(\Omega\times Z) \arrow[rrruuu, "T(\mu)"] \arrow[d, "T^3(\tau)"] \arrow[rrddd, bend left=15, "T^3(\tau)"]\\
    T\big( T\Omega\times Z \big) \arrow[d, "T(\rstrength{}{})"] && T^3(\Omega) \arrow[d,"\mu"]\\
    T^2( \Omega\times Z ) \arrow[d, "T^2(\tau)"]&& T^2(\Omega) \arrow[lldd, "\mu"]\\
    T^2( \Omega) \arrow[d, "\mu"] &&&& T^3(\Omega) \arrow[rd, "T(\mu)"]\\
    T( \Omega)   &&&&& T^2(\Omega)\arrow[lllll, "\mu"]\\
  \end{tikzcd}
}

Showing that $\tau^{T(\_\times Y)^Y}$ is an Eilenberg-Moore algebra is straightforward.

We then prove that $\beta\colon T(\_\times Z)\rightarrow T(\_\times Y)^Y$ is a strong monad morphism.
The condition on units is easy to show. We show that $\beta\circ \mu = \mu \circ T(\beta\times Y)^Y\circ \beta$.
By~\cref{lem:alphaStrengthInvariant}, we see that

\adjustbox{scale=1,center}{
  \begin{tikzcd}
    T\big((X\times Y)^Y\big)\times Z \arrow[rr, "\natu{T, Y}\times \id"]\arrow[d, "\rstrength{T}{}"] \arrow[rd, "\alpha"]& & T(X\times Y)^Y\times Z\arrow[dd, "\alpha"]\\
    T\big((X\times Y)^Y\times Z\big) \arrow[dd, "T(\alpha)"] & \Big(T\big((X\times Y)^Y\big)\times Y\Big)^Y \arrow[dr, "(\natu{T, Y}\times \id)^{\id}"] \arrow[d, "(\rstrength{T}{})^Y"] & \\
    & \Big(T\big((X\times Y)^Y\times Y\big)\Big)^Y \arrow[rd, "T(\evsyb)^Y"]& \Big(T(X\times Y)^Y\times Y\Big)^Y \arrow[d, "\evsyb^Y"]\\
    T\Big(\big((X\times Y)^Y\times Y\big)^Y\Big) \arrow[ru, "\natu{T, Y}"]\arrow[r, "T(\evsyb^Y)"] & T\big((X\times Y)^Y\big) \arrow[r, "\natu{T, Y}"] & T(X\times Y)^Y
  \end{tikzcd}
}

With this commuting diagram, we can see that

\adjustbox{scale=0.8,center}{
  \begin{tikzcd}
    T\big(T(X\times Z)\times Z\big) \arrow[rr, "T(T(\alpha)\times \id)"]\arrow[d, "T(\rstrength{T}{})"]&& T\Big(T\big((X\times Y)^Y\big)\times Z\Big)\arrow[d, "T(\rstrength{T}{})"]\arrow[rrrr, "T(\natu{T, Y}\times \id)"] && && T\big(T(X\times Y)^Y\times Z\big)\arrow[d, "T(\alpha)"]\\
    T^2(X\times Z\times Z)\arrow[d, "T^2(\mu)"]          \arrow[rr, "T^2(\alpha\times \id)"]&& T^2\big((X\times Y)^Y\times Z\big) \arrow[d, "T^2(\alpha)"]&&  T\big(T(X\times Y)^Y\big) \arrow[rrdd, bend right = 10,  "\natu{T, Y}"] && T\Big(\big(T(X\times Y)^Y\times Y \big)^Y\Big)\arrow[ll, "T((\evsyb)^Y )"]\arrow[d, "\natu{T, Y}"]\\
    T^2(X\times Z)\arrow[drr, "T^2(\alpha)"] \arrow[d, "\mu"]&& T^2\Big(\big((X\times Y)^Y\times Y\big)^Y\Big) \arrow[d, "T^2\big((\evsyb)^Y\big)"]&& && T\big(T(X\times Y)^Y\times Y \big)^Y\arrow[d, "T(\evsyb)^Y"]\\
    T(X\times Z) \arrow[rrd, "T(\alpha)"] && T^2\big((X\times Y)^Y \big) \arrow[d, "\mu"] \arrow[rruu, bend right = 20, swap, "T(\natu{T, Y})"] && && T^2(X\times Y)^Y\arrow[d, "\mu^Y"]\\
    && T\big((X\times Y)^Y \big) \arrow[rrrr, "\natu{T, Y}"]&& && T(X\times Y)^Y
  \end{tikzcd}
}

Lastly, we prove that $q^T$ is an inference query.
By~\cref{lem:alphaStrengthInvariant}, this is indeed shown by the following diagram:

\adjustbox{scale=0.8,center}{
  \begin{tikzcd}
    T\big(T(\Omega^{\_\times Z})\times Z\big) \arrow[r, "T(\rstrength{T}{})"] \arrow[d, "T(\alpha)"] & T^2(\Omega^{\_\times Z}\times Z) \arrow[r, "\mu"] \arrow[d, "T^2(\alpha)"]& T(\Omega^{\_\times Z}\times Z)\arrow[r, "T(\tau^{\_\times Z})"]\arrow[d, "T(\alpha)"] & T(\Omega^{\_\times Z})\arrow[d, "T(q)"] \\
    T\Big(\big(T(\Omega^{\_\times Z})\times Y\big)^Y\Big) \arrow[d, "\natu{T, Y}"] \arrow[dr, "T\big((\rstrength{T}{})^Y\big)"]& T^2\big((\Omega^{\_\times Z}\times Y)^Y\big) \arrow[d, "T(\natu{T, Y})"] \arrow[r, "\mu"] & T\big((\Omega^{\_\times Z}\times Y)^Y\big) \arrow[d, "\natu{T, Y}"] &T({\Omega^T}^{Y}) \arrow[ddddl, bend left = 20, "\natu{T, Y}"]\\
    T\big(T(\Omega^{\_\times Z})\times Y\big)^Y \arrow[d, "T(T(q)\times Y)^Y"] \arrow[rd, "T(\rstrength{T}{})^Y"]& T\Big(\big(T(\Omega^{\_\times Z}\times Y)\big)^Y\Big) \arrow[d, "\natu{T, Y}"]& T(\Omega^{\_\times Z}\times Y)^Y \arrow[d, "T(q\times \id)^{\id}"]\\
    T\big(T({\Omega^{T}}^Y)\times Y\big)^Y \arrow[d, "T(\natu{T, Y}\times \id)^{\id}"] \arrow[dr, "T(\rstrength{T}{})^{\id}"] & T^2(\Omega^{\_\times Z}\times Y)^Y \arrow[ur, "\mu^Y"] \arrow[d, "T^2(q\times \id)^{\id}"]& T\big({\Omega^{T}}^Y\times Y\big)^Y \arrow[dd, "T(\evsyb)^{\id}"]\\
    T\big(T({\Omega^{T}})^Y\times Y\big)^Y \arrow[dr, "T(\evsyb)^{\id}"]& T^2({\Omega^{T}}^Y\times Y)^Y \arrow[d, "T^2(\evsyb)^{\id}"] \arrow[ru, "\mu^{\id}"]\\
    & T^2(\Omega^T)^Y \arrow[r, "\mu^{\id}"]& T(\Omega^T)^Y \arrow[r, "(\tau^T)^Y"] & {\Omega^T}^Y
  \end{tikzcd}
}

We then show the condition on the strenghs by the following diagram.  

\adjustbox{scale=0.6,center}{
  \begin{tikzcd}
    X\times T(W\times Z) \arrow[d, "\lstrength{}{}"] \arrow[r, "\id\times T(\alpha)"]& X\times T\big((W\times Y)^Y\big) \arrow[d, "\lstrength{}{}"] \arrow[r, "\id\times \eta"] \arrow[dr, "\eta"]& X\times\Big(T\big((W\times Y)^Y\big)\times Y\Big)^Y\arrow[r, "\id\times (\rstrength{}{})^Y"] \arrow[d, "\lstrength{(\_ \times Y)^Y}{}"]&  X\times\Big(T\big((W\times Y)^Y\times Y\big)\Big)^Y \arrow[r, "\id\times T(\evsyb)^Y"] \arrow[d, ""]& X\times T(W\times Y)^Y\arrow[dd, "\eta"]\\
    T(X\times W\times Z) \arrow[d, "T(\alpha)"] \arrow[r, "T(\id\times \alpha)"] & T\big(X\times (W\times Y)^Y\big)\arrow[dl, "T(\lstrength{(\_ \times Y)^Y}{})"] \arrow[dr, "\eta"]& \Big(X\times T\big((W\times Y)^Y\big)\times Y\Big)^Y\arrow[d, "(\lstrength{}{}\times \id)^Y"] \arrow[r, "(\id\times \rstrength{}{})^Y"]&\Big( X\times T\big((W\times Y)^Y\times Y\big)\Big)^Y\arrow[ddl, bend left = 5, "(\lstrength{}{})^Y"] \arrow[ddd, "\big(\id\times T(\evsyb)\big)^Y"]\\
    T\big((X\times W\times Y)^Y\big) \arrow[d, "\eta"]&  & \Big(T\big(X\times (W\times Y)^Y\big)\times Y\Big)^Y \arrow[d, "(\rstrength{}{})^Y"] \arrow[dll, "\big(T(\lstrength{(\_ \times Y)^Y}{})\times Y\big)^Y"] & & \big(X\times T(W\times Y)^Y\times Y\big )^Y\arrow[ddl, "(\id\times \evsyb)^Y"]\\
    \Big(T\big((X\times W\times Y)^Y\big) \times Y\Big)^Y \arrow[d, "(\rstrength{}{})^Y"] & & T\big(X\times (W\times Y)^Y\times Y\big)^Y \arrow[dll, "T(\lstrength{(\_ \times Y)^Y}{}\times Y)^Y"] \arrow[ddll, bend left = 10, "T(\id\times \evsyb)^Y"] \\
    T\big((X\times W\times Y)^Y\times Y\big)^Y \arrow[d, "T(\evsyb)^Y"]& & &  \big(X\times T(W\times Y)\big)^Y\arrow[llld, bend left = 5, "(\lstrength{}{})^Y"]\\
    T(X\times W\times Y)^Y &
  \end{tikzcd}
}

\qed

\section{Omitted Definitions and Proofs in~\cref{sec:fibrationalapproach}}
 We sometimes write $\Gamma \times \bt$ instead of $(\prod_i \bt_i) \times \bt$ for abbreviation.

\begin{lemma}[weakening]
  \label{lem:weakening}
  For any well-typed term $\Gamma\vdash M\colon \bt_1$, 
  the term $\Gamma, x\colon \bt_2\vdash M\colon \bt_1$ is also well-typed.
  Moreover, the following equation holds:
  \begin{align*}
    \itp{A}{\Gamma, x\colon \bt_2\vdash M\colon \bt_1}{T} = \itp{A}{\Gamma\vdash M\colon \bt_1}{T}\circ \pi_1.
  \end{align*}
\end{lemma}
\begin{proof}
  By the straightforward induction.
\end{proof}

\begin{lemma}
  \label{lem:betaReduction}
  For any strong monad $T$ and well-typed terms $\Gamma\vdash \lambda x. M \colon \bt_1\rightarrow \bt_2$ and $\Gamma\vdash N \colon \bt_1$, the following equation holds:
  \begin{align*}
  \itp{A}{(\lambda x. M)(N)}{T} = \mu_{\itp{A}{\bt_2}{T}}\circ T(\itp{A}{M}{T})\circ \lstrength{T}{\itp{A}{\Gamma}{T}, \itp{A}{\bt_1}{T}} \circ \langle \id_{\itp{A}{\Gamma}{T}}, \itp{A}{N}{T}\rangle
  \end{align*}
\end{lemma}
\begin{proof}
  \begin{align*}
    \itp{A}{(\lambda x. M)(N)}{T} &= \mu_{\itp{A}{\bt_2}{T}} \circ T(\mu_{\itp{A}{\bt_2}{T}})\circ T^2(\ev{T(\itp{A}{\bt_2}{T})}{\itp{A}{\bt_1}{T}})\circ T\big(\lstrength{T}{T(\itp{A}{\bt_2}{T})^{\itp{A}{\bt_1}{T}},\itp{A}{\bt_1}{T}}\big)\\
                                  &\quad \circ \rstrength{T}{T(\itp{A}{\bt_2}{T})^{\itp{A}{\bt_1}{T}},T(\itp{A}{\bt_1}{T})}\circ \langle \itp{A}{\lambda x. M}{T}, \itp{A}{N}{T} \rangle\\
                                  &=  \mu_{\itp{A}{\bt_2}{T}}\circ \mu_{T(\itp{A}{\bt_2}{T})}\circ T^2(\ev{T(\itp{A}{\bt_2}{T})}{\itp{A}{\bt_1}{T}})\circ T\big(\lstrength{T}{T(\itp{A}{\bt_2}{T})^{\itp{A}{\bt_1}{T}},\itp{A}{\bt_1}{T}}\big)\\
    &\quad \circ \eta^T_{T(\itp{A}{\bt_2}{T})^{\itp{A}{\bt_1}{T}} \times T\itp{A}{\bt_1}{T}} \circ \langle \itp{A}{M}{T}^\dagger, \itp{A}{N}{T} \rangle\\
      &=  \mu_{\itp{A}{\bt_2}{T}}\circ T(\ev{T(\itp{A}{\bt_2}{T})}{\itp{A}{\bt_1}{T}})\circ \mu_{T(\itp{A}{\bt_2}{T})^{\itp{A}{\bt_1}{T}}\times \itp{A}{\bt_1}{T}}\circ T\big(\lstrength{T}{T(\itp{A}{\bt_2}{T})^{\itp{A}{\bt_1}{T}},\itp{A}{\bt_1}{T}}\big)\\
    &\quad \circ \eta^T_{T(\itp{A}{\bt_2}{T})^{\itp{A}{\bt_1}{T}} \times T\itp{A}{\bt_1}{T}}\circ \itp{A}{M}{T}^\dagger \times \id_{T(\itp{A}{\bt_1}{T})}\circ  \langle \id_{\itp{A}{\Gamma}{T}}, \itp{A}{N}{T} \rangle\\
    &=  \mu_{\itp{A}{\bt_2}{T}}\circ T(\ev{T(\itp{A}{\bt_2}{T})}{\itp{A}{\bt_1}{T}})\circ \lstrength{T}{T(\itp{A}{\bt_2}{T})^{\itp{A}{\bt_1}{T}},\itp{A}{\bt_1}{T}}
    \circ \itp{A}{M}{T}^\dagger \times \id_{T(\itp{A}{\bt_1}{T})}\circ  \langle \id_{\itp{A}{\Gamma}{T}}, \itp{A}{N}{T} \rangle\\
    &=  \mu_{\itp{A}{\bt_2}{T}}\circ T(\ev{T(\itp{A}{\bt_2}{T})}{\itp{A}{\bt_1}{T}})\circ T(\itp{A}{M}{T}^\dagger \times \id_{\itp{\mathbt{t}_1}{M}{T}}) \circ \lstrength{T}{\itp{A}{\Gamma}{T},\itp{A}{\bt_1}{T}}
    \circ  \langle \id_{\itp{A}{\Gamma}{T}}, \itp{A}{N}{T} \rangle\\
    &=  \mu_{\itp{A}{\bt_2}{T}}\circ T(\itp{A}{M}{T})\circ \lstrength{T}{\itp{A}{\Gamma}{T},\itp{A}{\bt_1}{T}}\circ  \langle \id_{\itp{A}{\Gamma}{T}}, \itp{A}{N}{T} \rangle.
  \end{align*}
\end{proof}

\subsection{Omitted Definitions of $\litp{A}{M}{T}{L}$} \label{ap:litp}

\begin{align*}
  \litp{A}{c \ M}{T}{L} &\coloneqq T(a(c)) \circ \litp{A}{M}{T}{L}, \\
  \litp{A}{e \ M}{T}{L} &\coloneqq \mu_{TL\itp{\mathcal{A}}{\carfunc(e)}{T}} \circ T(a(e)) \circ \litp{A}{M}{T}{L}, \\
  \litp{A}{()}{T}{L} &\coloneqq \eta \circ L! \text{ where }!\colon \litp{A}{\Gamma}{T}{L} \to 1 \text{ is the unique morphism}, \\
  \litp{A}{\pi_i M}{T}{L} &\coloneqq TL\pi_i \circ \litp{A}{M}{T}{L}, \\
  \litp{A}{\delta(M)}{T}{L} &\coloneqq TL! \circ \litp{A}{M}{T}{L}, \\
  \litp{A}{\iota_i M}{T}{L} &\coloneqq TL\iota_i \circ \litp{A}{M}{T}{L}, \\
  \litp{A}{\delta(M, x_1\colon \bt_1.~M_1, x_2\colon \bt_2.~M_2)}{T}{L} &\coloneqq \mu \circ T([\litp{A}{M_1}{T}{L}, \litp{A}{M_2}{T}{L}]) \circ T(\natcoprod{\litp{A}{\Gamma, \bt_1}{T}{L}, \litp{A}{\Gamma, \bt_2}{T}{L}}) \circ TL(\cong) \circ T(\lstrength{L}{\litp{A}{\Gamma}{T}{L}, \litp{A}{\bt_1}{T}{L}+\litp{A}{\bt_2}{T}{L}}) \\
      &\qquad \circ \lstrength{T}{\litp{A}{\Gamma}{T}{L}, L(\litp{A}{\bt_1}{T}{L}+\litp{A}{\bt_2}{T}{L})} \circ \langle \natdisc{\litp{A}{\Gamma}{T}{L}}, \litp{A}{M}{T}{L} \rangle \\
  &\qquad \text{where }\cong\colon \litp{A}{\Gamma, \bt_1}{T}{L}+\litp{A}{\Gamma, \bt_2}{T}{L} \to \litp{A}{\Gamma}{T}{L} \times (\litp{A}{\bt_1}{T}{L}+\litp{A}{\bt_2}{T}{L}), \\
  \litp{A}{\lambda x.~M}{T}{L} &\coloneqq \eta \circ L((\litp{A}{M}{T}{L} \circ \lstrength{L}{})^\dagger), \\
  \litp{A}{M \ N}{T}{L} &\coloneqq \mu \circ T(\mu) \circ T^2(\ev{L\litp{A}{\bt_1}{T}{L}}{TL\litp{A}{\bt_2}{T}{L}}) \circ T(\lstrength{T}{TL\litp{A}{\bt_2}{T}{L}^{L\litp{A}{\bt_1}{T}{L}}, L\litp{A}{\bt_1}{T}{L}}) \circ T(\id \times \litp{A}{N}{T}{L}) \\ 
      &\qquad \circ T(\langle \pi_1 \natdisc{L(TL\litp{A}{\bt_2}{T}{L}^{L\litp{A}{\bt_1}{T}{L}}) \times \litp{A}{\Gamma}{T}{L}}, L\pi_2 \rangle \circ \rstrength{L}{TL\litp{A}{\bt_2}{T}{L}^{L\litp{A}{\bt_1}{T}{L}}, \litp{A}{\Gamma}{T}{L}}) \\ 
      &\qquad  \circ \rstrength{T}{L(TL\litp{A}{\bt_2}{T}{L}^{L\litp{A}{\bt_1}{T}{L}}), \litp{A}{\Gamma}{T}{L}}\circ \langle \litp{A}{M}{T}{L}, \natdisc{\litp{A}{\Gamma}{T}{L}} \rangle
\end{align*}

\subsection{Proof of \cref{prop:trans}} \label{ap:transproof}
\begin{proof}
 We prove it by induction on the structure of $M$. 
 Here we let $L\colon \mathbb{C} \to \mathbb{C}$ for the functor $(\_ \times Y)$.
 Note that we consider a $\lambda_c(\uncurry{\Sigma})$ or $\lambda_c(\Sigma)$-structure $\tstructure = (\mathbb{C}, T, A, \uncurry{a})$.
 \begin{description}
  \item[Case $\varintro$.]
    For $\Gamma\vdash x\colon \bt_i$, we have
    \begin{align*}
    \itp{\tstructure}{\uncurry{x}}{T} 
    &= \itp{\tstructure}{(x, y)}{T} \\
    &= \mu_{T(\litp{\tstructure}{\bt_i}{T}{L} \times Y)} \circ T\lstrength{T}{\litp{\tstructure}{\bt_i}{T}{L}, Y} \circ \rstrength{T}{\litp{\tstructure}{\bt_i}{T}{L}, TY} \circ \langle \itp{\tstructure}{x}{T}, \itp{\tstructure}{y}{T}\rangle \\
    &= \mu_{T(\litp{\tstructure}{\bt_i}{T}{L} \times Y)} \circ T\lstrength{T}{\litp{\tstructure}{\bt_i}{T}{L}, Y} \circ \rstrength{T}{\litp{\tstructure}{\bt_i}{T}{L}, TY} \circ (\eta_{\litp{\tstructure}{\bt_i}{T}{L}} \times \eta_{Y}) \circ (\pi_i \times \id_Y) \\
    &= \mu_{T(\litp{\tstructure}{\bt_i}{T}{L} \times Y)} \circ \eta_{\litp{\tstructure}{\bt_i}{T}{L} \times Y} \circ (\pi_i \times \id_Y) \\
    &= \eta_{L\litp{\tstructure}{\bt_i}{T}{L}} \circ L\pi_i
    = \litp{\tstructure}{x}{T}{\_ \times Y}.
    \end{align*}
  \item[Case $\exchange$.] Easy.
  \item[Case $\conintro$.]
    \begin{align*}
    \itp{\tstructure}{\uncurry{(c \ M)}}{T} 
    &= \itp{\tstructure}{c \ \uncurry{M}}{T} \\
    &= T(\uncurry{a}(c)) \circ \itp{\tstructure}{\uncurry{M}}{T} \\
    &= T(\uncurry{a}(c)) \circ \litp{A}{M}{T}{L} \text{ by induction hypothesis} \\
    &= \litp{\tstructure}{c \ M}{T}{\_ \times Y}.
    \end{align*}
  \item[Case $\genintro$.] 
    \begin{align*}
    \itp{\tstructure}{\uncurry{(e \ M)}}{T} 
    &= \itp{\tstructure}{e \ \uncurry{M}}{T} \\
    &= \mu_{T\itp{\mathcal{A}}{\uncurry{\carfunc(e)}}{T}}\circ T\big( \uncurry{a}(e)\big)\circ \itp{\mathcal{A}}{\uncurry{M}}{T} \\
    &= \mu_{T\itp{\mathcal{A}}{\uncurry{\carfunc(e)}}{T}} \circ T(\uncurry{a}(e)) \circ \litp{A}{M}{T}{L} \text{ by induction hypothesis} \\
    &= \litp{\tstructure}{e \ M}{T}{\_ \times Y}.
    \end{align*}
  \item[Case $\unitprod$.] 
    \begin{align*}
    \itp{\tstructure}{\uncurry{()}}{T} 
    &= \itp{\tstructure}{((), y)}{T} \\
    &= \mu \circ T\lstrength{T}{} \circ \rstrength{T}{} \circ \langle \itp{\tstructure}{()}{T}, \itp{\tstructure}{y}{T}\rangle \\
    &= \eta_{1 \times Y} \circ L! \\
    &= \litp{\tstructure}{()}{T}{\_ \times Y}.
    \end{align*}
  \item[Case $\prodintro$.] 
    We can see the following equations.
    \begin{align*}
      &\itp{A}{(\lambda y. \uncurry{N})(\pi_2 z)}{T}\colon \itp{A}{\uncurry{\Gamma}, \primetrans{\bt_1} \times b^Y}{T} \to TL\itp{A}{\primetrans{\bt_2} \times b^Y}{T} \\
      &= \mu_{\itp{A}{\primetrans{\bt_2}, b^Y}{T}} \circ T\itp{A}{\uncurry{\Gamma}, y'\colon b^Y \vdash \uncurry{N}[y'/y]}{T} \circ \lstrength{T}{\itp{A}{\uncurry{\Gamma}}{T}, Y} \circ \itp{A}{\pi_2 z}{T} \text{ by }\cref{lem:betaReduction} \\
      &= \mu_{\itp{A}{\primetrans{\bt_2}, b^Y}{T}} \circ T\itp{A}{\uncurry{\Gamma} \vdash \uncurry{N}}{T} \circ T(\pi_1 \times \id_Y)\circ \lstrength{T}{\itp{A}{\uncurry{\Gamma}}{T}, Y} \circ \itp{A}{\pi_2 z}{T} \text{ by }\cref{lem:weakening} \\
      &= \mu_{\itp{A}{\primetrans{\bt_2}, b^Y}{T}} \circ T\itp{A}{\uncurry{N}}{T} \circ \lstrength{T}{\itp{A}{\primetrans{\Gamma}}{T}, Y} \circ \pi_1 \times (\eta_{Y} \circ \pi_2) \\
      &= \mu_{\itp{A}{\primetrans{\bt_2}, b^Y}{T}} \circ T\itp{A}{\uncurry{N}}{T} \circ \lstrength{T}{\itp{A}{\primetrans{\Gamma}}{T}, Y} \circ (\id_{\itp{A}{\primetrans{\Gamma}}{T}} \times \eta_Y) \circ (\pi_1 \times \pi_2) \\
      &= \itp{A}{\uncurry{N}}{T} \circ (\pi_1 \times \pi_2).
    \end{align*}
    \begin{align*}
      &\itp{\tstructure}{\big(\lambda z. (\pi_1 z, (\lambda y. \uncurry{N})(\pi_2 z) )\big) (\uncurry{M})}{T}\colon \itp{A}{\uncurry{\Gamma}}{T} \to TL\itp{A}{\primetrans{\bt_1} \times (\primetrans{\bt_2} \times b^Y)}{T} \\
      &= \mu \circ T(\itp{A}{(\pi_1 z, (\lambda y. \uncurry{N})(\pi_2 z))}{T}) \circ \lstrength{T}{\itp{A}{\uncurry{\Gamma}}, \itp{A}{\primetrans{\bt_1}\times b^Y}{T}} \circ \langle\id, \itp{A}{\uncurry{M}}{T} \rangle \text{ by }\cref{lem:betaReduction} \\
      &= \mu \circ T(T\mu \circ T\lstrength{T}{\itp{A}{\primetrans{\bt_1}}{T}, \itp{A}{\primetrans{\bt_2}\times b^Y}{T}} \circ T\rstrength{T}{\itp{A}{\primetrans{\bt_1}}{T}, T\itp{A}{\primetrans{\bt_2}\times b^Y}{T}} \circ \langle \itp{A}{\pi_1 z}{T}, \itp{A}{(\lambda y. \uncurry{N})(\pi_2 z)}{T} \rangle) \\
        &\qquad \circ \lstrength{T}{\itp{A}{\uncurry{\Gamma}}, \itp{A}{\primetrans{\bt_1}\times b^Y}{T}} \circ \langle\id, \itp{A}{\uncurry{M}}{T} \rangle \\
        &= \mu \circ T(T\mu \circ T\lstrength{T}{\itp{A}{\primetrans{\primetrans{\bt_1}}}{T}, \itp{A}{\primetrans{\bt_2}\times b^Y}{T}} \circ T\rstrength{T}{\itp{A}{\primetrans{\bt_1}}{T}, T\itp{A}{\primetrans{\bt_2}\times b^Y}{T}} \circ \eta_{\itp{A}{\primetrans{\bt_1}}{T}} \times \id \circ \langle \pi_1 \circ \pi_2, \itp{A}{(\lambda y. \uncurry{N})(\pi_2 z)}{T} \rangle) \\ 
        &\quad \circ \lstrength{T}{\itp{A}{\uncurry{\Gamma}}, \itp{A}{\primetrans{\bt_1}\times b^Y}{T}} \circ \langle\id, \itp{A}{\uncurry{M}}{T} \rangle \text{ since }\itp{A}{\pi_1 z}{T} = \eta_{\itp{A}{\primetrans{\bt_1}}{T}} \circ \pi_1 \circ \pi_2 \\
        &= \mu \circ T(\lstrength{T}{\itp{A}{\primetrans{\bt_1}}{T}, \itp{A}{\primetrans{\bt_2}\times b^Y}{T}} \circ \langle \pi_1 \circ \pi_2, \itp{A}{(\lambda y. \uncurry{N})(\pi_2 z)}{T} \rangle) \circ \lstrength{T}{\itp{A}{\uncurry{\Gamma}}, \itp{A}{\primetrans{\bt_1}\times b^Y}{T}} \circ \langle\id, \itp{A}{\uncurry{M}}{T} \rangle  \\
        &= \mu \circ T(\lstrength{T}{\itp{A}{\primetrans{\bt_1}}{T}, \itp{A}{\primetrans{\bt_2}\times b^Y}{T}} \circ \langle \pi_1 \circ \pi_2, \itp{A}{\uncurry{N}}{T} \circ \pi_1 \times \pi_2 \rangle) \circ \lstrength{T}{\itp{A}{\uncurry{\Gamma}}, \itp{A}{\primetrans{\bt_1}\times b^Y}{T}} \circ \langle\id, \itp{A}{\uncurry{M}}{T} \rangle  \\
        &= \mu \circ T(\lstrength{T}{\itp{A}{\primetrans{\bt_1}}{T}, \itp{A}{\primetrans{\bt_2}\times b^Y}{T}} \circ \langle \pi_1 \circ \pi_2, \itp{A}{\uncurry{N}}{T} \circ (\id_{\itp{A}{\primetrans{\Gamma}}{T}} \times \pi_2) \rangle) \circ  T(\pi_1 \times \id) \circ \lstrength{T}{\itp{A}{\uncurry{\Gamma}}, \itp{A}{\primetrans{\bt_1}\times b^Y}{T}} \circ \langle\id, \itp{A}{\uncurry{M}}{T} \rangle  \\
        &= \mu \circ T(\lstrength{T}{\itp{A}{\primetrans{\bt_1}}{T}, \itp{A}{\primetrans{\bt_2}\times b^Y}{T}} \circ (\id_{\itp{A}{\primetrans{\bt_1}}{T}} \times \itp{A}{\uncurry{N}}{T}) \circ \langle \pi_1 \circ \pi_2, \id_{\itp{A}{\primetrans{\Gamma}}{T}} \times \pi_2 \rangle) \circ \lstrength{T}{\itp{A}{\primetrans{\Gamma}}, \itp{A}{\primetrans{\bt_1}\times b^Y}{T}} \circ \langle\pi_1, \itp{A}{\uncurry{M}}{T} \rangle.
    \end{align*}

    \begin{align*}
      &\itp{A}{\big((\pi_1 z, \pi_1 \pi_2 z), \pi_2 \pi_2 z\big)}{T}\colon \itp{A}{\uncurry{\Gamma}, \primetrans{\bt_1} \times (\primetrans{\bt_2} \times b^Y)}{T} \to TL\itp{A}{(\primetrans{\bt_1} \times \primetrans{\bt_2}) \times b^Y}{T} \\
      &= \eta \circ \lstrength{L}{\itp{A}{\primetrans{\bt_1}}{T}, \itp{A}{\primetrans{\bt_2}}{T}} \circ \pi_2.
    \end{align*}

    Therefore, we have
    \begin{align*}
    &\itp{\tstructure}{\uncurry{(M, N)}}{T} \colon \itp{A}{\uncurry{\Gamma}}{T} \to TL\itp{A}{(\primetrans{\bt_1} \times \primetrans{\bt_2}) \times b^Y)}{T} \\
    &= \itp{\tstructure}{\Big(\lambda z. \big((\pi_1 z, \pi_1 \pi_2 z), \pi_2 \pi_2 z \big) \Big)\ \Big(\big(\lambda z. (\pi_1 z, (\lambda y. \uncurry{N})(\pi_2 z) )\big) (\uncurry{M})\Big)}{T} \\
    &= \mu_{\itp{A}{(\primetrans{\bt_1} \times \primetrans{\bt_2}) \times b^Y}{T}}\circ T(\itp{A}{\big((\pi_1 z, \pi_1 \pi_2 z), \pi_2 \pi_2 z\big)}{T})\circ \lstrength{T}{\itp{A}{\uncurry{\Gamma}}{T}, \itp{A}{\primetrans{\bt_1}, \primetrans{\bt_2} \times b^Y}{T}} \\ 
      &\qquad \circ  \langle \id, \itp{\tstructure}{\big(\lambda z. (\pi_1 z, (\lambda y. \uncurry{N})(\pi_2 z) )\big) (\uncurry{M})}{T}\rangle  \text{ by }\cref{lem:betaReduction} \\
      &= \mu_{\itp{A}{(\primetrans{\bt_1} \times \primetrans{\bt_2}) \times b^Y}{T}}\circ T(\eta \circ \lstrength{L}{\itp{A}{\primetrans{\bt_1}}{T}, \itp{A}{\primetrans{\bt_2}}{T}} \circ \pi_2)\circ \lstrength{T}{\itp{A}{\uncurry{\Gamma}}{T}, \itp{A}{\primetrans{\bt_1}, \primetrans{\bt_2} \times b^Y}{T}} \circ  \langle \id, \itp{\tstructure}{\big(\lambda z. (\pi_1 z, (\lambda y. \uncurry{N})(\pi_2 z) )\big) (\uncurry{M})}{T}\rangle \\
      &= T(\lstrength{L}{\itp{A}{\primetrans{\bt_1}}{T}, \itp{A}{\primetrans{\bt_2}}{T}}) \circ \pi_2\circ  \langle \id, \itp{\tstructure}{\big(\lambda z. (\pi_1 z, (\lambda y. \uncurry{N})(\pi_2 z) )\big) (\uncurry{M})}{T}\rangle \\
      &= T(\lstrength{L}{\itp{A}{\primetrans{\bt_1}}{T}, \itp{A}{\primetrans{\bt_2}}{T}}) \circ \itp{\tstructure}{\big(\lambda z. (\pi_1 z, (\lambda y. \uncurry{N})(\pi_2 z) )\big) (\uncurry{M})}{T} \\
      &= T(\lstrength{L}{\itp{A}{\primetrans{\bt_1}}{T}, \itp{A}{\primetrans{\bt_2}}{T}}) \circ \mu \circ T(\lstrength{T}{\itp{A}{\primetrans{\bt_1}}{T}, \itp{A}{\primetrans{\bt_2}\times b^Y}{T}} \circ (\id_{\itp{A}{\primetrans{\bt_1}}{T}} \times \itp{A}{\uncurry{N}}{T}) \circ \langle \pi_1 \circ \pi_2, \id_{\itp{A}{\primetrans{\Gamma}}{T}} \times \pi_2 \rangle) \\ 
      &\qquad \circ \lstrength{T}{\itp{A}{\primetrans{\Gamma}}, \itp{A}{\primetrans{\bt_1}\times b^Y}{T}} \circ \langle\pi_1, \itp{A}{\uncurry{M}}{T} \rangle \\
      &= \mu \circ T^2(\lstrength{L}{\itp{A}{\primetrans{\bt_1}}{T}, \itp{A}{\primetrans{\bt_2}}{T}}) \circ T(\lstrength{T}{\itp{A}{\primetrans{\bt_1}}{T}, \itp{A}{\primetrans{\bt_2}\times b^Y}{T}}) \circ T(\id_{\itp{A}{\primetrans{\bt_1}}{T}} \times \itp{A}{\uncurry{N}}{T}) \circ T(\langle \pi_1 \circ \pi_2, \id_{\itp{A}{\primetrans{\Gamma}}{T}} \times \pi_2 \rangle)  \\
      &\qquad \circ \lstrength{T}{\itp{A}{\primetrans{\Gamma}}, \itp{A}{\primetrans{\bt_1}\times b^Y}{T}} \circ \langle\pi_1, \itp{A}{\uncurry{M}}{T} \rangle \\
      &= \mu \circ T^2(\lstrength{L}{\litp{A}{\bt_1}{T}{L}, \litp{A}{\bt_2}{T}{L}}) \circ T(\lstrength{T}{\litp{A}{t_1}{T}{L}, \litp{A}{t_2}{T}{L}}) \circ T(\id \times \litp{A}{N}{T}{L}) \circ T(\langle \pi_1 \natdisc{\litp{A}{t_1}{T}{L} \times \litp{A}{\Gamma}{T}{L}}, L\pi_2 \rangle \circ \rstrength{L}{\litp{A}{t_1}{T}{L}, \litp{A}{\Gamma}{T}{L}}) \\
        &\qquad \circ \rstrength{T}{L\litp{A}{t_1}{T}{L}, \litp{A}{\Gamma}{T}{L}} \circ \langle \litp{\mathcal{A}}{M}{T}{L}, \pi_1\rangle \\
    &= \litp{\tstructure}{(M, N)}{T}{\_ \times Y}.
    \end{align*}
  \item[Case $\projintroone$, $\projintrotwo$.] 
    Noting that 
    $\itp{A}{(\pi_i \pi_1 z, \pi_2 z)}{T} 
    = \eta \circ \pi_i \times Y \circ \pi_2\colon \itp{A}{\uncurry{\Gamma}, (\primetrans{\bt_1} \times \primetrans{\bt_2}) \times b^Y}{T} \to \itp{A}{\primetrans{\bt_i} \times b^Y}{T}$,
    we have
    \begin{align*}
    \itp{\tstructure}{\uncurry{(\pi_i \ M)}}{T} 
    &= \itp{\tstructure}{\big(\lambda z. (\pi_i \pi_1 z, \pi_2 z)\big)(\uncurry{M})}{T} \\
    &= \mu_{\itp{A}{\primetrans{\bt_i} \times b^Y}{T}}\circ T\itp{A}{(\pi_i \pi_1 z, \pi_2 z)}{T}\circ \lstrength{T}{\itp{A}{\uncurry{\Gamma}}{T}, \itp{A}{(\primetrans{\bt_1} \times \primetrans{\bt_2}) \times b^Y}{T}} \circ \langle \id_{\itp{A}{\uncurry{\Gamma}}{T}}, \itp{A}{\uncurry{M}}{T}\rangle \\
    &= \mu_{\itp{A}{\primetrans{\bt_i} \times b^Y}{T}}\circ T(\eta \circ \pi_i \times Y \circ \pi_2)\circ \lstrength{T}{\itp{A}{\uncurry{\Gamma}}{T}, \itp{A}{(\primetrans{\bt_1} \times \primetrans{\bt_2}) \times b^Y}{T}} \circ \langle \id_{\itp{A}{\uncurry{\Gamma}}{T}}, \itp{A}{\uncurry{M}}{T}\rangle \\
    &= T(\pi_i \times Y \circ \pi_2)\circ \lstrength{T}{\itp{A}{\uncurry{\Gamma}}{T}, \itp{A}{(\primetrans{\bt_1} \times \primetrans{\bt_2}) \times b^Y}{T}} \circ \langle \id_{\itp{A}{\uncurry{\Gamma}}{T}}, \itp{A}{\uncurry{M}}{T}\rangle \\
    &= T(\pi_i \times Y)\circ \pi_2 \circ \langle \id_{\itp{A}{\uncurry{\Gamma}}{T}}, \itp{A}{\uncurry{M}}{T}\rangle \\
    &= TL\pi_i \circ \litp{A}{M}{T}{L}
    = \litp{\tstructure}{\pi_i \ M}{T}{\_ \times Y}.
    \end{align*}
  \item[Case $\unitcoprod$.] 
    Noting that 
    $\itp{A}{(\delta(\pi_1 z), \pi_2 z)}{T} 
    = \eta \circ (! \times Y) \circ \pi_2\colon \itp{A}{\uncurry{\Gamma}, 0 \times b^Y}{T} \to \itp{A}{\primetrans{\bt} \times b^Y}{T}$,
    \begin{align*}
    \itp{\tstructure}{\uncurry{\delta(M)}}{T} 
    &= \itp{\tstructure}{\big(\lambda z. (\delta(\pi_1 z), \pi_2 z)\big) (\uncurry{M})}{T} \\
    &= \mu_{\itp{A}{\primetrans{\bt} \times b^Y}{T}}\circ T\itp{A}{(\delta(\pi_1 z), \pi_2 z)}{T}\circ \lstrength{T}{\itp{A}{\uncurry{\Gamma}}{T}, \itp{A}{0 \times b^Y}{T}} \circ \langle \id_{\itp{A}{\uncurry{\Gamma}}{T}}, \itp{A}{\uncurry{M}}{T}\rangle \\
    &= \mu_{\itp{A}{\primetrans{\bt} \times b^Y}{T}}\circ T(\eta \circ (! \times Y) \circ \pi_2)\circ \lstrength{T}{\itp{A}{\uncurry{\Gamma}}{T}, \itp{A}{0 \times b^Y}{T}} \circ \langle \id_{\itp{A}{\uncurry{\Gamma}}{T}}, \itp{A}{\uncurry{M}}{T}\rangle \\
    &= TL! \circ \litp{A}{M}{T}{L}
    = \litp{\tstructure}{\delta(M)}{T}{\_ \times Y}.
    \end{align*}
  \item[Case $\coprodone$, $\coprodtwo$.]
    Noting that 
    $\itp{A}{(\iota_1 \pi_1 z, \pi_2 z)}{T} 
    = \eta \circ (\iota_i \times Y) \circ \pi_2\colon \itp{A}{\uncurry{\Gamma}, \primetrans{\bt_i} \times b^Y}{T} \to \itp{A}{(\primetrans{\bt_1} + \primetrans{\bt_2}) \times b^Y}{T}$, we have
    \begin{align*}
    \itp{\tstructure}{\uncurry{\iota_i \ M}}{T} 
    &= \itp{\tstructure}{\big(\lambda z. (\iota_1 \pi_1 z, \pi_2 z)\big)(\uncurry{M})}{T} \\
    &= \mu_{\itp{A}{(\primetrans{\bt_1} + \primetrans{\bt_2}) \times b^Y}{T}}\circ T\itp{A}{(\iota_1 \pi_1 z, \pi_2 z)}{T}\circ \lstrength{T}{\itp{A}{\uncurry{\Gamma}}{T}, \itp{A}{\primetrans{\bt_i} \times b^Y}{T}} \circ \langle \id_{\itp{A}{\uncurry{\Gamma}}{T}}, \itp{A}{\uncurry{M}}{T}\rangle \\
    &= \mu_{\itp{A}{(\primetrans{\bt_1} + \primetrans{\bt_2}) \times b^Y}{T}}\circ T(\eta \circ (\iota_i \times Y) \circ \pi_2)\circ \lstrength{T}{\itp{A}{\uncurry{\Gamma}}{T}, \itp{A}{\primetrans{\bt_i} \times b^Y}{T}} \circ \langle \id_{\itp{A}{\uncurry{\Gamma}}{T}}, \itp{A}{\uncurry{M}}{T}\rangle \\
    &= T(\iota_i \times Y \circ \pi_2)\circ \lstrength{T}{\itp{A}{\uncurry{\Gamma}}{T}, \itp{A}{\primetrans{\bt_i} \times b^Y}{T}} \circ \langle \id_{\itp{A}{\uncurry{\Gamma}}{T}}, \itp{A}{\uncurry{M}}{T}\rangle \\
    &= TL\iota_i \circ \litp{A}{M}{T}{L}
    = \litp{\tstructure}{\iota_i \ M}{T}{\_ \times Y}.
    \end{align*}
  \item[Case $\coprodelm$.]
  \newcommand{\congbi}[3]{\cong_{#1, #2, #3}}
    We write $\congbi{A}{B}{C}$ for the isomorphism $A \times (B + C) \to A \times B + A \times C$.
    We can see that
    \begin{align*}
      &\itp{A}{\iota_i(x_i, \pi_2 z)}{T} \colon \itp{A}{\uncurry{\Gamma}, z\colon (\primetrans{\bt_1} + \primetrans{\bt_2}) \times b^Y, x_i\colon \primetrans{\bt_i}}{T} \to \itp{A}{(\primetrans{\bt_1} \times b^Y) + (\primetrans{\bt_2} \times b^Y)}{T} \\
      &= T\iota_i \circ (\eta_{\itp{A}{\primetrans{\bt_1} \times b^Y}{T}} \circ \langle \pi_{x_i}, \pi_2 \pi_z  \rangle) \\
      &= \eta_{\itp{A}{\primetrans{\bt_1} \times b^Y}{T}} \circ \iota_i \circ \langle \pi_{x_i}, \pi_2 \pi_z  \rangle,
    \end{align*}
    where $\pi_{x_i}\colon \itp{A}{\uncurry{\Gamma}, z\colon (\primetrans{\bt_1} + \primetrans{\bt_2}) \times b^Y, x_i\colon \primetrans{\bt_i}}{T} \to \itp{A}{\primetrans{\bt_i}}{T}$
    and $\pi_z \colon \itp{A}{\uncurry{\Gamma}, z\colon (\primetrans{\bt_1} + \primetrans{\bt_2}) \times b^Y, x_i\colon \primetrans{\bt_i}}{T} \to \itp{A}{(\primetrans{\bt_1} + \primetrans{\bt_2}) \times b^Y}{T}$ are projections,
    \begin{align*}
      &\itp{A}{\delta(\pi_1 z, x_1\colon \primetrans{\bt_1}.~\iota_1(x_1, \pi_2 z), x_2\colon \primetrans{\bt_2}.~\iota_2(x_2, \pi_2 z))}{T} \\
      &= \mu \circ T([\itp{A}{\iota_1(x_1, \pi_2 z)}{T}, \itp{A}{\iota_2(x_2, \pi_2 z)}{T}]) \circ T(\congbi{\itp{A}{\uncurry{\Gamma}, (\primetrans{\bt_1} + \primetrans{\bt_2}) \times b^Y}{T}}{\itp{A}{\primetrans{\bt_1}}{T}}{\itp{A}{\primetrans{\bt_2}}{T}}) \\ 
        &\qquad \circ \lstrength{T}{\itp{A}{\uncurry{\Gamma}, (\primetrans{\bt_1} + \primetrans{\bt_2}) \times b^Y}{T}, \itp{A}{\primetrans{\bt_1} + \primetrans{\bt_2}}{T}} \circ \langle \id, \itp{A}{\pi_1 z}{T} \rangle \text{ by }\cref{lem:betaReduction} \\
      &= \mu \circ T([\itp{A}{\iota_1(x_1, \pi_2 z)}{T}, \itp{A}{\iota_2(x_2, \pi_2 z)}{T}]) \circ T(\congbi{\itp{A}{\uncurry{\Gamma}, (\primetrans{\bt_1} + \primetrans{\bt_2}) \times b^Y}{T}}{\itp{A}{\primetrans{\bt_1}}{T}}{\itp{A}{\primetrans{\bt_2}}{T}}) \\ 
        &\qquad \circ \eta \circ \langle \id, \pi_1 \rangle \text{ since }\itp{A}{\pi_1 z}{T} = \eta \circ \pi_1  \\
      &= [\itp{A}{\iota_1(x_1, \pi_2 z)}{T}, \itp{A}{\iota_2(x_2, \pi_2 z)}{T}] \circ \congbi{\itp{A}{\uncurry{\Gamma}, (\primetrans{\bt_1} + \primetrans{\bt_2}) \times b^Y}{T}}{\itp{A}{\primetrans{\bt_1}}{T}}{\itp{A}{\primetrans{\bt_2}}{T}} \circ \langle \id, \pi_1 \rangle   \\
      &= \eta \circ (\langle \pi_{x_1}, \pi_2 \pi_z  \rangle + \langle \pi_{x_2}, \pi_2 \pi_z  \rangle) \circ \congbi{\itp{A}{\uncurry{\Gamma}, (\primetrans{\bt_1} + \primetrans{\bt_2}) \times b^Y}{T}}{\itp{A}{\primetrans{\bt_1}}{T}}{\itp{A}{\primetrans{\bt_2}}{T}} \circ \langle \id, \pi_1 \rangle   \\
      &= \eta \circ \natcoprod{\itp{A}{\primetrans{\bt_1}}{T}, \itp{A}{\primetrans{\bt_2}}{T}} \circ \pi_1,
    \end{align*}
    \begin{align*}
    &\itp{A}{\big(\lambda z.~\delta(\pi_1 z, x_1\colon \primetrans{\bt_1}.~\iota_1(x_1, \pi_2 z), x_2\colon \primetrans{\bt_2}.~\iota_2(x_2, \pi_2 z))\big)(\uncurry{M})}{T}  \\
    &= \mu \circ T(\itp{A}{\delta(\pi_1 z, x_1\colon \primetrans{\primetrans{\bt_1}}.~\iota_1(x_1, \pi_2 z), x_2\colon \primetrans{\primetrans{\bt_2}}.~\iota_2(x_2, \pi_2 z))}{T}) \circ \lstrength{T}{\itp{A}{\uncurry{\Gamma}}{T}, \itp{A}{(\primetrans{\bt_1} + \primetrans{\bt_2}) \times b^Y}{T}} \circ \langle \id_{\itp{A}{\uncurry{\Gamma}}{T}}, \itp{A}{\uncurry{M}}{T}\rangle \\
    &= \mu \circ T(\eta \circ \natcoprod{\itp{A}{\primetrans{\bt_1}}{T}, \itp{A}{\primetrans{\bt_2}}{T}} \circ \pi_1) \circ \lstrength{T}{\itp{A}{\uncurry{\Gamma}}{T}, \itp{A}{(\primetrans{\bt_1} + \primetrans{\bt_2}) \times b^Y}{T}} \circ \langle \id_{\itp{A}{\uncurry{\Gamma}}{T}}, \itp{A}{\uncurry{M}}{T}\rangle \\
    &= T(\natcoprod{\itp{A}{\primetrans{\bt_1}}{T}, \itp{A}{\primetrans{\bt_2}}{T}}) \circ \pi_1 \circ \langle \id_{\itp{A}{\uncurry{\Gamma}}{T}}, \itp{A}{\uncurry{M}}{T}\rangle \\
    &= T(\natcoprod{\itp{A}{\primetrans{\bt_1}}{T}, \itp{A}{\primetrans{\bt_2}}{T}}) \circ \itp{A}{\uncurry{M}}{T}.
    \end{align*}
    Therefore, we have
    \begin{align*}
    &\itp{\tstructure}{\uncurry{\big(\delta(M, x_1\colon \bt_1.\ M_1, x_2\colon \bt_2.\ M_2)\big)}}{T}  \\
    &= \itp{\tstructure}{\delta(N, z\colon {\uncurry{\bt_1}}.\ \uncurry{M_1}[\pi_1 z / x_1, \pi_2 z / y],\, z\colon \uncurry{\bt_2}.\ \uncurry{M_2}[\pi_1 z / x_2, \pi_2 z / y])}{T} \\
    &= \mu \circ T([\itp{A}{\uncurry{M_1}[\pi_1 z / x_1, \pi_2 z / y]}{T}, \itp{A}{\uncurry{M_2}[\pi_1 z / x_2, \pi_2 z / y]}{T}]) \circ T(\congbi{\itp{A}{\uncurry{\Gamma}}{T}}{\itp{A}{\uncurry{\bt_1}}{T}}{\itp{A}{\uncurry{\bt_2}}{T}}) \\
      &\qquad \circ \lstrength{T}{\itp{A}{\uncurry{\Gamma}}{T}, \itp{A}{\bt_1 \times b^Y}{T}+\itp{A}{\bt_2 \times b^Y}{T}} \circ \langle \id, \itp{A}{N}{T} \rangle \\
      &= \mu \circ T([\itp{A}{\uncurry{M_1}}{T}, \itp{A}{\uncurry{M_2}}{T}] \circ ((\pi_1 \times \id_{\itp{A}{\uncurry{\bt_1}}{T}}) + (\pi_1 \times \id_{\itp{A}{\uncurry{\bt_2}}{T}}))) \circ T(\congbi{\itp{A}{\uncurry{\Gamma}}{T}}{\itp{A}{\uncurry{\bt_1}}{T}}{\itp{A}{\uncurry{\bt_2}}{T}})  \\
      &\qquad \circ \lstrength{T}{\itp{A}{\uncurry{\Gamma}}{T}, \itp{A}{\uncurry{\bt_1}}{T}+\itp{A}{\uncurry{\bt_2}}{T}} \circ \langle \id, T(\natcoprod{\itp{A}{\primetrans{\bt_1}}{T}, \itp{A}{\primetrans{\bt_2}}{T}}) \circ \itp{A}{\uncurry{M}}{T} \rangle \\
      &= \mu \circ T([\itp{A}{\uncurry{M_1}}{T}, \itp{A}{\uncurry{M_2}}{T}]) \circ T(\congbi{\itp{A}{\primetrans{\Gamma}}{T}}{\itp{A}{\uncurry{\bt_1}}{T}}{\itp{A}{\uncurry{\bt_2}}{T}}) \circ \lstrength{T}{\itp{A}{\primetrans{\Gamma}}{T}, \itp{A}{\uncurry{\bt_1}}{T}+\itp{A}{\uncurry{\bt_2}}{T}} \circ \pi_1 \times \id \\
      &\qquad \circ (\id \times T(\natcoprod{\itp{A}{\primetrans{\bt_1}}{T}, \itp{A}{\primetrans{\bt_2}}{T}})) \circ \langle \id, \itp{A}{\uncurry{M}}{T} \rangle \\
      &= \mu \circ T([\itp{A}{\uncurry{M_1}}{T}, \itp{A}{\uncurry{M_2}}{T}]) \circ T(\congbi{\itp{A}{\primetrans{\Gamma}}{T}}{\itp{A}{\uncurry{\bt_1}}{T}}{\itp{A}{\uncurry{\bt_2}}{T}}) \circ T(\id \times \natcoprod{\itp{A}{\primetrans{\bt_1}}{T}, \itp{A}{\primetrans{\bt_2}}{T}}))  \\
      &\qquad \circ \lstrength{T}{\itp{A}{\primetrans{\Gamma}}{T}, \itp{A}{(\primetrans{\bt_1} + \primetrans{\bt_2}) \times b^Y}{T}} \circ \langle \pi_1, \itp{A}{\uncurry{M}}{T} \rangle \\
      &= \mu \circ T([\litp{A}{M_1}{T}{L}, \litp{A}{M_2}{T}{L}]) \circ T(\natcoprod{\litp{A}{\Gamma, \bt_1}{T}{L}, \litp{A}{\Gamma, \bt_2}{T}{L}}) \circ TL(\cong) \circ T(\lstrength{L}{\litp{A}{\Gamma}{T}{L}, \litp{A}{\bt_1}{T}{L}+\litp{A}{\bt_2}{T}{L}}) \circ \lstrength{T}{\litp{A}{\Gamma}{T}{L}, L(\litp{A}{\bt_1}{T}{L}+\litp{A}{\bt_2}{T}{L})} \circ \langle \natdisc{\litp{A}{\Gamma}{T}{L}}, \litp{A}{M}{T}{L} \rangle \\
    &= \litp{\tstructure}{\big(\delta(M, x_1\colon \bt_1.\ M_1, x_2\colon \bt_2.\ M_2)\big)}{T}{\_ \times Y},
    \end{align*}
    where $N \defeq \big(\lambda z.~\delta(\pi_1 z, x_1\colon \primetrans{\bt_1}.~\iota_1(x_1, \pi_2 z), x_2\colon \primetrans{\bt_2}.~\iota_2(x_2, \pi_2 z)\big)(\uncurry{M})$.
  \item[Case $\abst$.]
    We have 
    \begin{align*}
      \itp{A}{\lambda z\colon \uncurry{\bt_1}.~\uncurry{M}[\pi_1 z/x, \pi_2 z/y]}{T} &= \eta \circ \itp{A}{\uncurry{M}[\pi_1 z/x, \pi_2 z/y]}{T}^\dagger, \text{ and } \\
      itp{A}{\uncurry{M}[\pi_1 z/x, \pi_2 z/y]}{T} &= \itp{A}{\uncurry{M}}{T} \circ \lstrength{L}{\itp{A}{\primetrans{\Gamma}}{T}, \itp{A}{\primetrans{\bt_1}}{T}} \circ \pi_1 \times \id_{\itp{A}{\primetrans{\bt_1}}{T} \times Y}.
    \end{align*}
    Therefore,
    \begin{align*}
    \itp{\tstructure}{\uncurry{(\lambda x\colon \bt_1.~M)}}{T} 
    &= \itp{\tstructure}{\big(\lambda z\colon \uncurry{\bt_1}\uncurry{M}[\pi_1 z/x, \pi_2 z/y],\, y\big)}{T} \\
    &= \mu \circ T\lstrength{T}{T\itp{A}{\uncurry{\bt_2}}{T}^{\itp{A}{\uncurry{\bt_1}}{T}}, Y} \circ \rstrength{T}{T\itp{A}{\uncurry{\bt_2}}{T}^{\itp{A}{\uncurry{\bt_1}}{T}}, TY} \circ \langle \itp{A}{\lambda z\colon \uncurry{\bt_1}\uncurry{M}[\pi_1 z/x, \pi_2 z/y]}{T}, \eta \pi_2\rangle \\
    &= \lstrength{T}{T\itp{A}{\uncurry{\bt_2}}{T}^{\itp{A}{\uncurry{\bt_1}}{T}}, Y} \circ \langle \itp{A}{\lambda z\colon \uncurry{\bt_1}\uncurry{M}[\pi_1 z/x, \pi_2 z/y]}{T}, \pi_2\rangle \\
    &= \eta \circ \langle \itp{A}{\uncurry{M}[\pi_1 z/x, \pi_2 z/y]}{T}^\dagger, \pi_2\rangle \\
    &= \eta \circ \langle (\itp{A}{\uncurry{M}}{T} \circ \lstrength{L}{})^\dagger \circ \pi_1, \pi_2\rangle \\
    &= \eta \circ L((\litp{A}{M}{T}{L} \circ \lstrength{L}{})^\dagger)
    = \litp{\tstructure}{\lambda x\colon \bt_1.~M}{T}{\_ \times Y}.
    \end{align*}
  \item[Case $\appl$.]
    Since 
    \begin{align*}
      &\itp{A}{(\lambda y. \uncurry{N})(\pi_2 z)}{T} \\
      &= \mu \circ T(\itp{A}{\uncurry{\Gamma}, z\colon (\uncurry{\bt_1} \to \uncurry{\bt_2}) \times b^Y, y'\colon b^Y \vdash \uncurry{N}[y'/y]}{T})\circ \lstrength{T}{\itp{A}{\uncurry{\Gamma}, \uncurry{\bt_1 \to \bt_2}}{T}, Y} \circ \langle \id_{\itp{A}{\uncurry{\Gamma}}{T}}, \itp{A}{\pi_2 z}{T}\rangle \\
      &= \mu \circ T(\itp{A}{\uncurry{\Gamma} \vdash \uncurry{N}}{T} \circ \pi_{\primetrans{\Gamma}} \times \id_Y) \circ \lstrength{T}{\itp{A}{\uncurry{\Gamma}, \uncurry{\bt_1 \to \bt_2}}{T}, Y} \circ \langle \id_{\itp{A}{\uncurry{\Gamma}}{T}}, \itp{A}{\pi_2 z}{T}\rangle \text{ by }\cref{lem:weakening} \\
      &= \mu \circ T(\itp{A}{\uncurry{N}}{T}) \circ \lstrength{T}{\itp{A}{\primetrans{\Gamma}}{T}, Y} \circ \id \times \eta_Y \circ \pi_1 \times \pi_2  \\
      &= \itp{A}{\uncurry{N}}{T} \circ \pi_1 \times \pi_2,
    \end{align*}
    where $\pi_{\primetrans{\Gamma}}\colon \itp{A}{\primetrans{\Gamma}, b^Y, \uncurry{\bt_1 \to \bt_2}}{T} \to \itp{A}{\primetrans{\Gamma}}{T}$ is the projection,
    the following equations hold.
    \begin{align*}
    &\itp{A}{\big(\pi_1 z\big)\big((\lambda y. \uncurry{N})(\pi_2 z) \big)}{T}  \\
    &= \mu\circ T\mu \circ T^2(\ev{\itp{A}{\uncurry{\bt_1}}{T}}{T\itp{A}{\uncurry{\bt_2}}{T}}) \circ T\lstrength{T}{\itp{A}{\uncurry{\bt_1} \to \uncurry{\bt_2}}{T}, \itp{A}{\uncurry{\bt_1}}{T}} \circ \rstrength{T}{\itp{A}{\uncurry{\bt_1} \to \uncurry{\bt_2}}{T}, T\itp{A}{\uncurry{\bt_1}}{T}} \circ \langle \itp{A}{\pi_1 z}{T}, \itp{A}{(\lambda y. \uncurry{N})(\pi_2 z)}{T}\rangle \\
    &= \mu\circ T\mu \circ T^2(\ev{\itp{A}{\uncurry{\bt_1}}{T}}{T\itp{A}{\uncurry{\bt_2}}{T}}) \circ T\lstrength{T}{\itp{A}{\uncurry{\bt_1} \to \uncurry{\bt_2}}{T}, \itp{A}{\uncurry{\bt_1}}{T}} \circ \rstrength{T}{\itp{A}{\uncurry{\bt_1} \to \uncurry{\bt_2}}{T}, T\itp{A}{\uncurry{\bt_1}}{T}} \circ \eta \times \itp{A}{\uncurry{N}}{T} \\ 
      &\qquad \circ \langle \pi_1 \pi_2, \pi_1 \times \pi_2\rangle \\
    &= \mu \circ T(\ev{\itp{A}{\uncurry{\bt_1}}{T}}{T\itp{A}{\uncurry{\bt_2}}{T}}) \circ \lstrength{T}{\itp{A}{\uncurry{\bt_1} \to \uncurry{\bt_2}}{T}, \itp{A}{\uncurry{\bt_1}}{T}} \circ \id \times \itp{A}{\uncurry{N}}{T} \circ \langle \pi_1 \pi_2, \pi_1 \times \pi_2\rangle.
    \end{align*}
    Therefore, we have
    \begin{align*}
    \itp{\tstructure}{\uncurry{(M \ N)}}{T} 
    &= \itp{\tstructure}{\Big(\lambda z. \big(\pi_1 z\big)\big((\lambda y. \uncurry{N})(\pi_2 z) \big)\Big)\big(\uncurry{M}\big)}{T} \\
    &= \mu\circ T(\itp{A}{\big(\pi_1 z\big)\big((\lambda y. \uncurry{N})(\pi_2 z) \big)}{T})\circ \lstrength{T}{\itp{A}{\uncurry{\Gamma}}{T}, \itp{A}{(\uncurry{\bt_1} \to \uncurry{\bt_2}) \times b^Y}{T}} \circ \langle \id_{\itp{A}{\uncurry{\Gamma}}{T}}, \itp{A}{\uncurry{M}}{T}\rangle \\
    &= \mu\circ T(\mu \circ T(\ev{\itp{A}{\uncurry{\bt_1}}{T}}{T\itp{A}{\uncurry{\bt_2}}{T}}) \circ \lstrength{T}{\itp{A}{\uncurry{\bt_1} \to \uncurry{\bt_2}}{T}, \itp{A}{\uncurry{\bt_1}}{T}} \circ \id \times \itp{A}{\uncurry{N}}{T} \circ \langle \pi_1 \pi_2, \pi_1 \times \pi_2\rangle) \\ 
      &\qquad \circ \lstrength{T}{\itp{A}{\uncurry{\Gamma}}{T}, \itp{A}{(\uncurry{\bt_1} \to \uncurry{\bt_2}) \times b^Y}{T}} \circ \langle \id_{\itp{A}{\uncurry{\Gamma}}{T}}, \itp{A}{\uncurry{M}}{T}\rangle \\
      &= \mu\circ T(\mu \circ T(\ev{\itp{A}{\uncurry{\bt_1}}{T}}{T\itp{A}{\uncurry{\bt_2}}{T}}) \circ \lstrength{T}{\itp{A}{\uncurry{\bt_1} \to \uncurry{\bt_2}}{T}, \itp{A}{\uncurry{\bt_1}}{T}} \circ \id \times \itp{A}{\uncurry{N}}{T} \circ \langle \pi_1 \pi_2, \id_{\itp{A}{\primetrans{\Gamma}}{T}} \times \pi_2\rangle) \\ 
      &\qquad \circ \lstrength{T}{\itp{A}{\primetrans{\Gamma}}{T}, \itp{A}{(\uncurry{\bt_1} \to \uncurry{\bt_2}) \times b^Y}{T}} \circ \langle \pi_1, \itp{A}{\uncurry{M}}{T}\rangle \\
    &= \mu \circ T(\mu) \circ T^2(\ev{L\litp{A}{\bt_1}{T}{L}}{TL\litp{A}{\bt_2}{T}{L}}) \circ T(\lstrength{T}{TL\litp{A}{\bt_2}{T}{L}^{L\litp{A}{\bt_1}{T}{L}}, L\litp{A}{\bt_1}{T}{L}}) \circ T(\id \times \litp{A}{N}{T}{L}) \\ &\qquad \circ T(\langle \pi_1 \natdisc{}, L\pi_2 \rangle \circ \rstrength{L}{TL\litp{A}{\bt_2}{T}{L}^{L\litp{A}{\bt_1}{T}{L}}, \litp{A}{\Gamma}{T}{L}}) \circ \rstrength{T}{L(TL\litp{A}{\bt_2}{T}{L}^{L\litp{A}{\bt_1}{T}{L}}), \litp{A}{\Gamma}{T}{L}} \circ \langle \litp{A}{M}{T}{L}, \natdisc{\litp{A}{\Gamma}{T}{L}} \rangle \\
    &= \litp{\tstructure}{M \ N}{T}{\_ \times Y}.
    \end{align*}
 \end{description}
\end{proof}

\subsection{Proof of \cref{thm:mainfib}} \label{ap:proof_liftlambdac}
We write $f\colon P \darrow Q$ when there is a morphism $P \to Q$ in $\mathbb{K}_D$ above $f\colon pP \to pQ$.
Since $p$ is full and faithful, there is at most one such morphism.
Here let $L \coloneqq \id \times (\_ \times Y)\colon \mathbb{C} \to \mathbb{C}$ and $U \coloneqq S \times T$.
We consider natural transformations $d \coloneqq \pi_2$ and $\lstrength{L}{}\coloneqq \langle \langle \pi_1, \pi_1 \pi_2 \rangle, \pi_2 \pi_2 \rangle$.
\begin{lemma} \label{lem:liftcv}
  \begin{enumerate}
    \item \label{item:liftnat} $\dot{L}$ is a strong functor lifting of $L$, and there exists a lifting of the natural transformation $\natdisc{}\colon L \Rightarrow \id$.
    \item \label{item:ce} For any $c \in K$, $(a \times \uncurry{a})(c) \colon \dot{L}V(\arfunc(c)) \darrow \dot{L}V(\carfunc(c))$,
      and for any $e \in E$, $(a \times \uncurry{a})(e) \colon \dot{L}V(\arfunc(e)) \darrow C(\carfunc(e))$.
    \item \label{item:eta} For any type $\bt$, 
      $\eta_{L\litp{A}{\bt}{U}{L}}\colon \dot{L}V(\bt) \darrow C(\bt)$.
    \item \label{item:bind} If $f\colon V(\bt_1) \times \dot{L}V(\bt_2) \darrow C(\bt_3)$, then 
      $\mu \circ U(f) \circ \lstrength{U}{}\colon V(\bt_1) \times C(\bt_2) \darrow C(\bt_3)$.
      In particular, it induces 
      $U\lstrength{L}{\litp{A}{\bt_1}{U}{L}, \litp{A}{\bt_2}{U}{L}}\circ \lstrength{U}{\litp{A}{\bt_1}{U}{L}, L\litp{A}{\bt_2}{U}{L}}\colon V(\bt_1) \times C(\bt_2) \darrow C(\bt_1 \times \bt_2)$.
    \item \label{item:coprod} 
      For any types $\bt_1, \bt_2$, 
      $!_{\litp{A}{\bt_1}{U}{L}}\colon V(0) \darrow V(\bt_1)$, $\iota_i\colon V(\bt_i) \to V(\bt_1 + \bt_2)$ for $i \in \{1, 2\}$, and $\cong\colon V(\Gamma \times (\bt_1 + \bt_2)) \darrow V(\Gamma \times \bt_1 + \Gamma \times \bt_2)$. 
      Moreover, if $f_i\colon \dot{L}V(\bt_i) \darrow C(\bt)$ for $i \in \{1, 2\}$ then $[f_1, f_2] \circ \natcoprod{\litp{A}{\bt_1}{U}{L}, \litp{A}{\bt_2}{U}{L}}\colon \dot{L}V(\bt_1 + \bt_2) \darrow C(\bt)$.
    \item \label{item:exp}
      For any $f\colon \litp{A}{\Gamma}{U}{L} \times L\litp{A}{\bt_1}{U}{L} \to UL\litp{A}{\bt_2}{U}{L}$,
      $f\colon V(\Gamma) \times \dot{L}V(\bt_1) \darrow C(\bt_2)$ if and only if $f^\dagger\colon V(\Gamma) \darrow V(\bt_1 \to \bt_2)$.
  \end{enumerate}
\end{lemma}
\begin{proof}
  \begin{enumerate}
    \item For each $P, Q \in \mathbb{K}_D$ above $(I_1, I_2), (J_1, J_2)$ respectively,
    $\lstrength{L}{}\colon P \times \dot{L}Q \darrow \dot{L}(P \times Q)$
    because $P \times \dot{L}Q = \{(\langle f_1, f_2 \rangle, \langle g_1, \langle g_2, \pi_2 \rangle \rangle) \mid (f_1, g_1) \in P, (f_2, g_2) \in Q\}$
    and $\dot{L}(P \times Q) = \{(\langle f_1, f_2 \rangle, \langle \langle g_1, g_2 \rangle \pi_2 \rangle) \mid (f_1, g_1) \in P, (f_2, g_2) \in Q\}$.
    It is easy to see that $\pi_2\colon \dot{L}P \darrow P$ for each $P \in \mathbb{K}_D$.
    \item It is easy by definition of $\uncurry{a}$.
    \item It is because the transpose of the unit at $X$ for $S$ is equal to the unit at $X \times Y$ for $T$.
    \item Assume that $f\colon V(\bt_1) \times \dot{L}V(\bt_2) \darrow C(\bt_3)$, which is equivalent to $T(\rho_{\bt_3} \times \id_Y) \circ f_1^\dagger = f_2 \circ \rho_{\bt_1 \times (\bt_2 \times b^Y)} \circ (\lstrength{L}{\itp{A}{\bt_1}{S}, \itp{A}{\bt_2}{S}})^{-1}$.
    Then for each $\langle g_1, g_2 \rangle\colon H \to \itp{A}{\bt_1}{S} \times S\itp{A}{\bt_2}{S}$,
    \begin{align*}
      &T(\rho_{\bt_3} \times \id_Y) \circ (\mu^S \circ Sf_1 \circ \lstrength{S}{\itp{A}{\bt_1}{S}, \itp{A}{\bt_2}{S}} \circ \langle g_1, g_2 \rangle)^\dagger \\
      &= T(\rho_{\bt_3} \times \id_Y) \circ \mu \circ T(\ev{T(\itp{A}{\bt_3}{S} \times Y)}{Y}) \circ T(f_1 \times \id_Y) \circ T(\lstrength{L}{\itp{A}{\bt_1}{S}, \itp{A}{\bt_2}{S}})  \\
      &\qquad \circ \lstrength{T}{\itp{A}{\bt_1}{S}, \itp{A}{\bt_2}{S} \times Y} \circ \id \times \ev{T(\itp{A}{\bt_2}{S} \times Y)}{Y} \circ \langle g_1, g_2 \rangle \times \id_Y \\ 
      &= \mu \circ T^2(\rho_{\bt_3} \times \id_Y) \circ T(f_1^\dagger) \circ T(\lstrength{L}{\itp{A}{\bt_1}{S}, \itp{A}{\bt_2}{S}}) \circ \lstrength{T}{\itp{A}{\bt_1}{S}, \itp{A}{\bt_2}{S} \times Y} \circ \id \times \ev{T(\itp{A}{\bt_2}{S} \times Y)}{Y} \circ \langle g_1, g_2 \rangle \times \id_Y \\ 
      &= \mu \circ T(f_2 \circ \rho_{\bt_1 \times (\bt_2 \times b^Y)}) \circ \lstrength{T}{\itp{A}{\bt_1}{S}, \itp{A}{\bt_2}{S} \times Y} \circ \id \times \ev{T(\itp{A}{\bt_2}{S} \times Y)}{Y} \circ \langle g_1, g_2 \rangle \times \id_Y \\ 
      &= \mu \circ T(f_2) \circ \lstrength{T}{\litp{A}{\bt_1}{T}{\_ \times Y}, \itp{A}{\bt_2}{T}{\_ \times Y} \times Y} \circ  \rho_{\bt_1} \times T(\rho_{\bt_2} \times Y) \circ \langle g_1 \pi_2, g_2^\dagger\rangle.
    \end{align*}
    Therefore, we have $\mu \circ U(f) \circ \lstrength{U}{}\colon V(\bt_1) \times C(\bt_2) \darrow C(\bt_3)$.
    \item It holds because $!_{\litp{A}{\bt_1}{T}{\_ \times Y}} = \rho_{\bt_1} \circ !_{\itp{A}{\bt}{S}}$, $\rho_{\bt_1 + \bt_2} \circ \iota_i = \iota_i \circ \rho_{\bt_i}$, and $\rho_{(\Gamma \times \bt_1) + (\Gamma \times \bt_2)} \circ (\cong) = (\cong') \circ \rho_{\Gamma \times (\bt_1 + \bt_2)}$
    where $\cong\colon \itp{A}{\Gamma}{S} \times (\itp{A}{\bt_1}{S}+\itp{A}{\bt_2}{S}) \to \itp{A}{\Gamma, \bt_1}{S}+\itp{A}{\Gamma, \bt_2}{S}$
    and $\cong'\colon \litp{A}{\Gamma}{T}{\_ \times Y} \times (\litp{A}{\bt_1}{T}{\_ \times Y}+\litp{A}{\bt_2}{T}{\_ \times Y}) \to \litp{A}{\Gamma, \bt_1}{T}{\_ \times Y}+\litp{A}{\Gamma, \bt_2}{T}{\_ \times Y}$.

    Assume $(f_{i, 1}, f_{i, 2})\colon \dot{L}V(\bt_i) \darrow C(\bt)$ for $i \in \{1, 2\}$.
    Since $(\id, \rho_{\bt_i} \times \id_Y) \in \dot{L}V(\bt_i)(\litp{A}{\bt_i}{U}{L})$,
    we have $T(\rho_\bt \times Y) \circ (f_{i, 1})^\dagger = f_{i, 2} \circ \rho_{\bt_i} \times \id_Y$.
    Therefore, for each $H \in \mathbb{C}$ and $f\colon H \to \litp{A}{\bt_1 + \bt_2}{U}{L}$,
    \begin{align*}
    &[f_{1, 2}, f_{2, 2}] \circ \natcoprod{} \circ \rho_{\bt_1 + \bt_2} \times \id_Y \circ f \times \id_Y \\ 
    &= [f_{1, 2}, f_{2, 2}] \circ (\rho_{\bt_1} \times \id_Y) + (\rho_{\bt_1} \times \id_Y) \circ \natcoprod{} \circ f \times \id_Y \\ 
    &= T(\rho_\bt \times \id_Y) \circ [f_{1, 1}, f_{2, 1}]^\dagger \circ f \times \id_Y.
    \end{align*}
    It yields that $[f_1, f_2] \circ \natcoprod{\litp{A}{\bt_1}{U}{L}, \litp{A}{\bt_2}{U}{L}}\colon \dot{L}V(\bt_1 + \bt_2) \darrow C(\bt)$.
    \item
      The condition $(f_1, f_2)\colon V(\Gamma) \times \dot{L}V(\bt_1) \darrow C(\bt_2)$ is equivalent to $T(\rho_{\bt_2} \times \id_Y) \circ f_1^\ddagger = f_2  \circ \rho_{\Gamma, \bt_1} \times \id_Y\colon \itp{A}{\Gamma, \bt_1}{S} \times Y \to T(\litp{A}{\bt_2}{T}{\_ \times Y} \times Y)$ where $f_1^\ddagger$ is the transpose of $f_1$ for $(\_ \times Y) \dashv (\_)^Y$.
      Similarly, $(f_1^\dagger, f_2^\dagger)\colon V(\Gamma) \darrow V(\bt_1 \to \bt_2)$ is equivalent to 
      $\rho_{\bt_1 \to \bt_2} \circ f_1^\dagger = f_2^\dagger  \circ \rho_{\Gamma}\colon \itp{A}{\Gamma}{S} \to T(\litp{A}{\bt_2}{T}{\_ \times Y} \times Y)^{\litp{A}{\bt_1}{T}{\_ \times Y} \times Y}$,
      which is also equivalent to
      $\rho_{\bt_1 \to \bt_2}^\dagger \circ f_1^\dagger \times \id = f_2 \circ \rho_\Gamma \times \id$
      where it follows 
      by taking the transpose with respect to $(\_ \times (\litp{A}{\bt_1}{T}{\_ \times Y} \times Y)) \dashv (\_)^{\litp{A}{\bt_1}{T}{\_ \times Y} \times Y}$.
      Then the following equations conclude the proof:
      \begin{align*}
        &\rho_{\bt_1 \to \bt_2}^\dagger \circ f_1^\dagger \times \id  \circ (\lstrength{\_ \times Y}{})^{-1} \circ \id \times (\rho_{\bt_1} \times \id_Y)\\
        &= \ev{T(\litp{A}{\bt_2}{T}{\_ \times Y} \times Y)}{Y} \circ \ev{S\litp{A}{\bt_2}{T}{\_ \times Y}}{\litp{A}{\bt_1}{T}{\_ \times Y}} \times \id_Y \circ (S\rho_{\bt_2})^{\rho_{\bt_1}^{-1}} \circ \id \times \rho_{\bt_1} \times \id_Y \circ f_1^\dagger \times \id \\
        &= \ev{T(\litp{A}{\bt_2}{T}{\_ \times Y} \times Y)}{Y} \circ S(\rho_{\bt_2}) \times \id_Y \circ \ev{S\itp{A}{\bt_2}{S}}{\itp{A}{\bt_1}{S}} \times \id_Y \circ f_1^\dagger \times \id \\
        &= T(\rho_{\bt_2} \times \id_Y) \circ \ev{T(\itp{A}{\bt_2}{S} \times Y)}{Y} \circ \ev{S\itp{A}{\bt_2}{S}}{\itp{A}{\bt_1}{S}} \times \id_Y \circ f_1 \times \id \circ \eta^{(\_ \times (\litp{A}{\bt_1}{T}{\_ \times Y} \times Y))^{\litp{A}{\bt_1}{T}{\_ \times Y} \times Y}} \times \id \\
        &= T(\rho_{\bt_2} \times \id_Y) \circ f_1^\ddagger.
      \end{align*}
  \end{enumerate}
\end{proof}
Letting $\Gamma\coloneqq \bt_1 \to \bt_2$ and $f\coloneqq \id$ in \cref{lem:liftcv}.\ref{item:exp},
we have $V(\bt_1 \to \bt_2) \subseteq C(\bt_2)^{\dot{L}V(\bt_1)}$,
although the converse does not hold in general.

\begin{proof}[Proof of \cref{thm:mainfib}]
 We prove it by induction on the structure of $M$. 
 \begin{description}
  \item[Case $\varintro$.] 
    For $\Gamma\vdash x\colon \bt_i$, 
    $\eta_{L\litp{\tstructure}{\bt_i}{U}{L}} \circ L\pi_i\colon \dot{L}V(\Gamma) \darrow C(\bt_i)$ by \cref{lem:liftcv}.\ref{item:eta}.
  \item[Case $\exchange$.] Easy.
  \item[Case $\conintro$.] It is easy by \cref{lem:liftcv}.\ref{item:ce}.
  \item[Case $\genintro$.] It is easy by \cref{lem:liftcv}.\ref{item:ce} and \ref{item:bind}.
  \item[Case $\unitprod$.] It is easy by \cref{lem:liftcv}.\ref{item:eta}.
  \item[Case $\prodintro$.] 
    By \cref{lem:liftcv}.\ref{item:liftnat} and \ref{item:bind}, $\litp{A}{N}{U}{L}\colon \dot{L}V(\Gamma) \darrow C(\bt_2)$
    implies $U\lstrength{L}{\litp{A}{\bt_1}{U}{L}, \litp{A}{\bt_2}{U}{L}} \circ \lstrength{U}{\litp{A}{\bt_1}{U}{L}, \litp{A}{\bt_2}{U}{L}} \circ \id \times \litp{A}{N}{U}{L} \circ \langle \pi_1 \natdisc{\litp{A}{\bt_1}{U}{L} \times \litp{A}{\Gamma}{U}{L}}, L\pi_2 \rangle\colon \dot{L}V(\bt_1 \times \Gamma) \darrow C(\bt_1 \times \bt_2)$, which further implies
    $\mu \circ 
    U^2(\lstrength{L}{\litp{A}{\bt_1}{U}{L}, \litp{A}{\bt_2}{U}{L}}) \circ U(\lstrength{U}{\litp{A}{\bt_1}{U}{L}, \litp{A}{\bt_2}{U}{L}}) \circ U(\id \times \litp{A}{N}{U}{L}) \circ U(\langle \pi_1 \natdisc{\litp{A}{\bt_1}{U}{L} \times \litp{A}{\Gamma}{U}{L}}, L\pi_2 \rangle)\colon C(\bt_1 \times \Gamma) \darrow C(\bt_1 \times \bt_2)$ by \cref{lem:liftcv}.\ref{item:bind}.

    Then the statement holds by \cref{lem:liftcv}.\ref{item:liftnat} and \ref{item:bind}.
  \item[Case $\projintroone$, $\projintrotwo$.]
    By \cref{lem:liftcv}.\ref{item:eta}, 
    we know that $\eta \circ \pi_i\colon \dot{L}V(\bt_1 \times \bt_2) \darrow C(\bt_i)$,
    which implies $UL\pi_i \colon C(\bt_1 \times \bt_2) \darrow C(\bt_i)$
    by \cref{lem:liftcv}.\ref{item:bind}.
    Therefore, 
    $UL\pi_i \circ \litp{A}{M}{U}{L}\colon \dot{L}V(\Gamma) \darrow C(\bt_i)$.
  \item[Case $\unitcoprod$, $\coprodone$, $\coprodtwo$.] It is easy by \cref{lem:liftcv}.\ref{item:eta}, \ref{item:bind}, and \ref{item:coprod}.
  \item[Case $\coprodelm$.]
    By \cref{lem:liftcv}.\ref{item:liftnat}, we obtain the morphism $\langle \natdisc{\litp{A}{\Gamma}{T}{L}}, \litp{A}{M}{T}{L} \rangle\colon L\litp{A}{\Gamma}{T}{L} \darrow \litp{A}{\Gamma}{T}{L} \times TL(\litp{A}{\bt_1 + \bt_2}{T}{L})$,
    and \cref{lem:liftcv}.\ref{item:coprod} and \ref{item:eta} yield $\mu \circ T([\litp{A}{M_1}{T}{L}, \litp{A}{M_2}{T}{L}]) \circ T(\natcoprod{\litp{A}{\Gamma, \bt_1}{T}{L}, \litp{A}{\Gamma, \bt_2}{T}{L}})\colon C(\Gamma \times \bt_1 + \Gamma \times \bt_2) \darrow C(\bt)$.
    Then the result follows from \cref{lem:liftcv}.\ref{item:bind} and \ref{item:coprod}.
  \item[Case $\abst$.]
    By \cref{lem:liftcv}.\ref{item:liftnat} and \ref{item:exp},
    it follows that $(\litp{A}{M}{T}{L} \circ \lstrength{L}{})^\dagger\colon V(\Gamma) \darrow V(\bt_1 \to \bt_2)$.
    Then the result follows from \cref{lem:liftcv}.\ref{item:eta}.
  \item[Case $\appl$.] 
    \cref{lem:liftcv}.\ref{item:exp} and $\ref{item:bind}$ imply $\mu \circ T(\ev{L\litp{A}{\bt_1}{T}{L}}{TL\litp{A}{\bt_2}{T}{L}}) \circ \lstrength{T}{TL\litp{A}{\bt_2}{T}{L}^{L\litp{A}{\bt_1}{T}{L}}, L\litp{A}{\bt_1}{T}{L}}\colon V(\bt_1 \to \bt_2) \times C(\bt_1) \darrow C(\bt_2)$.
    Then the result follows from \cref{lem:liftcv}.\ref{item:liftnat} and \ref{item:bind}.
 \end{description}
\end{proof}

\subsection{Proof of \cref{prop:transLambdacfix}} \label{ap:proof_transLambdacfix}
\begin{lemma} \label{lem:fixpi}
  For any morphism $f\colon X\times T(Z)\rightarrow T(Z)$,
  $\fixop{X \times X'}{Z}{f \circ (\pi_1 \times \id_{TZ})} = \fixop{X}{Z}{f} \circ \pi_1$.
\end{lemma}
\begin{proof}
  Since the following diagram commutes,
  \begin{displaymath}
    \xymatrix{
      X \times X' \ar[d]_{\pi_1} \ar[r]^-{\langle \id, \bot_Z\rangle} &X \times X' \times TZ \ar[d]_{\pi_1 \times \id} \ar[r]^-{\langle \pi_1, f \circ (\pi_1 \times \id) \rangle} &X \times X' \times TZ \ar[d]_-{\pi_1 \times \id} \ar[r]^-{\pi_2} &TZ \\
      X \ar[r]_-{\langle \id, \bot_Z \rangle} &X \times TZ \ar[r]_-{\langle \pi_1, f \rangle} &X \times TZ \ar[ur]_-{\pi_2}
    }
  \end{displaymath}
  we have
  \begin{align*}
    \fixop{X \times X'}{Z}{f \circ (\pi_1 \times \id_{TZ})} 
    &= \pi_2\circ \bigvee_{n\in \nat} \langle\pi_1, f \circ (\pi_1 \times \id_{TZ}) \rangle^{n}\circ \langle \id_{X \times X'}, \leastelement{T}{Z} \rangle \\
    &= \pi_2\circ \bigvee_{n\in \nat} (\pi_1 \times \id) \circ \langle\pi_1, f \circ (\pi_1 \times \id_{TZ}) \rangle^{n}\circ \langle \id_{X \times X'}, \leastelement{T}{Z} \rangle \\
    &= \pi_2\circ \bigvee_{n\in \nat} \langle\pi_1, f \rangle^{n}\circ (\pi_1 \times \id) \circ \langle \id_{X \times X'}, \leastelement{T}{Z} \rangle \\
    &= \pi_2\circ \bigvee_{n\in \nat} \langle\pi_1, f \rangle^{n}\circ \langle \id_{X}, \leastelement{T}{Z} \rangle \circ \pi_1 
    = \fixop{X}{Z}{f} \circ \pi_1.
  \end{align*}
\end{proof}

\begin{proof}[Proof of \cref{prop:transLambdacfix}]
    It follows that 
    \begin{align*}
        &\itp{\mathcal{A}}{\lambda z.\ \uncurry{M}[\pi_1 z/x, \pi_2 z/y]}{T} \\
        &= \eta \circ \itp{\mathcal{A}}{\uncurry{\Gamma}, f\colon \uncurry{\bt_1} \to \uncurry{\bt_2}, z\colon \uncurry{\bt_2} \vdash \uncurry{M}[\pi_1 z/x, \pi_2 z/y]}{T}^\dagger  \\
        &= \eta \circ (\itp{\mathcal{A}}{\primetrans{\Gamma}, f\colon \uncurry{\bt_1} \to \uncurry{\bt_2}, x\colon \primetrans{\bt_1}, y\colon b^Y \vdash \uncurry{M}}{T} \circ \lstrength{\_ \times Y}{} \circ \pi_1 \times \id_{(T\itp{A}{\uncurry{\bt_2}}{T}^{\itp{A}{\uncurry{\bt_1}}{T}} \times \itp{A}{\bt_1 \times b^Y}{T})})^\dagger
    \end{align*}
    \begin{align*}
    &\itp{\tstructure}{\uncurry{(\mu f x.~M)}}{T}  \\
    &= \itp{\tstructure}{(\mu f z.\, \uncurry{M}[\pi_1 z/x, \pi_2 z/y], y)}{T} \\
    &= \mu \circ T(\lstrength{T}{T\itp{A}{\uncurry{\bt_2}}{T}^{\itp{A}{\uncurry{\bt_1}}{T}}, Y}) \circ \rstrength{T}{T\itp{A}{\uncurry{\bt_2}}{T}^{\itp{A}{\uncurry{\bt_1}}{T}}, TY} \circ \langle \itp{A}{\mu f z.~\uncurry{M}[\pi_1 z/x, \pi_2 z/y]}{T}, \eta \pi_2\rangle \\
    &= \lstrength{T}{T\itp{A}{\uncurry{\bt_2}}{T}^{\itp{A}{\uncurry{\bt_1}}{T}}, Y} \circ \langle \itp{A}{\mu f z.~\uncurry{M}[\pi_1 z/x, \pi_2 z/y]}{T}, \pi_2\rangle \\
    &= \eta \circ \langle \cal{T}{\itp{\mathcal{A}}{\uncurry{\bt_1}}{T}}{\itp{\mathcal{A}}{\uncurry{\bt_2}}{T}}\circ \fixop{\itp{\mathcal{A}}{\uncurry{\Gamma}}{T}}{\itp{\mathcal{A}}{\uncurry{\bt_1}\rightarrow \uncurry{\bt_2}}{T}}{\itp{\mathcal{A}}{\lambda z.\ \uncurry{M}[\pi_1 z/x, \pi_2 z/y]}{T} \circ \id_{\itp{\mathcal{A}}{\uncurry{\Gamma}}{T}}\times \cal{T}{\itp{\mathcal{A}}{\uncurry{\bt_1}}{T}}{\itp{\mathcal{A}}{\uncurry{\bt_2}}{T}}}, \pi_2\rangle \\
    &= \eta \circ \langle \cal{T}{\itp{\mathcal{A}}{\uncurry{\bt_1}}{T}}{\itp{\mathcal{A}}{\uncurry{\bt_2}}{T}}\circ \fixop{\itp{\mathcal{A}}{\uncurry{\Gamma}}{T}}{\itp{\mathcal{A}}{\uncurry{\bt_1}\rightarrow \uncurry{\bt_2}}{T}}{\eta \circ (\itp{\mathcal{A}}{\uncurry{M}}{T} \circ \lstrength{\_ \times Y}{} \circ \pi_1 \times \id)^\dagger \circ \id_{\itp{\mathcal{A}}{\uncurry{\Gamma}}{T}}\times \cal{T}{\itp{\mathcal{A}}{\uncurry{\bt_1}}{T}}{\itp{\mathcal{A}}{\uncurry{\bt_2}}{T}}}, \pi_2\rangle \\
    &= \eta \circ \langle \cal{T}{\itp{\mathcal{A}}{\uncurry{\bt_1}}{T}}{\itp{\mathcal{A}}{\uncurry{\bt_2}}{T}}\circ \fixop{\itp{\mathcal{A}}{\uncurry{\Gamma}}{T}}{\itp{\mathcal{A}}{\uncurry{\bt_1}\rightarrow \uncurry{\bt_2}}{T}}{\eta \circ (\itp{\mathcal{A}}{\uncurry{M}}{T} \circ \lstrength{\_ \times Y}{})^\dagger \circ \pi_1 \times \cal{T}{\itp{\mathcal{A}}{\uncurry{\bt_1}}{T}}{\itp{\mathcal{A}}{\uncurry{\bt_2}}{T}}}, \pi_2\rangle \\
    &= \eta \circ \langle \cal{T}{\itp{\mathcal{A}}{\uncurry{\bt_1}}{T}}{\itp{\mathcal{A}}{\uncurry{\bt_2}}{T}}\circ \fixop{\itp{\mathcal{A}}{\primetrans{\Gamma}}{T}}{\itp{\mathcal{A}}{\uncurry{\bt_1}\rightarrow \uncurry{\bt_2}}{T}}{\eta \circ (\itp{\mathcal{A}}{\uncurry{M}}{T} \circ \lstrength{\_ \times Y}{})^\dagger \circ \id \times \cal{T}{\itp{\mathcal{A}}{\uncurry{\bt_1}}{T}}{\itp{\mathcal{A}}{\uncurry{\bt_2}}{T}}} \circ \pi_1, \pi_2\rangle \text{ by }\cref{lem:fixpi} \\
    &= \eta \circ  (\cal{T}{\litp{\mathcal{A}}{\bt_1}{T}{\_ \times Y} \times Y}{\litp{\mathcal{A}}{\bt_2}{T}{\_ \times Y} \times Y}\circ \fixop{\litp{\mathcal{A}}{\Gamma}{T}{\_ \times Y}}{\litp{\mathcal{A}}{\bt_1\rightarrow \bt_2}{T}{\_ \times Y}}{
      \eta \circ (\litp{\mathcal{A}}{M}{T}{\_ \times Y} \circ \lstrength{\_ \times Y}{})^\dagger \circ \id \times \cal{T}{\litp{\mathcal{A}}{\bt_1}{T}{\_ \times Y} \times Y}{\litp{\mathcal{A}}{\bt_2}{T}{\_ \times Y} \times Y}}) \times Y \\
    &= \litp{\tstructure}{\mu f x.~M}{T}{\_ \times Y}.
    \end{align*}
  \end{proof}

  \subsection{Proof of \cref{thm:mainfibfix}} \label{ap:proof_mainfibfix}
  We use the same notations as in \cref{ap:proof_liftlambdac}.

\begin{lemma}
  \label{lem:commuteCanAl}
  The following diagram commutes: 

   \adjustbox{scale=0.6,center}{
      \begin{tikzcd}
        T\Big(T(\itp{A}{\uncurry{\bt_2}}{T})^{\itp{A}{\uncurry{\bt_1}}{T}}\Big)\arrow[d, " T(\rho^{-1})"]\arrow[rrrrrr, "\cal{T}{\itp{A}{\uncurry{\bt_1}}{T}}{\itp{A}{\uncurry{\bt_2}}{T}}"] &&&&&& T(\itp{A}{\uncurry{\bt_2}}{T})^{\itp{A}{\uncurry{\bt_1}}{T}} \arrow[d, "\rho^{-1}"]\\
        T\big(S(\itp{A}{\bt_2}{S})^{\itp{A}{\bt_1}{S}}\big)\arrow[rr, "\eta"]&& \Big(T\big(S(\itp{A}{\bt_2}{S})^{\itp{A}{\bt_1}{S}}\big)\times Y\Big)^Y \arrow[rr, "(\rstrength{T}{})^{\id}"]&&  S\big(S(\itp{A}{\bt_2}{S})^{\itp{A}{\bt_1}{S}}\big) \arrow[rr, "\cal{S}{\itp{A}{\bt_1}{S}}{\itp{A}{\bt_2}{S}}"] && S(\itp{A}{\bt_2}{S})^{\itp{A}{\bt_1}{S}}
      \end{tikzcd}
    }
\end{lemma}
\begin{proof}
  By the following diagram:

  \adjustbox{scale=0.4,center}{
    \begin{tikzcd}
        T\Big(T(\itp{A}{\uncurry{\bt_2}}{T})^{\itp{A}{\uncurry{\bt_1}}{T}}\Big)\arrow[d, " T(\rho^{-1})"]\arrow[rr, "\eta"] &&  \Big(T\Big(T(\itp{A}{\uncurry{\bt_2}}{T})^{\itp{A}{\uncurry{\bt_1}}{T}}\Big)\times \itp{A}{\uncurry{\bt_1}}{T}\Big)^{\itp{A}{\uncurry{\bt_1}}{T}} \arrow[rr, "(\rstrength{}{})^{\id}"]&& T\Big(T(\itp{A}{\uncurry{\bt_2}}{T})^{\itp{A}{\uncurry{\bt_1}}{T}}\times \itp{A}{\uncurry{\bt_1}}{T}\Big)^{\itp{A}{\uncurry{\bt_1}}{T}} \arrow[rr, "\evsyb^{\id}"]&& T^2(\itp{A}{\uncurry{\bt_2}}{T})^{\itp{A}{\uncurry{\bt_1}}{T}} \arrow[dd, "\mu^{\id}"]\\
        T\Big(T(\itp{A}{\bt_2}{S}\times Y)^{\itp{A}{\bt_1}{S}\times Y}\Big) \arrow[d, "T(\cong)"] \arrow[rr, "\eta"]&& \Big(T\Big(T(\itp{A}{\bt_2}{S}\times Y)^{\itp{A}{\bt_1}{S}\times Y}\Big)\times \itp{A}{\bt_1}{S}\times Y\Big)^{\itp{A}{\bt_1}{S}\times Y} \arrow[rr, "(\rstrength{}{})^{\id}"] \arrow[d, ""]&& T\Big(T(\itp{A}{\bt_2}{S}\times Y)^{\itp{A}{\bt_1}{S}\times Y}\times \itp{A}{\bt_1}{S}\times Y\Big)^{\itp{A}{\bt_1}{S}\times Y} \arrow[rru, ""] \arrow[d, "T(\evsyb)^{\id}"] \arrow[llllddddd, bend left = 10, ""]&& \\
        T\Big(S(\itp{A}{\bt_2}{S})^{\itp{A}{\bt_1}{S}}\Big) \arrow[d, "\eta"] \arrow[rr, "\eta"]&& \Big(T\Big(S(\itp{A}{\bt_2}{S})^{\itp{A}{\bt_1}{S}}\Big)\times \itp{A}{\bt_1}{S}\times Y \Big)^{\itp{A}{\bt_1}{S}\times Y} \arrow[rru, ""] \arrow[ddddll, bend left = 15, ""] && T^2(\itp{A}{\bt_2}{S}\times Y)^{\itp{A}{\bt_1}{S}\times Y}  \arrow[drr, "\mu^{\id}"] \arrow[ddddd, "\id"] &&  T(\itp{A}{\uncurry{\bt_2}}{T})^{\itp{A}{\uncurry{\bt_1}}{T}} \arrow[d, ""] \\
        \Big(T\Big(S(\itp{A}{\bt_2}{S})^{\itp{A}{\bt_1}{S}}\Big)\times Y\Big)^{Y} \arrow[d, "(\rstrength{}{})^{\id}"]&& && && T(\itp{A}{\bt_2}{S}\times Y)^{\itp{A}{\bt_1}{S}\times Y} \arrow[dddd, "\cong"]\\
        S\Big(S(\itp{A}{\bt_2}{S})^{\itp{A}{\bt_1}{S}}\Big) \arrow[d, "\eta"]&& && && \\
        \Big(S\Big(S(\itp{A}{\bt_2}{S})^{\itp{A}{\bt_1}{S}}\Big) \times \itp{A}{\bt_1}{S}\Big)^{\itp{A}{\bt_1}{S}} \arrow[d, "(\rstrength{}{})^{\id}"] && && && \\
        S\Big(S(\itp{A}{\bt_2}{S})^{\itp{A}{\bt_1}{S}} \times \itp{A}{\bt_1}{S}\Big)^{\itp{A}{\bt_1}{S}} \arrow[d, "S(\evsyb)^{\id}"] && && && \\
        S^2(\itp{A}{\bt_2}{S})^{\itp{A}{\bt_1}{S}} \arrow[rrrr, ""]&& && T^2(\itp{A}{\bt_2}{S}\times Y)^{\itp{A}{\bt_1}{S}\times Y} \arrow[rr, ""]&& S(\itp{A}{\bt_2}{S})^{\itp{A}{\bt_1}{S}}
      \end{tikzcd}
    }

\end{proof}
  \begin{definition} \label{def:admissible}
    Let $p\colon \mathbb{E} \to \mathbb{C}$ be a fibration for logical relations.
    We say that an object $X \in \mathbb{E}_{TI}$ is \emph{admissible} 
    if (i) $\bot_I\colon 1 \darrow X$ and 
    (ii) for any $Y \in \mathbb{E}$ and an $\omega$-chain $f_0 \leq f_1 \leq \cdots $ in $\mathbb{C}(pY, TI)$, $\bigvee_{n \in \mathbb{N}} f_n\colon Y \darrow X$.
  \end{definition} 
  \begin{lemma} \label{lem:admissible_lift}
    Let $p\colon \mathbb{E} \to \mathbb{C}$ be a fibration for logical relations,
    and $X$ be an admissible object in $\mathbb{E}_{TI}$.
    If $f\colon Z \times X \darrow X$, then $\fixop{pZ}{I}{f}\colon Z \darrow X$.
    \qed
  \end{lemma}
  \begin{lemma} \label{lem:admissible_lift_al}
    For any types $\bt_1, \bt_2$, 
    $(\cal{U}{L\litp{\mathcal{A}}{\bt_1}{U}{L}}{L\litp{\mathcal{A}}{\bt_2}{U}{L}})^*V(\bt_1 \to \bt_2)$ is admissible.
  \end{lemma}
  \begin{proof}
    For the condition (i) in \cref{def:admissible},
    it is enough to show that for each $H \in \mathbb{C}$, the following diagram commutes.
    \begin{displaymath}
      \xymatrix@C=5em{
        H \ar[r]^{!} \ar[d]^{\pi_1} &1 \ar[r]^-\bot &S(S(\itp{A}{\bt_2}{S})^{\itp{A}{\bt_1}{S}}) \ar[r]^{\cal{S}{\itp{\mathcal{A}}{\bt_1}{S}}{\itp{\mathcal{A}}{\bt_2}{S}}}  &S(\itp{A}{\bt_2}{S})^{\itp{A}{\bt_1}{S}} \\
        H \times Y \ar[r]^{!} &1 \ar[u]^{\id}\ar[r]^-\bot &T(T(\itp{A}{\uncurry{\bt_2}}{T})^{\itp{A}{\uncurry{\bt_1}}{S}}) \ar[u]^{\alpha \circ S\rho_{\bt_1 \to \bt_2}}\ar[r]^{\cal{T}{\itp{\mathcal{A}}{\uncurry{\bt_1}}{T}}{\itp{\mathcal{A}}{\uncurry{\bt_2}}{T}}} &T(\itp{A}{\uncurry{\bt_2}}{T})^{\itp{A}{\uncurry{\bt_1}}{T}}  \ar[u]_{\rho_{\bt_1 \to \bt_2}}
      }
    \end{displaymath}
    The middle square commutes because $\alpha$ preserves $\bot$,
    and the right square commutes by \cref{lem:commuteCanAl}.
    Noting that $\eta \circ \cal{U}{L\litp{A}{\bt_1}{U}{L}}{L\litp{A}{\bt_2}{U}{L}} = \id$, the condition (ii) in \cref{def:admissible} can be easily checked.
  \end{proof}
  These lemma above yield the following result.
  \begin{lemma} \label{lem:lift_fix_al}
    If $f\colon V(\Gamma) \times V(\bt_1 \to \bt_2) \darrow V(\bt_1 \to \bt_2)$,
    then $\eta \circ \cal{U}{L\litp{A}{\bt_1}{U}{L}}{L\litp{A}{\bt_2}{U}{L}} \circ \fixop{\litp{\mathcal{A}}{\Gamma}{T}{\_ \times Y}}{\litp{A}{\bt_1 \to \bt_2}{U}{L}}{\eta \circ f \circ \id \times \cal{U}{L\litp{A}{\bt_1}{U}{L}}{L\litp{A}{\bt_2}{U}{L}}}\colon V(\Gamma) \darrow C(\bt_1 \to \bt_2)$.
  \end{lemma}
  \begin{proof}
    Since $\cal{U}{L\litp{A}{\bt_1}{U}{L}}{L\litp{A}{\bt_2}{U}{L}}$ is an Eilenberg-Moore algebra,
    $\eta \circ \cal{U}{L\litp{A}{\bt_1}{U}{L}}{L\litp{A}{\bt_2}{U}{L}} = \id$, which implies $\eta\colon V(\bt_1 \to \bt_2) \darrow (\cal{U}{L\litp{\mathcal{A}}{\bt_1}{U}{L}}{L\litp{\mathcal{A}}{\bt_2}{U}{L}})^*V(\bt_1 \to \bt_2)$.
    Therefore, we have $\eta \circ f \circ (\id \times \cal{U}{L\litp{A}{\bt_1}{U}{L}}{L\litp{A}{\bt_2}{U}{L}})\colon V(\Gamma) \times (\cal{U}{L\litp{\mathcal{A}}{\bt_1}{U}{L}}{L\litp{\mathcal{A}}{\bt_2}{U}{L}})^*V(\bt_1 \to \bt_2) \darrow (\cal{U}{L\litp{\mathcal{A}}{\bt_1}{U}{L}}{L\litp{\mathcal{A}}{\bt_2}{U}{L}})^*V(\bt_1 \to \bt_2)$.
    By \cref{lem:admissible_lift} and \cref{lem:admissible_lift_al}, it follows that 
    $\fixop{\litp{\mathcal{A}}{\Gamma}{T}{\_ \times Y}}{\litp{A}{\bt_1 \to \bt_2}{U}{L}}{\eta \circ f \circ \id \times \cal{U}{L\litp{A}{\bt_1}{U}{L}}{L\litp{A}{\bt_2}{U}{L}}}\colon V(\Gamma) \darrow (\cal{U}{L\litp{\mathcal{A}}{\bt_1}{U}{L}}{L\litp{\mathcal{A}}{\bt_2}{U}{L}})^*V(\bt_1 \to \bt_2)$.
    Then the statement follows from \cref{lem:liftcv}.\ref{item:eta}.
  \end{proof}
  \begin{proof}[Proof of \cref{thm:mainfibfix}]
    It is enough to show that $\litp{A}{M}{U}{L}\colon \dot{L}V(\Gamma, \bt_1 \to \bt_2, \bt_1) \to C(\bt_2)$ implies 
    $\litp{A}{\mu f x.~M}{U}{L}\colon \dot{L}V(\Gamma) \to C(\bt_1 \to \bt_2)$.

    By \cref{lem:liftcv}.\ref{item:liftnat} and \ref{item:exp},
    we have $(\litp{\mathcal{A}}{M}{T}{\_ \times Y} \circ \lstrength{\_ \times Y}{})^\dagger \colon V(\Gamma) \times V(\bt_1 \to \bt_2) \darrow V(\bt_1 \to \bt_2)$.
    Then \cref{lem:lift_fix_al}, \cref{lem:liftcv}.\ref{item:eta} concludes the proof.
  \end{proof}

\section{Omitted Definitions and Proofs in~\cref{sec:caseStudy}}

\subsection{Preliminaries}
\label{subsec:preliminaries}

\begin{definition}[DCPO]
  A poset $X$ is a \emph{dcpo} if all directed family $D$ has a supremum $\bigvee D$.
\end{definition}
We fix the base bi-CCC $\mathbb{C}$ as $\dCPO$. Note that $\dCPO$ is stable. 

\begin{definition}[Scott topology]
  Given a dcpo $X$, the \emph{Scott topology} is the topology whose open sets are sets $U\subseteq X$ such that 1) for any $u_1\in U$ and $u_2\in X$, $u_1\preceq_X u_2$ implies $u_2\in U$;
  and 2) for any directed family $D$, $\vee D\in U$ implies $D_{\lambda}\in U$ for some index $\lambda$.
  We write $\scottTop{X}$ for the set of open sets in the Scott topology of $X$.
  Note that $\scottTop{X}$ is a dcpo with the inclusion ordering.
\end{definition}

Scott-continuous maps are precisely continuous maps in Scott topology~\cite{Berger07}.

\begin{definition}[d-topology]
  Given a dcpo $X$, the \emph{d-topology} is the topology whose closed sets consist of sub-dcpos of $X$.

\end{definition}
\begin{lemma}[\cite{ZhaoF10}]
  Scott-continuous maps are continuous maps in d-topology.
\end{lemma}


\subsubsection{Commutative Probabilistic Monad}
\label{sec:appendixProbMonad}
We recall the definition of integration for continuous valuations and lower semicontinuous functions; see \cite{Goubault-Larrecq19,GoubaultLarrecqJT23} as references.

\begin{definition}[continuous valuation]
  Given a dcpo $X$, a \emph{continuous valuation} is a Scott continuous function $\nu\colon\scottTop{X}\rightarrow [0, \infty]$ such that
  1) $\nu(\emptyset) = 0$; 2) $\nu(U) + \nu(V) = \nu(U\cup V) + \nu(U\cap V)$.
\end{definition}

\begin{definition}[lower semicontinuous]
  Given a dcpo $X$, a function $f\colon X\rightarrow [0, \infty]$ is \emph{lower semicontinuous} (l.s.c.)
  if $f^{-1}\big((r, \infty]\big)$ is open (w.r.t. the Scott topology) for any $r\geq 0$.
\end{definition}

\begin{definition}[integration]
    Given a dcpo $X$, a continuous valuation $\nu\in\scottTop{X}\rightarrow [0, \infty] $, and a l.s.c. function $f\colon X\rightarrow [0, \infty]$,
    the \emph{integration} $\int f d\nu$ is defined by   $\int f d\nu \defeq \int^{\infty}_0 \nu\big(f^{-1}((t, \infty])\big)dt$,
    where the right integration is the ordinal Riemann integral.
\end{definition}

We list some useful properties related to the integration.

\begin{definition}[Dirac valuation] \label{def:dirac}
  Given $x\in X$, the Dirac valuation $\eta_X(x)$ is the continuous valuation defined by $\eta_X(x)(V) \defeq 1$ if $x\in V$ and  $\eta_X(x)(V) \defeq 0$ otherwise.
\end{definition}

\begin{definition}[simple valuation]
  A continuous valuation $\nu\colon\scottTop{X}\rightarrow [0, \infty]$ is \emph{simple}
  if $\nu$ is a weighted sum of the Dirac valuation $\nu \defeq \sum^{m}_{i=1} w_i \cdot \eta_{X}(x_i)$, where $w_i\in [0, \infty)$ for each $i$.
\end{definition}

\begin{lemma}[\cite{Goubault-Larrecq19}]
  The following statements holds:
\begin{enumerate}
  \item Given a simple valuation $\nu \defeq \sum^{m}_{i=1} w_i \cdot \eta_{X}(x_i)$, and l.s.c. function $f$, the equation $\int f d\nu = \sum^{m}_{i=1} w_i\cdot f(x_i)$ holds.
  \item Given a continuous valuation $\nu\colon\scottTop{X}\rightarrow [0, \infty]$,  l.s.c. functions $f_1, f_2\colon X\rightarrow [0, \infty]$ and $r_1, r_2\in [0, \infty)$, the equation $\int (r_1 f_1 + r_2 f_2)d\nu = r_1 \int f_1 d\nu +  r_2 \int f_2 d\nu$ holds.
  \item Given a continuous valuation $\nu\colon\scottTop{X}\rightarrow [0, \infty]$ and directed family $(f_\lambda)_{\lambda\in \Lambda}$ of l.s.c. functions (in the pointwise order), the equation
        $\int \bigvee_{\lambda\in \Lambda} f_\lambda d\nu = \bigvee_{\lambda\in \Lambda} \int f_{\lambda} d\nu $ holds.
  \item Given $U\in \scottTop{X}$ and a continuous valuation $\nu\colon\scottTop{X}\rightarrow [0, \infty]$, the equation $\int \chi_U d\nu = \nu(U) $ holds, where $\chi_U$ is the characteristic function of $U$.
  \item Given continuous valuations $\nu_1, \nu_2\colon\scottTop{X}\rightarrow [0, \infty]$,  $r_1, r_2\in [0, \infty)$, and l.s.c. function $f\colon X\rightarrow [0, \infty]$,
  the equation $\int f d(r_1\nu_1 + r_2\nu_2) = r_1 \int f d\nu_1 + r_2 \int f d\nu_2$ holds.
\end{enumerate}
\end{lemma}

We then recall a commutative probabilistic monad $\comProbPowerdommonad$ (that contains the simple valuations) on $\dCPO$~\cite{JiaLMZ21}.
\begin{definition}[commutative  probabilistic monad~\cite{JiaLMZ21}]
  The \emph{commutative probabilistic monad} is the  monad $\comProbPowerdommonad$ on $\dCPO$ given by the following Kleisli triple $(\comProbPowerdommonad, \eta, (\_)^{\ast})$, where
  1) $\comProbPowerdommonad(X)$ is the smallest dcpo that contains simple valuations $f\colon \scottTop{X}\rightarrow [0, 1]$; 2) the \emph{unit} $\eta_X\colon X\rightarrow \comProbPowerdommonad(X)$ is defined by the Dirac valuation;
  and 3) the \emph{extension} $f^{\ast}$ for $f\colon X\rightarrow \comProbPowerdommonad(Y)$ is given by $f(\nu)(V) \defeq \int_{x} f(x)(V) d\nu$ for each $V\in \scottTop{Y}$.
\end{definition}
The multiplication $\mu^{\comProbPowerdommonad}_X\colon \comProbPowerdommonad\big(\comProbPowerdommonad(X)\big) \rightarrow \comProbPowerdommonad(X)$ is given by $\mu^{\comProbPowerdommonad}_X(\nu)(V) \defeq \int_{\sigma} \sigma(V) d\nu $.
The monad $\comProbPowerdommonad$ is a strong monad with the strengh $\lstrength{}{X, Y}$ defined by
\begin{align*}
  \lstrength{}{X, Y}(x, \nu) \defeq \lambda U.\int_y \chi_U(x, y) d\nu.
\end{align*}

We use the following lemma to prove that some algebras of $\comProbPowerdommonad$ are indeed Eilenberg-Moore algebras: see~\cite{KeimelP16,JiaLMZ21} for the details.
\begin{lemma}[\cite{JiaLMZ21}]
  \label{lem:KegelspitzeEilenberg}
  Given a continuous Kegelspitze $K$, the linear barycenter map $\tau\colon \comProbPowerdommonad(K)\rightarrow K$ is an Eilenberg-Moore algebra.
\end{lemma}

\subsubsection{Monad for Angelic Non-Determinism}
We define a lower monad based on an adjunction~\cite{Jacobs17}, which coincides with the extended Hoare powerdomain monad that includes emptyset as an least element~\cite{AbramskyJ95}.

\begin{definition}
  The \emph{lower monad} $\hoaremonad$ is defined by an adjunction (\cite{Jacobs17}\footnote{Jacobs defines an adjunction from a subcategory of $(\cljoin)^{\op}$, but it can be naturally extended to $(\cljoin)^{\op}$.}) $\mathcal{O} \dashv (\_)^{\op} \colon (\cljoin)^{\op}\rightarrow \dCPO$, where
  $\cljoin$ is the category of complete lattices and join preserving maps, $\op$ is a functor that inverses the order, and $\mathcal{O}\colon \dCPO\rightarrow (\cljoin)^{\op}$ is an functor that maps a dcpo to its Scott-open set $\mathcal{O}(X)$.
\end{definition}

Indeed, $\hoaremonad(X)$ is isomorphic to the dcpo of Scott-closed sets.

\subsection{Omitted Proof in~\cref{subsec:probAcceptingTrace}}
\label{subsec:omitproofEx1}
We prove that the tuple $(\comProbPowerdommonad, (A^{\ast}, =), (Y, =), \tau^{\_ \times A^*}, \tau^\comProbPowerdommonad, \alpha^d, q)$ defined in~\cref{subsec:probAcceptingTrace} is a product situation.

We show that the expectation $\tau^{\comProbPowerdommonad}$ is an Eilenberg-Moore algebra. This is indeed immediate by~\cref{lem:KegelspitzeEilenberg}: the dcpo $[0, 1]$ is a continuous Kegelspitze, and the expectation is the linear barycenter map of $[0, 1]$.
To see that the expectation is the linear barycenter map of $[0, 1]$, it is enough to show that $\int p  d\nu = \sum_{w_i} w_i\cdot x_i$ for any simple valuation $\nu \defeq \sum^{m}_{i=1} w_i \cdot \eta_{X}(x_i)$ on $[0, 1]$.
This is because the set of simple valuations is dense in $\comProbPowerdommonad([0, 1])$ in the d-topology, the expectation and the linear barycenter map are Scott-continuous, and $[0, 1]$ is a Hausdorff space with the d-topology.

We then show that $\alpha^d$ is a strong monad morphism, which is defined in~\cref{def:monadMorDFA}.
By \cref{ap:st_monad_morphism_bij}, it suffices to check that 
$\widehat{d}^\dagger$ satisfies conditions (i) and (ii) written in \cref{ap:st_monad_morphism_bij}.
These conditions can be easily checked because
$\widehat{d}(y)(\epsilon) = y$
and $\widehat{d}(y)(w_2\cdot w_1) = \widehat{d}\big(\widehat{d}(y)(w_2) \big)(w_1)$ hold for each $y \in Y$ and $w_1, w_2 \in A^\ast$.

Finally, we prove that $q\colon A^{\ast}\rightarrow [0, 1]^{Y}$ is indeed an inference query.
We can see that:
\begin{align*}
  & (q\circ \tau^{\_\times A^{\ast}})(w_1, w_2)(y) =  \begin{cases}
    1 &\text{ if $w_2\cdot w_1$ is accepting in $d$ from $y$,}\\
    0 &\text{ otherwise.}
  \end{cases}\\
  & (\tau^{(\_\times Y)^{Y}}\circ (q\times Y)^{Y}\circ \alpha^d_{A^{\ast}})(w_1, w_2)(y) = \begin{cases}
    1 &\text{ if $w_1$ is accepting in $d$ from $\widehat{d}(y)(w_2)$,}\\
    0 &\text{ otherwise.}
  \end{cases}
\end{align*}
It is then clear that $w_2\cdot w_1$ is accepting in $d$ from $y$ iff $w_1$ is accepting in $d$ from $\widehat{d}(y)(w_2)$ by the construction of $\widehat{d}$. \qed

\subsection{Omitted Definitions for~\cref{ex:hoRW}}
\label{subsec:omitSpec}
We present the specification for~\cref{ex:hoRW}.
\begin{equation*}
  \label{eq:exHOProbSpec}
\begin{tikzpicture}
      \node[state, initial] (y1) at (0, 0) {\tiny $y_1$};
      \node[state] (y2) at (2, 0) {\tiny $y_2$};
      \node[state, accepting] (y3) at (4, 0) {\tiny $y_3$};
      \draw[->] (y1) to node[pos=0.5, inner sep=3pt, above] {$\up$} (y2);
      \draw[->] (y1) edge [loop above] node[pos=0.5, inner sep=3pt, above] {$\down$} (y1);
      \draw[->] (y2) to node[pos=0.5, inner sep=3pt, above] {$\up$} (y3);
      \draw[->] (y2) edge [loop above] node[pos=0.5, inner sep=3pt, above] {$\down$} (y2);
      \draw[->] (y3) edge [loop above] node[pos=0.5, inner sep=3pt, above] {$\up, \down$} (y3);
\end{tikzpicture}
\end{equation*}

\subsection{Omitted Proofs in~\cref{subsec:expectedRewardAcceptingTrace}}
\label{sec:omitProofPartialExp}

\begin{lemma}
  \label{lem:cartesianprodhaus}
  The Cartesian product $[0, 1]\times [0, \infty]$ is a Hausdorff space with the d-topology.
\end{lemma}
\begin{proof}
We show that the Cartesian product of open intervals are open in the $d$-topology, which are sufficnet to separate two different points.
This is true because open intervals in $[0, 1]$ and $[0, \infty]$ are open in the d-topology, and the d-topology of $[0, 1]\times [0, \infty]$ is finer than the product topology of d-topologies~\cite{KeimelL09}.
\end{proof}

We first show that $\tau^{\comProbPowerdommonad(\_\times \nonnegrat)}\colon \comProbPowerdommonad([0, 1]\times [0, \infty]\times \nonnegrat)\rightarrow [0, 1]\times [0, \infty]$ is an Eilenberg-Moore algebra.
By~\cite{Beck69,AguirreKK22}, it suffices to show that there are two Eilenberg-Moore algebras such that $\tau^{\comProbPowerdommonad}\colon \comProbPowerdommonad([0, 1]\times [0, \infty])\rightarrow [0, 1]\times [0, \infty]$ and
$\tau^{\_\times \nonnegrat}\colon [0, 1]\times [0, \infty]\times \nonnegrat\rightarrow [0, 1]\times [0, \infty]$ defined by
\begin{align*}
  \tau^{\comProbPowerdommonad}(\nu) \defeq \Big(\int_{p, r}  p d\nu, \int_{p, r}  r d\nu\Big), \\
  \tau^{\_\times \nonnegrat}(p, r, n) \defeq (p, n\cdot p + r),\text{ and }
\end{align*}
the following equation holds:
\begin{align*}
  \tau^{\comProbPowerdommonad}\circ \comProbPowerdommonad(\tau^{\_\times \nonnegrat}) \circ \rstrength{\comProbPowerdommonad}{} = \tau^{\_\times \nonnegrat} \circ \tau^{\comProbPowerdommonad}\times \nonnegrat.
\end{align*}

We show that $\tau^{\comProbPowerdommonad}$ is an Eilenberg-Moore algebra.
Since $[0, 1]\times [0, \infty]$ is a continuous Kegelspitze,
 it suffices to show that $\tau^{\comProbPowerdommonad}$ is the linear barycenter map.
This can be shown by the similar argument in~\cref{subsec:omitproofEx1} by~\cref{lem:cartesianprodhaus}.

We then show that $\tau^{\_\times \nonnegrat}$ is an Eilenberg-Moore algebra.
For the condition on units, we have
\begin{align*}
  \tau^{\_\times \nonnegrat}\circ \eta^{\_\times \nonnegrat}_{[0, 1]\times [0, \infty]}(p, r) = (p, r), \text{ and }
\end{align*}
for the condition on multiplications, we have
\begin{align*}
  &\tau^{\_\times \nonnegrat}\circ \mu^{\_\times\nonnegrat}_{[0, 1]\times [0, \infty]}(p, r, n_1, n_2) = \big(p, (n_1+n_2)\cdot p + r\big),\\
  &\tau^{\_\times \nonnegrat}\circ (\tau^{\_\times \nonnegrat}\times \nonnegrat)(p, r, n_1, n_2) = \tau^{\_\times \nonnegrat}(p, n_1\cdot p + r, n_2) = \big(p, (n_1+n_2)\cdot p + r\big).
\end{align*}

Finally, we show the last equation.
For any $n\in \nonnegrat$ and simple valuation $\nu \defeq \sum^{m}_{i=1} w_i \cdot \eta_{[0, 1]\times [0, \infty]}(x_i, y_i)$, we have
\begin{align*}
  &\big(\tau^{\_\times \nonnegrat} \circ \tau^{\comProbPowerdommonad}\times \nonnegrat\big)(\nu, n) = \Big( \int_{p, r}  p d\nu,\, n\cdot \int_{p, r}  p d\nu + \int_{p, r}  r d\nu\Big)\\
  =& \Big( \sum^m_{i=1} w_i\cdot x_i, \sum^m_{i=1} w_i\cdot (n\cdot x_i + y_i)\Big),\\
  &\big(\tau^{\comProbPowerdommonad}\circ \comProbPowerdommonad(\tau^{\_\times \nonnegrat}) \circ \rstrength{\comProbPowerdommonad}{}\big)(\nu, n)\\
  =&\big( \tau^{\comProbPowerdommonad}\circ \comProbPowerdommonad(\tau^{\_\times \nonnegrat})\big)\Big(\lambda U.\int_{x, y} \chi_U(x, y, n) d\nu\Big)\\
  =& \big( \tau^{\comProbPowerdommonad}\circ \comProbPowerdommonad(\tau^{\_\times \nonnegrat})\big)\Big(\sum^{m}_{i=1} w_i \cdot \eta_{[0, 1]\times [0, \infty]\times \nonnegrat}(x_i, y_i, n)\Big)\\
  =& \big( \tau^{\comProbPowerdommonad}\big)\Big(\sum^{m}_{i=1} w_i \cdot \eta_{[0, 1]\times [0, \infty]}(x_i, n \cdot x_i + y_i)\Big)\\
  =& \Big( \sum^m_{i=1} w_i\cdot x_i, \sum^m_{i=1} w_i\cdot (n\cdot x_i + y_i)\Big).
\end{align*}
This implies the desired equality because the set of simple valuations are dense in the d-topology.

We then show that the $q\colon A^{\ast}\rightarrow ([0, 1]\times [0, \infty])^{Y}$ is an inference query. 
We can see that: 
\begin{align*}
  & (q\circ \tau^{\_\times A^{\ast}})(w_1, w_2)(y) =   \begin{cases}
    (1, 0) &\text{ if $w_2\cdot w_1$ is accepting in $d$ from $y$,}\\
    (0, 0) &\text{ otherwise.}
  \end{cases}\\
  & (\tau^{(\_\times Y)^{Y}}\circ (q\times Y)^{Y}\circ \alpha^d_{A^{\ast}})(w_1, w_2)(y)\\
=& \begin{cases}
    (1, 0) &\text{ if $w_1$ is accepting in $d$ from $\widehat{d}(y)(w_2)$,}\\
    (0, 0) &\text{ otherwise.}
  \end{cases}
\end{align*}
We know that $w_2\cdot w_1$ is accepting in $d$ from $y$ iff $w_1$ is accepting in $d$ from $\widehat{d}(y)(w_2)$ by the construction of $\widehat{d}$.
\qed

\subsection{Omitted Proofs in~\cref{subsec:emptinessCheckingAcceptingTraces}}
\label{subsection:omittedProofsEmptinessCheck}
We show that $q\colon A^{\ast}\rightarrow \boolsets^{Y}$ is an inference query. 
\begin{align*}
  & (q\circ \tau^{\_\times A^{\ast}})(w_1, w_2)(y) = \begin{cases}
    \top &\text{ if $w_2\cdot w_1$ is accepting in $d$ from $y$,}\\
    \bot &\text{ otherwise,}
  \end{cases}\\
  & (\tau^{(\_\times Y)^{Y}}\circ (q\times Y)^{Y}\circ \alpha^d_{A^{\ast}})(w_1, w_2)(y)\\
=& \begin{cases}
    \top &\text{ if $w_1$ is accepting in $d$ from $\widehat{d}(y)(w_2)$,}\\
    \bot &\text{ otherwise.}
  \end{cases}
\end{align*}
Again, we know that $w_2\cdot w_1$ is accepting in $d$ from $y$ iff $w_1$ is accepting in $d$ from $\widehat{d}(y)(w_2)$ by the construction of $\widehat{d}$. \qed

\subsection{Omitted Proofs in~\cref{subsec:maximumRewardsAcceptingTraces}}
\label{subsection:omittedProofsRewardMachine}
We show that $\alpha^d\colon \_\times A^{\ast}\Rightarrow (\_\times U\times \boolsets\times \nonnegreal)^{U\times \boolsets\times \nonnegreal}$ is a strong monad morphism. 
By \cref{ap:st_monad_morphism_bij}, it suffices to check that 
$\widehat{d}^\dagger$ satisfies conditions (i) and (ii) written in \cref{ap:st_monad_morphism_bij}.
These conditions can be easily checked because
$\widehat{d}(u, b, m)(\epsilon) = (u, b, m)$
and $\widehat{d}(u, b, m)(w_2\cdot w_1) = \widehat{d}\big(\widehat{d}(u, b, m)(w_2) \big)(w_1)$ hold for each $u, b, m \in U \times \boolsets \times \nonnegreal$ and $w_1, w_2 \in A^\ast$.

We then prove that $\tau^{\hoaremonad}\colon \hoaremonad([0, \infty])\rightarrow[0, \infty]$ is an Eilenberg-Moore algebra. 
Clearly, $\tau^{\hoaremonad}$ is Scott-continuous. It is easy to show the condition on units. We prove the condition on multiplication. 
We see this by the following calculation: 
\begin{align*}
  &\big(\tau^{\hoaremonad}\circ \mu^{\hoaremonad}_{[0, \infty]}\big)(S) = \sup \big\{r\in T \ \big|\ T\in S \big\} =   \sup \big\{\sup T \ \big|\ T\in S \big\}  = \big(\tau^{\hoaremonad}\circ \hoaremonad(\tau^{\hoaremonad})\big)(S).
\end{align*}

Finally, we show that $q\colon A^{\ast}\rightarrow {[0, \infty]}^{U\times \boolsets\times  \nonnegreal}$ is an inference query. 
We have 
\begin{align*}
  & (q\circ \tau^{\_\times A^{\ast}})(w_1, w_2)(u, b, m) = \begin{cases}
    m' &\text{ if $w_2\cdot w_1$ gains a reward $m'$ in $d$ from $(u, b, m)$,}\\
    0 &\text{ otherwise,}
  \end{cases}\\
  & (\tau^{(\_\times U\times \boolsets\times \nonnegreal)^{U\times \boolsets\times \nonnegreal}}\circ (q\times U\times \boolsets\times \nonnegreal)^{U\times \boolsets\times \nonnegreal}\circ \alpha^d_{A^{\ast}})(w_1, w_2)(u, b, m)\\
=& \begin{cases}
    m' &\text{ if $w_1$ gains a reward $m'$ in $d$ from $\widehat{d}(u, b, m)(w_2)$,}\\
    0 &\text{ otherwise.}
  \end{cases}
\end{align*}
They coincide because $w_2\cdot w_1$ gains a reward $m'$ in $d$ from $(u, b, m)$ iff $w_1$ gains a reward $m'$ in $d$ from $\widehat{d}(u, b, m)(w_2)$. \qed

\section{Omitted Benchmarks}
\label{sec:omittedBenchmarks}

In addition to the examples in~\cref{subsec:ovSafety} and~\cref{ex:hoRW}, we use the following benchmark.

\begin{example}
  \label{ex:horeward}
  Consider the following higher-order probabilistic program with rewards:
    \begin{equation*}
        \letrec{\gainreward}{f}{\exprobbranch{\big(\exprobbranch{{(\gainreward\ f)}^{\fail}}{1/2}{1}{{(\gainreward\ g)}^{\success}}\big)}{3/4}{0}{\big(f\ ()\big)}}{\gainreward\ (\lambda x. x)},
    \end{equation*}
    where $f\colon\unit\rightarrow \unit$ and  $g\defeq \lambda x.\ \exprobbranch{(f\ x)}{1/2}{0}{()}$.
    We use the same specification presented in~\cref{ex:rewrdFO}. 
\end{example}

\end{document}
\endinput